\documentclass[12pt]{article}
\usepackage{amsmath}
\usepackage{graphicx}
\usepackage{enumerate}
\usepackage{natbib}
\usepackage{url} % not crucial - just used below for the URL 
\usepackage{color}
\usepackage{natbib,amsmath,amssymb,amsthm,graphicx,setspace,paralist,booktabs,rotating,times,bm}
\usepackage{subfigure}
\usepackage{mathrsfs}
\usepackage{verbatim}
\usepackage{bbm}
\usepackage{url}
\usepackage[colorlinks=true,linkcolor=blue]{hyperref}
\usepackage[ruled,vlined]{algorithm2e}
\usepackage{algpseudocode}
\usepackage{mathrsfs}
\usepackage{multirow}
\usepackage{comment}
\newcommand{\blind}{1}

\newtheorem{theorem}{Theorem}
\newtheorem{lemma}{Lemma}
\newtheorem{corollary}{Corollary}
\theoremstyle{definition}

\newtheorem{example}{Example}

\algnewcommand{\algorithmicand}{\textbf{ and }}
\algnewcommand{\algorithmicor}{\textbf{ or }}
\algnewcommand{\OR}{\algorithmicor}
\algnewcommand{\AND}{\algorithmicand}
\def\Ca{{C1}}
\def\Cb{{C2}}
\def\WCc{{WC3}}
\def\WCd{{WC4}}
\def\Cc{{C3}}
\def\Cd{{C4}}

\def\Cf{{C5}}
\def\Cg{{C6}}
\def\Ch{{C7}}
\def\Ca{{C1}}
\def\Cb{{C2}}
\def\WCc{{WC3}}
\def\WCd{{WC4}}
\def\Cc{{C3}}
\def\Cd{{C4}}

\def\Cf{{C5}}
\def\Cg{{C6}}
\def\Ch{{C7}}
\def\thmone{{S1}}
\def\thmtwo{{4}}

\def\thmmain{{1}}
\def\thmoneRes{{3}}
\def\thmzeroRes{{2}}
\def\lemlower{{4}}
\def\lemtailEudis{{1}}
\def\lemalpha{{2}}
\def\lemtailtwo{{3}}
\def\corHone{{1}}

\def\Pr{{\mathbb{P}}}
\def\Ex{{\mathbb{E}}}
\def\var{\hbox{Var}}
\def\cov{\hbox{cov}}

\def\cov{\mbox{cov}}

\def\argmax{\mathop{\arg\max}}
\def\diag{\mbox{diag}}

\def\trans{^{\rm T}}

\def\bTheta{\boldsymbol\Theta}
\def\btheta{\boldsymbol\theta}
\def\ba{{\boldsymbol\alpha}}

\def\bxi{\boldsymbol\xi}

\def\0{{\bf 0}}
\def\1{{\bf 1}}
\def\A{{\bf A}}

\def\V{{\bf V}}

\def\x{{\bf x}}

\def\B{{\bf B}}

\def\D{{\bf D}}
\def\V{{\bf V}}

\def\b{{\bf b}}
\def\I{{\bf I}}

\def\pr{ \mathbb{P}}

\def\V{{\bf V}}
\def\t{{\bf t}}

\def\bP{{\bf P}}
\def\bQ{{\bf Q}}
\def\bV{{\bf V}}

\def\m{{\bf m}}
\def\v{{\bf v}}

\def\x{{\bf x}}
\def\I{{\bf I}}

\def\0{{\bf 0}}

\def\diag{\hbox{diag}}

\def\bq{\begin{equation}}
\def\eq{\end{equation}}

\def\trans{^{\rm T}}

\def\diag{\hbox{diag}}
\def\log{\hbox{log}}

\def\squarebox#1{\hbox to #1{\hfill\vbox to #1{\vfill}}}
\def\btheta{{\boldsymbol \theta}}
\def\balpha{{\boldsymbol \alpha}}

\def\bx{{\bf x}}

\def\bse{\begin{eqnarray*}}
	\def\ese{\end{eqnarray*}}
\def\be{\begin{eqnarray}}
\def\ee{\end{eqnarray}}
\def\bsq{\begin{equation*}}
\def\esq{\end{equation*}}
\def\bq{\begin{equation}}
\def\eq{\end{equation}}

\def\diag{\hbox{diag}}

\def\diag{\hbox{diag}}
\def\trans{^{\rm T}}
\def\btheta{\bm{\theta}}
\def\rd{{\rm d}}

\def\colblue{{}}

\def\trans{^{\rm T}}

\begin{document}

\def\spacingset#1{\renewcommand{\baselinestretch}%
{#1}\small\normalsize} \spacingset{1}

%%%%%%%%%%%%%%%%%%%%%%%%%%%%%%%%%%%%%%%%%%%%%%%%%%%%%%%%%%%%%%%%%%%%%%%%%%%%%%

\if1\blind
{
  \title{\bf Feature screening for clustering analysis}
  \author{
    Changhu Wang\hspace{.2cm}\\
    School of Mathematical Sciences, Peking University\\
    and \\
    Zihao Chen\hspace{.2cm}\\
    School of Mathematical Sciences, Peking University\\
    and \\
    Ruibin Xi \thanks{
    	The authors gratefully acknowledge \textit{the National Key Basic Research Project of China (2020YFE0204000), the National Natural Science Foundation of China (11971039), and Sino-Russian Mathematics Center.}}\\
    School of Mathematical Sciences, Peking University\\
    Center for Statistial Sciences, Peking University}
  \maketitle
} \fi

\if0\blind
{
  \bigskip
  \bigskip
  \bigskip
  \begin{center}
    {\LARGE\bf  Feature screening for clustering analysis}
\end{center}
  \medskip
} \fi

\bigskip
\begin{abstract}
In this paper, we consider feature screening for ultrahigh dimensional clustering analyses. Based on the observation that the marginal distribution of any given feature is a mixture of its conditional distributions in different clusters, we propose to screen clustering features by independently evaluating the homogeneity of each feature’s mixture distribution. Important cluster-relevant features have heterogeneous components in their mixture distributions and unimportant features have homogeneous components. The well-known EM-test statistic is used to evaluate the homogeneity. Under general parametric settings, we establish the tail probability bounds of the EM-test statistic for the homogeneous and heterogeneous features, and further show that the proposed screening procedure can achieve the sure independent screening and even the consistency in selection properties. Limiting distribution of the EM-test statistic is also obtained for general parametric distributions. The proposed method is computationally efficient, can accurately screen for important cluster-relevant features and help to significantly improve clustering, as demonstrated in our extensive simulation and real data analyses. 
\end{abstract}

\noindent%

\vfill

\newpage
\spacingset{1.9} % DON'T change the spacing!
\section{Introduction}
High dimensional data is prevalent in a wide range of research fields and applications, such as biological studies, financial studies and image data analyses. In high dimensional data, the number of features $p$ is very large and can be much larger than the number of samples $n$ ($ p \gg n$). One of the most important tasks of high dimensional data analyses is to cluster the samples and uncover unknown groups and structures in the data. In real applications, cluster-relevant features are often only a small proportion of the $p$ features, and other features are cluster-irrelevant. Incorporation of the irrelevant features in clustering analyses can blur the differences between clusters, significantly influence the clustering accuracy, and make clustering computationally more demanding, especially when $p$ is large. If one can accurately distinguish  cluster-relevant features from  cluster-irrelevant features, clustering analyses could be significantly improved in terms of both clustering accuracy and computational efficiency (See Figure \ref{fig:example} for an example). 

\begin{figure}[htb]
	\centering
	\includegraphics[width=0.8\linewidth]{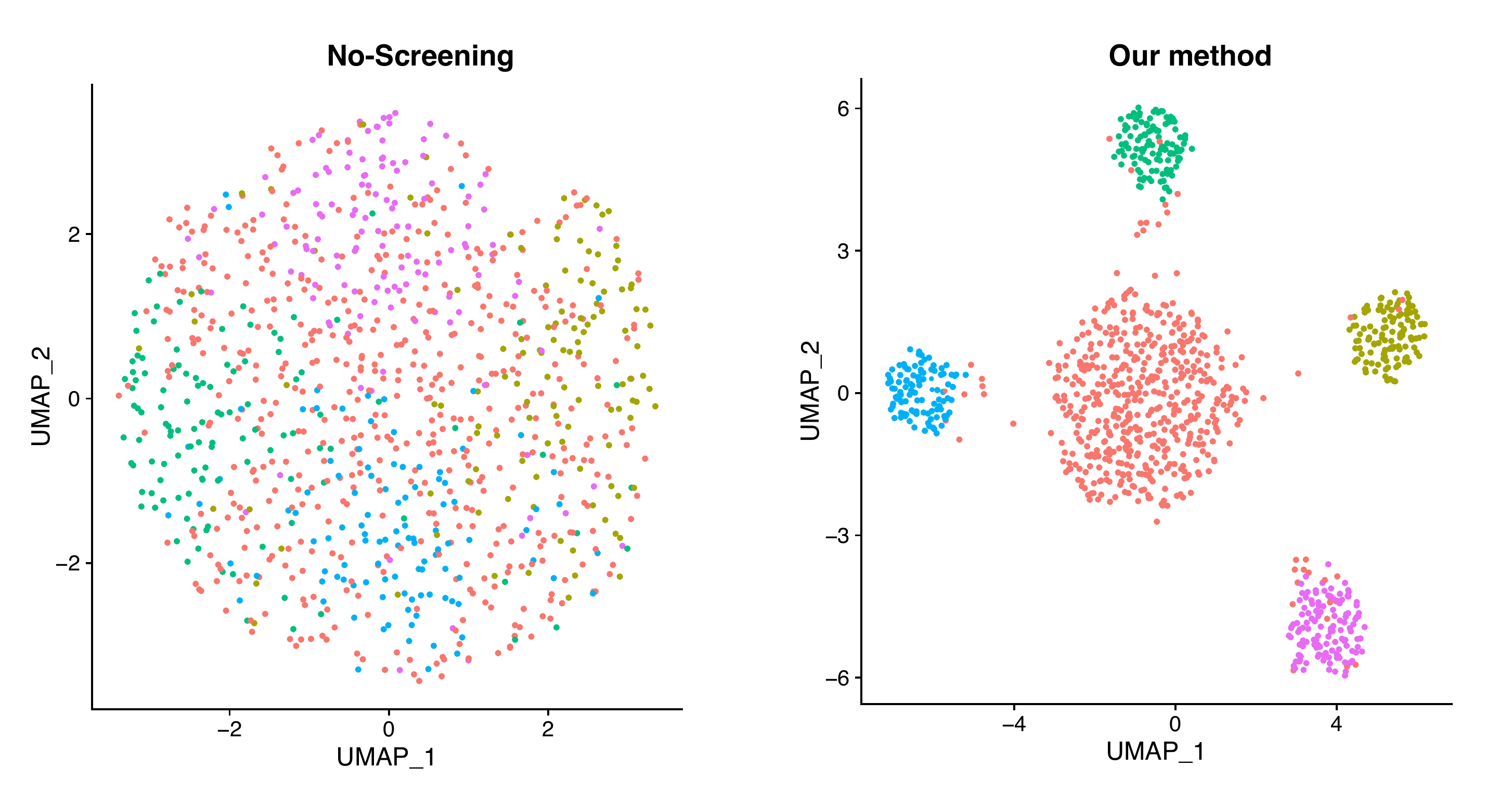}
	\caption{An example of simulated data generated  as in Section \ref{sec:simu}.   The data has 5 clusters and 5000 features. The first 30 features are cluster-relevant and the other 4970 features are cluster-irrelevant. We perform dimension reduction using the uniform manifold approximation and projection (UMAP). The left plot is the UMAP plot using all features and the right plot is the UMAP plot using the features selected by the proposed method. The points in the plots are colored by their true clustering labels. The large amount of the cluster-irrelevant features make the clusters difficult to be distinguished. After feature screening, different clusters are much easier to be distinguished. }
	\label{fig:example}
\end{figure}

We consider the feature screening problem in clustering analyses of high dimensional data. Suppose that  $\bx_{i} = (x_{i1},\dots,x_{ip})\trans \in \mathbb{R}^p$ ($i=1,\dots,n$) are $n$ independent observations. The $n$ samples are from $G$ clusters and their cluster labels are unknown. We assume that only $s$ of the $p$ features (often, $s \ll p$) contain cluster label information and all other features are independent of the clusters. We aim to develop a computationally efficient statistical method that can effectively screen out the cluster-irrelevant features, while retaining all or almost all cluster-relevant features. Traditional clustering algorithms such as the k-means algorithm can then be applied to the retained features and sample clusters can be obtained. We are most interested in high dimensional count data, although the method developed here can also be applied to continuous data.   

One motivation of this work is the single-cell RNA sequencing (scRNA-seq) analysis \citep{kiselev2019challenges}. In recent scRNA-seq studies, gene expressions in single-cells are profiled for over $10,000$ genes and the raw expression values are rather small count data ($<$ 20 for majority of genes). The unknown cell types of single-cells are often assigned based on clustering analyses of gene expressions. However, only marker genes differentially expressed among different cell types are useful for cell type identification. To address the high-dimensional clustering problems, scRNA-seq studies often only use the so-called highly variable genes for clustering analysis, based on the assumption that genes with larger expression variances are more likely to be marker genes. Though this strategy has been widely adopted, selecting highly variable genes can include many non-marker genes and exclude many marker genes, thus leading to the inaccurate clustering \citep{Andrews2019M3DropDF}.

	The supervised screening problem has been extensively studied and many methods have been developed, such as \citet{fan2008sure},  \citet{zhu2011model} and \citet{li2012feature} among many others \citep{liu2015selective}. These methods are particularly suitable for ultra-high dimensional supervised learning problems, which bring many statistical and computational challenges to the traditional variable selection methods. With the response variable, available supervised screening methods can measure each predictor's association with the response variable independently. The predictors are then ranked by their association strengths and top predictors are retained. With a proper threshold, these screening methods can correctly select all important features with a high probability or even can correctly distinguish the important and unimportant features with a high probability, which are known as the sure independent screening property \citep{fan2008sure} and the consistency in selection property \citep{li2012feature}, respectively.

The unsupervised feature screening is more challenging because there is no response variable. Variable selection methods for clustering analyses have been developed \citep{witten2010framework, fop2018variable, liu2022clustering}.  However, similar to the supervised variable selection methods, when the dimensionality is ultrahigh, their performance is challenged in terms of both statistical accuracy and computational efficiency. To address the ultra-high dimensional problems, the pioneer work by \citet{chan2010using} developed a feature screening method by testing the unimodality of each feature's distribution. Features with unimodal distributions are cluster-irrelevant and should be screened out. More recently, \citet{jin2016influential} developed an innovative method called IF-PCA for ultra-high dimensional clustering analysis. IF-PCA first screens cluster-relevant using the Kolmogorov-Smirnov (KS) test and then applies the k-means algorithm to cluster. 
\cite{liu2022clustering}  proposed  a non-parametric  feature screening  method called SC-FS.  SC-FS performs feature screening by correlating each feature with a pre-clustering label. 
However, the few available screening methods are either developed for continuous data or require pre-clustering of the data. The continuous methods are not suitable for count data, and in ultrahigh dimensional settings, the pre-clustering can be very inaccurate and the methods relying on the pre-clustering results will also be inaccurate.

In this paper, we develop a general parametric feature screening method that can be applied to both continuous and count data. 
Marginally, all features can be viewed as following mixture distributions. However, the mixture components of a cluster-relevant feature are not all the same (heterogeneous distribution), and those of a cluster-irrelevant feature will be the same (homogeneous distribution). Therefore, we can test whether a feature is cluster-relevant without the cluster labels. Observe that multi-modal distributions are mixture distributions of unimodal distributions. Thus, our method essentially uses the same characteristic as \citet{chan2010using} for feature screening of clustering analyses.

We propose to use the EM-test, a well-known homogeneity test of mixture models, for feature screening. The EM-test was originally developed to overcome the critical problems of likelihood ratio tests for homogeneity \citep{hartigan1985failure,chernoff1995asymptotic}. Limiting distributions under the homogeneity were available for mixture models of one-parameter distributions or of two components \citep{li2009non,niu2011testing,li2010testing}. In this paper, in addition to the limiting distribution, we establish theoretical properties of the EM-test for feature screening of clustering analyses under general settings of mixture models. The mixture models are allowed to be mixtures of multi-parameter distributions and/or of multiple-components. The major theoretical results include the following.
\begin{itemize}
	\item Under the homogeneous model, the EM-test statistic is bounded with a high probability, or more specifically, the probability of the EM-test statistic greater than any $t>0$ decays to zero at a polynomial rate with respect to $t$. 
	\item Under the heterogeneous model, the EM-test statistic diverges to infinity with a probability approaching to one at an exponential rate.
	\item When the dimensionality $p$ goes to infinity exponentially with the sample size $n$ (more precisely, with $\exp\left(n^\beta\right)$ for $0<\beta<1/2$), the screening procedure based on the EM-test achieves the sure independent screening property \citep{fan2008sure}. If $p$ goes to infinity at any polynomial order of $n$, we can even achieve the consistency in selection \citep{li2012feature}. 
\end{itemize}
We perform extensive simulation studies and find that the EM-test can accurately screen for cluster-relevant features. After feature screening, clustering accuracy can also be significantly improved. In an application of scRNA-seq data, we find that the EM-test renders more accurate single-cell clustering and enables the detection of a rare cell-type that is difficult to be detected by other methods.

The rest of the paper is organized as follows. Section 2 introduces the model setup, the EM-test and defines basic notations.  Section \ref{sec:theory} gives the bounds of the tail probabilities under the homogeneous and heterogeneous models, and further establishes the sure independent screening and model selection consistency property.  The limiting distribution of the EM-test statistic is also presented in Section 3. Simulation and real data analyses are presented in Section \ref{sec:simu} and Section \ref{sec:app}, respectively.   Finally, in Section \ref{sec:diss}, we discuss the limitations of this research and future research directions of high dimensional clustering feature screening.  
The  proofs of  the  results are presented in the Supplementary material.

\section{Model setup and the EM-test} 
In this section, we present the statistical model setup for the feature screening of clustering analyses and introduce the screening procedure based on the EM-test statistic. 
\subsection{Model setup for feature screening of clustering analyses}  Suppose that we have $n$ independent observations $\bx_{i} = (x_{i1},\dots,x_{ip})\trans \in \mathbb{R}^p$ ($i=1,\dots,n$) from $G$ clusters and $\ba=(\alpha_1,\dots,\alpha_G)$ be the proportions of different clusters $ (\sum_{g=1}^{G} \alpha_{g} = 1, \alpha_{g} > 0, g=1, \dots, G)$. We denote the unknown cluster labels as $g_i \in \{1,\dots,G\}$ ($i=1,\dots,n$). Assume that given the cluster label $g$, the conditional distribution $F_j(x|g)$ of $x_{ij}$ is from a known identifiable parametric distribution family
$\mathcal{P} = \{f(x;\btheta): \btheta \in  \Theta \subset \mathbb{R}^d \} $, where $f(x;\btheta)$ is the density function with respect to a $\sigma$-finite measure $\mu$, on $\mathbb{R}$ parameterized by $\btheta$, and $\Theta \subset \mathbb{R}^d$ is a convex compact parameter space.   Note that for count data, the measure $\mu$ can be taken as the counting measure of the nonnegative integers; for continuous data,  $\mu$ is the Lebesgue measure on $\mathbb{R}$. Thus, our method and theory apply to both count and continuous data.  Define $\Xi  =\Theta^G$  as the product space of $\Theta$. 

In high dimensional clustering problems, only a small portion of the $p$ features contain information about the cluster labels and the majority of them are irrelevant to the sample clusters. Our goal is to screen out the cluster-irrelevant features to facilitate downstream clustering analysis. Intuitively, if the $j$th random variable $x_{j} \in \mathbb{R}$ is unrelated with the cluster label $g$, the conditional distribution $F_j(x|g)$ of $x_{j}$ given the cluster label $g$ should be independent of the cluster label $g$, or $F_j(x|g = 1)=\cdots=F_j(x|g=G)$. If, on the other hand, the $j$th random variable $x_{j} \in R$ is a cluster-relevant feature, there are at least two $g\neq g^\prime$ such that $F_j(x| g) \neq F_j(x| g^\prime)$. 

Let $f(x;\btheta_{jg})$ be the density function of the conditional distribution $F_j(x|g)$. The labels are unknown and the random variable $x_j$ should follow a mixture distribution $\varphi(x;\bxi_j, \balpha)  = \sum_{g=1}^{G}\alpha_{g} f(x;\btheta_{jg})$, where $\bxi_j = (\btheta_{j1}\trans, \dots, \btheta_{jG}\trans)\trans \in \Xi =\Theta^G$.  Define the interior of the $G-1$ dimensional probability simplex as 
$\mathbb{S}^{G-1} = \{\ba \in \mathbb{R}^G : \sum_{g=1}^G \alpha_g = 1, \alpha_g > 0, \mbox{ for }  g=   1,\dots, G\}$ and the $G$-mixture distribution family  as
$$\mathcal{P}^G = \left\{\sum_{g=1}^{G}\alpha_gf(x;\btheta_g): \ba \in \mathbb{S}^{G-1}, \btheta_g \in \Theta, g=1,\dots,G \right\}.$$
We have $\varphi(x;\bxi_j, \balpha) \in \mathcal{P}^G$.  In this paper, we assume that $\mathcal{P}^G$ is an identifiable finite mixture, in other words,  $ \mathcal{P} $ is a linearly independent set over the field of real numbers \citep{Yakowitz1968OnTI}.  For a cluster-irrelevant feature $j$, $\btheta_{j1}=\cdots=\btheta_{jG}$ and thus $\varphi(x;\bxi_j, \balpha) \in \mathcal{P}$. For a cluster-relevant feature $j$, there are at least two $g\neq g^\prime$ such that $\btheta_{jg} \neq \btheta_{jg^{\prime}}$ and $\varphi(x;\bxi_j, \balpha) \in \mathcal{P}^G\backslash \mathcal{P}$. Therefore, we can consider the following hypothesis testing problems to screen for the cluster-relevant features.
\begin{equation} \label{homotest}
\mathbb{H}_{j0}: \varphi(x;\bxi_j, \balpha) \in \mathcal{P} \mbox{ v.s. } \mathbb{H}_{j1}: \varphi(x;\bxi_j, \balpha) \in \mathcal{P}^G\backslash \mathcal{P}. 
\end{equation}
We call the models under the null hypotheses $\mathbb{H}_{j0}$ homogeneous models, and those under the alternative hypotheses $\mathbb{H}_{j1}$ heterogeneous models. In real applications, the number of clusters $G$ is often unknown. However, we often can have a rough estimate of $G$ and can choose $G$ to be larger than the true number of clusters. In such cases, the null and alternative hypotheses still hold for the cluster-irrelevant and relevant features, respectively.    Simulation shows that the choice of $G$ has little influence on the performance of EM-test, especially when $G$ is chosen to be larger than the true clusters (Supplementary Section D).

\subsection{The EM-test statistic and the screening procedure} 
We use the EM-test statistic for feature screening of clustering analyses. Theoretical results of the EM-test statistic are developed for the hypothesis testing problem (\ref{homotest}) under general settings with multiple parameters ($d\geq 1$), multiple components ($G \geq 2$) and both continuous and count data. Let $\x = (x_1,\dots,x_n)$ be a random sample of size $n$ from a $G$-mixture model 
\begin{equation} \label{eq:Gmix}
\varphi(x;\bxi,\balpha) = \sum_{g=1}^{G}\alpha_gf(x;\btheta_g),
\end{equation}
where $\btheta_g \in \Theta$, ($g= 1,\ldots,G$), $\bxi = (\btheta_1,\ldots, \btheta_G)$ and $\balpha=(\alpha_1,\dots,\alpha_G)$.
Let
$l_n(\bxi, \balpha) = \sum_{i=1}^{n} \log \ \varphi(x_i;\bxi, \balpha)$
be the log-likelihood function, and define the penalized log-likelihood function as  
\begin{equation} \label{penalized-loglikelihood}
pl_n(\bxi, \balpha) = \sum_{i=1}^{n} \log \ \varphi(x_i;\bxi, \balpha)+ p(\balpha),
\end{equation} 
where $p(\balpha)=\lambda\left(\sum_{g=1}^G{\rm log}\alpha_g  + G\log(G)\right)$ is a penalty function, where $\lambda>0$ is a penalty parameter and is always set as 0.00001 in the simulation and real data analyses of this paper.   Simulation shows that EM-test is robust to the choice of the penalty parameter $\lambda$ and $\lambda=0.00001$ gives very similar results to other choices of $\lambda$ (Supplementary Table S3).  Largely speaking, the EM-test statistic is defined as the difference between the maximum penalized log-likelihoods of the heterogeneous and homogeneous models. The maximum penalized log-likelihood under the heterogeneous model is obtained using the EM algorithm. More specifically, we use the following procedure to calculate the EM-test statistic.

Suppose that $\hat{\bxi}_0 = \left(\hat{\btheta}_{0},\ldots,\hat{\btheta}_{0}\right)$ is the estimator that maximizes the penalized log-likelihood function (\ref{penalized-loglikelihood}) under the homogeneous model. Under the heterogeneous model, given any initial value $\ba^{(0)} \in \mathbb{S}^{G-1} $ 
%with 
%\begin{equation}\label{eq:ba0}
%\sum_{g=1}^G\alpha_{g}^{(0)} = 1  \mbox{ and } \alpha_{g}^{(0)} >0 \  (g = 1,\dots,G), 
%\end{equation}
we first compute 
\begin{equation}\label{eq:xi0}
\bxi^{(0)} =  \argmax_{\substack{\bxi \in \Xi }}  \sum_{i=1}^n \log \ \varphi\left(x_i;\bxi, \balpha^{(0)}\right) + p\left(\balpha^{(0)}\right).  
\end{equation}
Assume that $\balpha^{(k)}$ and $\bxi^{(k)} $ are the  estimators at the $k$-th iteration of the EM algorithm. The E-step updates the posterior probability of the $i$-th sample coming from the $g$-th component by
\begin{equation}\label{eq:defw}
w_{gi}^{(k)} = \frac{\alpha^{(k)}_{g} f\left(x_i; \btheta_g^{(k)}\right)}{\varphi\left(x_i; \bxi^{(k)}, \balpha^{(k)}\right)} .	
\end{equation}
At the $k+1$-th iteration, the M-step updates $\balpha$ and $\bxi$ such that 
\begin{equation}\label{eq:Malpha}
\balpha^{(k+1)} = \argmax_{{\ba} \in \mathbb{S}^{G-1}} \sum_{g=1}^G \sum_{i=1}^n w_{gi}^{(k)} \log(\alpha_g) + p(\ba), ~\mbox{and}
\end{equation}
\begin{equation}\label{eq:Mtheta}
\bxi^{(k+1)} = \argmax_{{\bxi \in \Xi }} \sum_{g=1}^G \sum_{i=1}^n w_{gi}^{(k)} \log f(x_i; \btheta_g).
\end{equation}
Let $K>0$ be the maximum number of EM updates. We define
${M}_n^{(K)} (\ba^{(0)}) =2 \{pl_n(\bxi^{(K)}, \balpha^{(K)}) -  pl_n(\hat{\bxi}_{0}, \ba_0)\}, $
where $\balpha_0 = (1/G,\dots, 1/G)\trans$. To improve the  performance, we choose a set of initial values $\{  \ba_1,\dots,\ba_T  \}$ and define the EM-test statistic as 	
${\rm EM}_n^{(K)} = \max\{ M_n^{(K)}(\ba_t),  t=1,\dots,T\}.$
Intuitively, under the homogeneous model,  $\bxi^{(K)}$ and $\hat{\bxi}_{0}$ are close to $\bxi_0$ and are close to each other, while under the heterogeneous model, $\bxi^{(K)}$ and $\hat{\bxi}_{0}$ are far away from each other. Hence, we reject the null hypothesis in (\ref{homotest})   when ${\rm EM}_n^{(K)}$ is large. In this paper, we always assume that $K \geq 3$ is a fixed number. In simulation and real data analyses, we  set $K=100$.  Simulation shows that EM-test is robust to the choice of $K$. When $K$ is chosen too large, EM-test tends to have slightly more false positives and $K=100$ is a reasonable choice (Supplementary Section  D).

With the EM-test statistic, we can use the following procedure to screen for the cluster-relevant features. Let ${\rm EM}_{nj}^{(K)}$ be the EM-test statistic corresponding to the $j$th hypothesis testing problem (\ref{homotest}). Given a threshold $t_n>0$, if ${\rm EM}_{nj}^{(K)}<t_n$, we will screen out the $j$th feature; Otherwise, we will retain the $j$th feature as a cluster-relevant feature. Theoretical results in Section \ref{sec:theory} show that if we choose $t_n = n^{\vartheta}$ ($0 < \vartheta < 1$), this feature screening procedure can have the sure independent screening property or even the consistency of selection property.

\subsection{Notations} 
We use ${\rm diam}(\Xi)$ to represent the Euclidean diameter of $\Xi$. Denote by $|\cdot|$ the absolute value of a real number or cardinality of a set. For two sequences of random variable  $\left\{a_n\right\}_{n=1}^{\infty}$ and $\left\{b_n\right\}_{n=1}^{\infty}$, we write  $a_n=o_p\left(|b_n|\right)$ if $a_n / |b_n| \rightarrow 0 $ in probability,  and $a_n=O_p\left(|b_n|\right)$ if there exists a positive constant $C$ such that $a_n \leq C |b_n|$  in probability.  For real numbers $a$ and $b$, let $a \wedge b=\min \{a, b\}$ and $a \vee b=\max \{a, b\}$.  Define $\left\|{\bf a} \right\|_2 = \sqrt{\sum_{i=1 }^n a_i^2}$ as the $L_2$-norm of the vector ${\bf a} = (a_1,\dots,a_n)\trans \in \mathbb{R}^n$, and ${\rm vech }(\A) = \left(a_{11}, a_{22}, \ldots, a_{dd}, a_{12}, \ldots, a_{1d}, \ldots, a_{d(d-1)}\right)\trans \in \mathbb{R}^{d(d+1)/2}$ as the vectorization of the symmetric $d$-dimensional matrix $\A = (a_{ij})$. We use $\A \succeq {\bf 0}$ to represent that the matrix $\A$ is positive semi-definite. For a random variable $X$, we define its sub-exponential norm as 
$\|X\|_{\psi_1}=\inf \{t>0: \mathbb{E} \left(\exp (|X| / t)\right) \leq 2\}$.  If $X$ is a random variable from a homogeneous model $\mathbb{H}_0$ and $g(x)$ is a function, we define the $L_m$-norm of $g(X)$ as $\left\| g(X) \right\|_{L^m} = \left( \Ex \left[g^m(X) \right]\right)^{1/m}$. 

We denote $\ba^*=(\alpha_1^*,\dots,\alpha_G^*)$ as the true proportions of the $G$ clusters and $\bxi^*_{j} = (\btheta^*_{j1}, \ldots, \btheta^*_{jG})$ as the true parameters of the mixture model corresponding to the $j$th feature. We assume that $ \ba^*$ is fixed and $\min_{g} \alpha_g^* > 0$. We use $\delta>0$ as a fixed very small constant. If the $j$th feature is from the homogeneous model $\mathbb{H}_{j0}$ (cluster-irrelevant), we write $\bxi^*_{j} = \bxi_{j0} = (\btheta_{j0}, \ldots, \btheta_{j0})$ as its true parameters. When appropriate, we drop the subscript $j$ and use $\bxi^*$ as the true parameters of a general mixture model and $\bxi_0 = (\btheta_{0}, \ldots, \btheta_{0})$ as the true parameter of some general homogeneous model $\mathbb{H}_0$.  We always assume that $\btheta_{0} $ is an interior point of $\Theta$ and use $\balpha_0 = (1/G,\dots, 1/G)\trans$.
Let 
\begin{eqnarray}\label{eq:b}
Y_{ih} = \frac{1}{f(x_i,  \btheta_0)} \frac{\partial f(x_i, \btheta_0)}{\partial \theta_h}, ~ Z_{ih} = \frac{1}{2f(x_i,  \btheta_0)}\frac{\partial^2 f(x_i,  \btheta_0)}{\partial \theta_h^2},\nonumber \\ 
U_{ih\ell} = \frac{1}{f(x_i,  \btheta_0)}\frac{\partial^2 f(x_i,  \btheta_0)}{\partial \theta_h \partial \theta_\ell} (h< \ell),~\b_{1i} = (Y_{i1}, \ldots, Y_{id})\trans, \\
\b_{2i}= (Z_{i1},\ldots,Z_{id},U_{i12},\ldots,U_{i(d-1)d})\trans, ~\mbox{and}~\b_i = \left(\b_{1i}\trans,\b_{2i}\trans \right)\trans. \nonumber
\end{eqnarray}
%Given $h,\ell \in \{1,\dots,d\}$, let 
%$${m}_{1h}(\ba, \bxi_1, \bxi_2) = \sum_{g=1}^G {\alpha}_g\left({\theta}_{1gh} - {\theta}_{2gh}\right),~ {m}_{2h}(\ba, \bxi_1, \bxi_2) =  \sum_{g=1}^G {\alpha}_g\left({\theta}_{1gh} - {\theta}_{2gh}\right)^2,$$
%$$\mbox{ and }{s}_{h\ell}(\ba, \bxi_1, \bxi_2) =\sum_{g=1}^G {\alpha}_g\left({\theta}_{1gh} - {\theta}_{2gh}\right)\left({\theta}_{1g\ell} - {\theta}_{2g\ell}\right) ~(h<\ell),$$
%where $\bxi_j=(\btheta_{j1},\dots,\btheta_{jG})$ and $\btheta_{jg}=(\theta_{jg1},\dots,\theta_{jgd})$ ($j=1,2$, $g=1,\dots,G$).   
%We define two vector functions as  
%\begin{equation}\label{eq:m1}
%{\m}_1(\ba, \bxi_1, \bxi_2) = ({m}_{11}, \dots, {m}_{1d}) \trans, 
%\end{equation}
%\begin{equation}\label{eq:m}
%{\m}(\ba, \bxi_1, \bxi_2) = ({m}_{11}, \dots, {m}_{1d},{m}_{21},  \dots, {m}_{2d}, {s}_{12},  \dots, {s}_{(d-1)d} ) \trans, 
%\end{equation}
%and simplify ${\m}_1(\ba, \bxi, \bxi_0)$ and ${\m}(\ba, \bxi, \bxi_0)$ as $\m_1(\ba, \bxi)$ and $\m(\ba, \bxi)$, respectively.

\section{Theoretical results}\label{sec:theory}
In this section, we investigate the theoretical properties of the screening procedure.  
Without loss of generality, we assume that there is only one initial value $\ba^{(0)}$ (i.e. $T = 1$) and $\min_{g=1,\ldots,G}$ $27^{-1}{\alpha}_g^{(0)} \geq \delta>0$.

We first present the main theoretical result of the paper---the feature screening property of the EM-test statistic. Define the set of cluster-irrelevant features as $S_0 =  \{j: 1\leq j \leq p, \bxi^*_j=(\btheta_{j0}, \dots, \btheta_{j0})\}$, and the set of cluster-relevant features as $S_1 = \{1,\dots,p \}  \backslash S_0$. Denote $s = |S_1|$ as the number of cluster-relevant features.
For a small fixed $\gamma>0$, we define 
\begin{equation}\label{eq:Xi1}
\Xi_{1} = \left\{ \bxi :  \max_{g \neq g'} \Vert\btheta_g - \btheta_{g'} \Vert_2 \geq \gamma , \bxi \in \Xi \right\}
\end{equation}
as the parameter set of heterogeneous models with a minimum component-difference $\gamma$. When $j \in S_1$, we assume $\bxi^*_j \in \Xi_1$. Given a threshold $t_n >0$, we define
\begin{equation}\label{eq:retained}
\hat{S}_1(t_n) = \{  1\leq j \leq p: {\rm EM}_{nj}^{(K)} \geq t_n  \}
\end{equation}
as an estimator of $S_1$. Under Condition (\Ca)--(\Ch) that will be specified in Section \ref{sec:boundh0} and \ref{sec:boundh1}, the following theorem guarantees that the screening procedure based on the EM-test statistic can effectively filter cluster-irrelevant features while retaining all cluster-relevant features with a high probability.

\begin{theorem}\label{thm:main}
	Assume that for any cluster-relevant feature $j\in S_1$, $\bxi^*_{j} \in \Xi_1$, where $\Xi_1$ is defined as in (\ref{eq:Xi1}). Under Condition (\Ca)--(\Ch), given a fixed $K$, choosing the threshold $t_n = n^{\vartheta} (0 < \vartheta < 1)$, when $n$ is sufficiently large, we have
	$$ 
	\Pr\left( S_1 \subset \hat{S}_1(t_n)\right) \geq 1 -  s\exp\left(-C_3 n^{1/2}  + C_4 n^{\vartheta - 1/2}\right),  \mbox{ and } $$
	\begin{align*}
	\Pr\left( S_1 =  \hat{S}_1(t_n)\right) &\geq 1 - (p-s) \left((C_1 n)^{-m/4} + (C_2n)^{-\vartheta m} \right) \\
	& \quad  -  s\exp\left(-C_3 n^{1/2}  + C_4 n^{\vartheta - 1/2}\right),
	\end{align*}
	where $C_1,C_2,C_3$ and $C_4$ are four constants depending on $K, G,d, \delta$, $  {\rm diam}(\Xi)$ and the constants specified in Condition (C3)--(C7),  $s = |S_1|$ and $m$ is the integer in Condition (\Cc).
\end{theorem}
%\textbf{\colred In Theorem \ref{thm:main}, we do not need to assume that different features are independent and follow the same parametric family.  Even if the features are dependent and contain both continuous and count data, the EM-test statistic can also correctly screen out the cluster-irrelevant features.} 
 If $p$ does not go to infinity too fast, Theorem \ref{thm:main} implies that we can achieve the sure independent screening property or model selection consistency in high dimensional settings.
\begin{itemize}
	\item If $p = {O} (\exp(n^\beta)),  0 < \beta < 1/2$, we have $\Pr\left( S_1 \subset \hat{S}_1(t_n)\right) \rightarrow 1, $ as $n \rightarrow \infty$. Thus, the feature screening method based on the EM-test statistic has the sure independent screening property. 
	\item If $p= {O} ( n^{\kappa})$ with $0<\kappa<m/\max\{4, \vartheta^{-1}\}$, we have $ \Pr\left( S_1 =  \hat{S}_1(t_n)\right)  \rightarrow 1,$ as $n \rightarrow \infty$, or in other words, we can achieve model selection consistency. The condition $ \kappa < m/\max\{4, \vartheta^{-1}\}$ is a very lenient condition. For most common parametric distribution families, $m$ in Condition (\Cc) can be taken as any positive integer and thus $\kappa$ can be any positive number.
\end{itemize}

Empirical studies show that choosing $\vartheta \in [0.3, 0.35]$ can make a good balance between the type I and type II error (Supplementary Section D) and hence we suggest to choose $\vartheta \in [0.3, 0.35]$ in real applications.  Note that in Theorem 1, for notational simplicity, we assume that different features are in the same parametric family. The similar screening property can also be proved even if different features are in different parametric families (e.g. some are continuous variables and some are count variables), as long as these features satisfy conditions similar to the ones in Theorem 1. In addition, Theorem 1 does not need to assume that different features are independent. Even if the features are dependent, the same screening properties also hold. 

The proof of Theorem \ref{thm:main} is based on the tail probability bounds of the EM-test statistic under the null and alternative hypotheses. Under $\mathbb{H}_0$, we show that the EM-test is bounded with a high probability, and under $\mathbb{H}_1$, the EM-test statistic will diverge to infinity with a high probability. We present these tail probability bounds in Section \ref{sec:boundh0} and \ref{sec:boundh1}. The proof of Theorem \ref{thm:main} is in Supplementary material.  In the Supplementary material, we show that many commonly used distributions, such as many exponential family distributions and the negative binomial distributions, satisfy the conditions in Theorem 1, and thus the screening properties hold for these distributions.

\subsection{The probability bound of the EM-test statistic under $\mathbb{H}_0$}\label{sec:boundh0}
We need the following regularity conditions before presenting the tail probability bounds of the EM-test statistic under $\mathbb{H}_0$.

\begin{itemize}  
	\item[(\Ca)]  
	For every $\btheta_{0} \in \bTheta$ and  sufficiently small ball $V \subset \Theta$ around $\btheta_{0}$,  we assume that the function  $\sup_{\btheta \in V}  {f^{1/2}(x; \btheta)}{f^{1/2}(x; \btheta_0)}  $ is measurable and 
	$\int \sup_{\btheta \in V}  {f^{1/2}(x; \btheta)}{f^{1/2}(x; \btheta_0)}   \mu(\rd x) < \infty.$
	In addition, for every sufficiently small ball $U \subset \Xi$ around $\bxi_{0}=(\btheta_{0},\cdots,\btheta_{0})$ and $\ba \in \mathbb{S}^{G-1}$, we assume that the function   $\sup_{\bxi \in U} \log \{ \varphi(x;  \bxi, \ba)  \}$ is measurable and  
	$\Ex (\sup_{\bxi \in U} \log \{ 1+\varphi(x;  \bxi , \ba)  \}) <  + \infty.$   
	
	\item[(\Cb)] The density function $f(x;\btheta)$ has a common support for $\btheta\in \Theta$ and continuous 5th order partial derivatives with respect to $\btheta$.
	\item[(\Cc)] Let $m > 0$ be an integer and $M> 0$ be a constant. There are a function $g(x,\btheta )>0$ with $\sup_{\btheta \in \Theta}\left\|g(x; {\btheta})\right\|_{L^{8m}} \leq M$, a function $r(x)>0$ with $\int  r^2(x) \mu(\rd x) \leq M$ and a constant $\tau>0$, such that, for all $\btheta_0 \in \Theta$, $h =1, \dots, 5$ and ${j_1},\ldots,{j_h} \in \{1,\dots,d\}$, 
	$$\sup_{\Vert\btheta - \btheta_0\Vert_2 \leq \tau} \left|\frac{\partial^h f(x; \btheta)}{\partial \theta_{j_1} \cdots \partial \theta_{j_h}} \bigg/  f(x; \btheta_0)\right| \leq g(x;{\btheta_0}), \mbox{ and }$$
	$$\sup_{\btheta \in \Theta}f(x;\btheta) + \sup_{\btheta \in \Theta} \frac{1}{{f(x;\btheta)}}\left\| \frac{\partial f(x;\btheta)}{\partial \btheta} \right\|_2^2 \leq r^2(x).$$
	
	\item [(\Cd)] The minimum eigenvalue $\lambda_{\rm min}  (\B(\btheta_0))$ of the covariance matrix $\B(\btheta_0) = \cov(\b_i)$ satisfies
	$\lambda_{\rm min} := \inf_{\btheta_{0} \in \Theta} \lambda_{\rm min}(\B(\btheta_0)) > 0.$
\end{itemize}
Condition (\Ca) is the Wald consistency condition, which can be founded in \citet{van2000asymptotic}. It also ensures the continuity of the Hellinger distance which we will define below. Condition (\Cb) guarantees the smoothness of $f(x ; \btheta)$.   Condition (\Cc) ia a technical  condition  on the  partial derivatives.  It  guarantees that there is a dominating function of the remainder term in the Taylor expansion, and thus allows us to give polynomial tail probability bounds of the higher-order infinitesimal terms of the EM-test statistic under $\mathbb{H}_0$.  Condition (\Cd) is the strong identifiability condition \citep{chen1995optimal, nguyen2013convergence}. Most of the commonly used one-parameter distributions, such as the Poisson distribution and the exponential distribution,  satisfy Condition (\Cc) and (\Cd).  Many multiple-parameter distributions including the negative binomial distribution and the gamma  distribution also satisfy Condition (\Cc) and (\Cd).

\begin{theorem} \label{thm:H01}
	Assume that $x_1,\dots,x_n$ are independent samples from the homogeneous distribution $f(x;\btheta_0)$.  
	Under Condition (\Ca)--(\Cd),  
	for any  $t>0$, when  
	$n$ is sufficiently large, 
	we have 
	$$\Pr\left({\rm EM}_n^{(K)}  \leq t+C n^{-1/16} {\rm log}^{3/2}n\right) \geq  1-  { (C_1n)^{-m/4} } -  {(C_2t)}^{-m} ,$$
	where $C, C_1$ and $C_2 $ are three  positive constants  depending on $\tau, K, G,d, \lambda_{\rm min}, m, M, \delta$ and $  {\rm diam}(\Xi)$. 
\end{theorem}

Observe that  when $n \rightarrow \infty$,   $C n^{-1/16} {\rm log}^{3/2}n  $  approaches to zero.  Therefore, roughly speaking, Theorem \ref{thm:H01} shows that when $n$ is sufficiently large, under $\mathbb{H}_0$, the tail probability of the EM-test statistic greater than $t $ has a polynomial decay rate. To prove Theorem \ref{thm:H01}, we first derive the tail probability bound for the mixture parameter estimators $\bxi^{(k)}$ by analyzing the empirical processes indexed by $\bxi$ \citep{wong1995probability}. Then, we analyze the Taylor expansion of ${\rm EM}_n^{(K)}$ and bound each term in the expansion using concentration inequalities \citep{wainwright2019high}. Details of the proof are given in the Supplementary material. 

\subsection{The probability bound of the EM-test statistic under $\mathbb{H}_1$} \label{sec:boundh1}
Our next goal is to show that the EM-test statistic diverges to infinity under $\mathbb{H}_1$ with a high probability. Recall that $\ba^*$ is the true proportion parameter and that under $\mathbb{H}_1$, $\bxi^* \in 	\Xi_{1} $, where $\Xi_{1}$ is defined in (\ref{eq:Xi1}).
We define 
%\begin{equation}\label{eq:theta0dag}
%{\btheta}^{\dagger}_0 = \argmax_{\btheta \in {\Theta}} \Ex_{\balpha^*,\bxi^{*}}  \left[\log f(x;\btheta)\right]
%\end{equation}
${\btheta}^{\dagger}_0 = \argmax_{\btheta \in {\Theta}} \Ex_{\balpha^*,\bxi^{*}}  \left[\log f(x;\btheta)\right]$
as the parameter of the homogeneous model that is closest to the true heterogeneous model in terms of the Kullback-Leibler divergence.   
Denote $\bxi_0^{\dagger} = \left({\btheta}^{\dagger}_0, \dots, {\btheta}^{\dagger}_0\right) $. 
Similarly, given an initial value $\ba^{(0)}$, we can find a heterogeneous model with a proportion parameter $\ba^{(0)}$ that is closest to the true heterogeneous model and denote its parameter as
%\begin{equation}\label{eq:xidag}
%{\bxi}^{\dagger} = \argmax_{\bxi \in \Xi} \Ex_{\balpha^*,\bxi^{*}}  \left[  \log \ \varphi\left(x;\bxi, \ba^{(0)}\right) \right].
%\end{equation}
$
{\bxi}^{\dagger} = \argmax_{\bxi \in \Xi} \Ex_{\balpha^*,\bxi^{*}}  \left[  \log \ \varphi\left(x;\bxi, \ba^{(0)}\right) \right].
$
Note that ${\bxi}^{\dagger}$ and ${\btheta}^{\dagger}_0$ depend on the true value $\ba^*, \bxi^* $. 

Define $R(x;\bxi^*) =  \log \ \varphi\left(x;{\bxi}^{\dagger}, \ba^{(0)}\right) -\log\ f\left(x;{\btheta}_0^{\dagger}\right)$ as the difference between two ``working"  log-likelihood $\log \ \varphi\left(x;{\bxi}^{\dagger}, \ba^{(0)}\right)$ and $\log\ f\left(x;{\btheta}_0^{\dagger}\right)$. If the initial value $\ba^{(0)}$ is close to the true proportion $\ba^*$, the expectation of $R(x;\bxi^*)$ would be bounded away from zero. So, the one-step EM-test statistic, and thus the EM-test statistic, would be large. Thus, we would correctly reject the null hypothesis with a high probability. Furthermore, denoting $D(\btheta) = \Ex_{\balpha^*,\bxi^{*}} [\log f(x;\btheta)]$, we define a mean-zero empirical
process indexed by $\btheta\in\Theta$ as
$$Z_{\btheta}(\bxi^*) =   n^{-1/2} \sum_{i=1}^n \left\lbrace  \log \ f(x_i;\btheta) - \log \ f\left(x_i;{\btheta}^{\dagger}_0\right) - \left[D(\btheta )- D\left({\btheta}^{\dagger}_0\right) \right]  \right\rbrace .$$ 
We need the following conditions under $\mathbb{H}_1$. 

\begin{itemize}
	\item[(\Cf)] The initial value $\balpha^{(0)}$ fulfills
	$\varrho = \inf_{\bxi^{*} \in \Xi_1  }  \Ex_{\balpha^*,\bxi^{*}} \left[ R(x;\bxi^*)  \right]  > 0.$
	\item[(\Cg)] 
	There exists a constant $M_{\psi_1}$ such that
	$\sup_{\bxi^{*} \in \Xi_1  }   \left\| R(x;\bxi^*) - \Ex_{\balpha^*,\bxi^{*}} \left[R(x;\bxi^*)\right]\right\|_{\psi_1} \leq M_{\psi_1}.$
	\item [(\Ch)] $\{Z_{\btheta}(\bxi^*) : \btheta \in \Theta\}$ is a $\psi_1 $- process such that for any $\btheta, \btheta' \in \Theta$,
	$\sup_{\bxi^{*} \in \Xi_1  }  \Vert Z_{\btheta}(\bxi^*) - Z_{\btheta'}(\bxi^*)\Vert_{\psi_1} \leq C_{\rho}\Vert \btheta - \btheta' \Vert_2,$
	where $C_{\rho}>0$ is a constant.	
\end{itemize}

Condition (\Cf) is a key assumption. Because the EM algorithm cannot guarantee convergence to the global maximum, we need to choose an initial value $\ba^{(0)}$ such that the theoretical best heterogeneous model that we can achieve in one step EM update is uniformly closer to the true heterogeneous model than the best homogeneous model. Note that we always have $\varrho \geq 0$, but it is hard to give a necessary and sufficient condition for the choice of $\ba^{(0)}$ such that $\varrho > 0$. However, we can show that if $\left\| \ba^* - \ba^{(0)} \right\|_2 \leq \tau(\gamma)$  for some constant $\tau(\gamma) $,  Condition (\Cf) holds (see Section C.1 in Supplementary material for more discussion).  Condition (\Cg) and  (\Ch) are two weaker conditions and hold for many commonly used distribution families.Under these conditions, we obtain the following tail probability bound of the EM-test statistic under $\mathbb{H}_1$.

\begin{theorem} \label{thm:H1res}
	Assume that $x_1,\dots,x_n$ are independent samples from the heterogeneous model distribution $\varphi(x;\ba^*, \bxi^*)$ with $\bxi^* \in \Xi_1$. 
	Under Condition (\Cf)--(\Ch),  for any $t$,  we have 
	\begin{align*}
	& \quad \Pr \left( {\rm EM}_n^{\left(K\right)} \geq 2^{-1}n\varrho  -  n^{1/2}C_J\left[J\left(D\right) +t\right] -   p_0\right) \\
	& \geq 1 - 2\exp\left[-C' \min\left(\frac{n\varrho^2}{4 M_{\psi_1}^2} , \frac{n\varrho}{2M_{\psi_1}}\right) \right] -    2\exp\left(\frac{-t}{D}\right),
	\end{align*}
	where $D, C_J, J(D)$ and $C'$ are four constants and defined in Lemma 17 and 18 in the Supplementary material and $ p_0 =  \lambda G \log(\delta G)$. 
\end{theorem}

\begin{corollary}\label{coro:H1}
	From Theorem \ref{thm:H1res},  for $t_n = n^{\vartheta}, 0 < \vartheta < 1$, there are two constants $C_3,C_4$ such that 
	$$ 
	\Pr\left( {\rm EM}_n^{\left(K\right)} \geq t_n \right) \geq 1 -  \exp\left(-C_3 n^{1/2}  + C_4 n^{\vartheta - 1/2}\right).  $$
\end{corollary}

Theorem \ref{thm:H1res}  and Corollary \ref{coro:H1} say that under $\mathbb{H}_1$,  by selecting a suitable threshold, the cluster-relevant  feature  can be retained with high probability. 

%We briefly describe the proof of Theorem \ref{thm:H1res}. Detailed proofs are given in the Supplementary material. Observe that the EM-test statistic is larger than the penalized log-likelihood ratio  $pl_n\left(\bxi^{(0)}, \ba^{(0)}\right) - pl_n\left(\hat{\bxi}_{0}, \ba_0\right)$, which can be decomposed as a summation of three parts, $pl_n\left(\bxi^{(0)}, \ba^{(0)}\right) -  pl_n\left({\bxi}^{\dagger}, \ba^{(0)}\right)$,  
%$pl_n\left({\bxi}^{\dagger}, \ba^{(0)}\right)-  pl_n\left({\bxi}_{0}^{\dagger}, \ba_{0}\right)$
%and $pl_n\left({\bxi}_{0}^{\dagger}, \ba_{0}\right) - pl_n\left(\hat{\bxi}_{0}, \ba_0\right)$. 
%All three parts can be bounded. The first part is non-negative.  
%The second part can be written as 
%$$pl_n\left({\bxi}^{\dagger}, \ba^{(0)}\right)-  pl_n\left({\bxi}_{0}^{\dagger}, \ba_{0}\right) = \sum_{i=1}^n R(x_i;\bxi^*) + p(\ba^{(0)}) -  p(\ba_{0}),$$
%and can be bounded using the Bernstein inequality. 
%For the third part, since $ D(\btheta )\leq  D\left({\btheta}^{\dagger}_0\right) $ for all $\btheta \in \Theta$, we have 
%$$pl_n\left({\bxi}_{0}^{\dagger}, \ba_{0}\right) - pl_n\left(\hat{\bxi}_{0}, \ba_0\right) =  \sum_{i=1}^n \left\lbrace  \log \ f(x_i;\hat{\btheta}_0) - \log \ f\left(x_i;{\btheta}^{\dagger}_0\right) \right\rbrace  \geq  -n^{1/2}\sup_{\btheta \in \Theta} Z_{\btheta}(\bxi^*). $$
%Thus, the third part can be bounded by analyzing the supremum of the empirical process $\{Z_{\btheta}(\bxi^*), \btheta \in \Theta\}$ using the generalized Dudley inequality.

\subsection{The limiting distribution of the EM-test statistic under $\mathbb{H}_0$} \label{sec:Asym}
In many applications, giving a valid $p$-value of the retained feature is also crucial. In this section, we give the limiting distribution of the EM-test statistic under $\mathbb{H}_0$. 
To derive the limiting distribution, we only need the following weaker conditions in replacement of Condition (\Cc) and (\Cd). 
\begin{itemize}
	\item [(\WCc)] For all $h = 1, \dots, 5$ and $\theta_{j_1},\ldots,\theta_{j_h}$, there exists a function $g(x; \btheta_0) \geq 0$ and a constant $\tau>0$ such that 
	$$\sup_{\Vert\btheta - \btheta_0\Vert \leq \tau} \left|\frac{\partial^h f(x,  \btheta)}{\partial \theta_{j_1} \cdots \partial \theta_{j_h}} \bigg/  f(x,  \btheta_0)\right| \leq g(x; \btheta_0)$$
	and $\left\|g(x; \btheta_0)\right\|_{L^3} < \infty$.
	\item [(\WCd)] The covariance matrix $\B(\btheta_{0})=\cov(\b_i)$ of $\b_i$  is positive definite. 
\end{itemize}

%We first derive
%the upper bound of $\left\|\bxi^{(k)} - \bxi_0\right\|_2$.  Namely, we provide the following results. 
%\begin{theorem}\label{rate}
%	Assume that $x_1,\dots,x_n$ are independent samples from the homogeneous distribution $f(x;\btheta_0)$. Under Condition (\Ca)--(\Cb) and (\WCc)--(\WCd),  given any initial value $\ba^{(0)} \in \mathbb{S}^{G-1}$, for any fixed $K>0$ and $1 \leq k \leq K$, we have 
%	$\left\|\balpha^{(k)} - \balpha^{(0)}\right\|_2 = o_p(1), \  \left\|\bxi^{(k)} - \bxi_0\right\|_2 = O_p\left(n^{-1/4}\right)$
%	and 
%	$\left\|\sum_{g=1}^G {\alpha}_g^{(k)} \left({\btheta}_g^{(k)} - \btheta_0\right)\right\|_2 = O_p\left(n^{-1/2}\right).$
%\end{theorem}
%%\begin{remark}
%Theorem \ref{rate} says that the convergence rate of $ \bxi^{(k)}$ under the homogeneous model $\mathbb{H}_0$ is only $O_p\left(n^{-1/4}\right)$, but not the common convergence rate $O_p\left(n^{-1/2}\right)$. The reason is that under $\mathbb{H}_0$,  the heterogeneous model (\ref{eq:Gmix}) is unidentifiable, and, in consequence, the Fisher information matrix is not positive definite. However, the weighted average of $\bxi^{(k)} $, $\sum_{g=1}^G {\alpha}_g^{(k)} \left({\btheta}_g^{(k)} - \btheta_0\right) $, is a $\sqrt{n}$-consistent estimator. Before presenting the asymptotic distribution of the EM-test statistic, we first introduce a few notations. 
Let $r = \min(G-1,d)$ and 
\begin{equation} \label{eq:V}
\mathcal{V} = \left\{ {\rm vech }(\V): \V \in \mathbb{R}^{d \times d} \mbox{  is symmetric},~{\rm rank}(\V)\leq r, \V \succeq 0\right\}.
\end{equation}
For $j,k = 1,2$, let $\B_{jk} = \Ex_{\btheta_{0}}(\{\b_{ji} - \Ex(\b_{ji})\}\{\b_{ki} - \Ex(\b_{ki})\}\trans )$ and 
$\tilde{\b}_{2i} = \b_{2i} - \B_{21}\B_{11}^{-1}\b_{1i}$. The covariance matrix of $\tilde{\b}_{2i}$ is $\widetilde{\B}_{22} = {\B}_{22} - \B_{21}\B_{11}^{-1}\B_{12}$. 
\begin{theorem} \label{thm:limitdist}
	Assume that $x_1,\cdots,x_n$ are independent samples from the homogeneous distribution $f(x;\btheta_0)$. If $G\geq 2$ and one of the $\balpha_t$ ($t=1,\dots,T$) is $\ba_0$, then under Condition (\Ca)--(\Cb) and (\WCc)--(\WCd),  as $n \rightarrow \infty$,  we have
	$${\rm EM}_n^{(K)} \overset{d}{\longrightarrow} \sup_{{\bf v} \in \mathcal{V}} 2{\bf v}\trans {\bf{w}}  - {\bf v}\trans\widetilde{\B}_{22} {\bf v},$$
	where ${\bf{w}} = (w_1,\ldots,w_{d(d+1)/2})\trans$ is a zero-mean multivariate normal random vector with a covariance matrix $\widetilde{\B}_{22}$ and $\mathcal{V} $ is  as in (\ref{eq:V}). 
\end{theorem}

If $d = 1$, the limiting distribution in Theorem \ref{thm:limitdist} is $0.5 \chi^2_1 + 0.5 \chi^2_0$, the same as the one in \cite{li2009non}, while, if $G = 2$, it is the distribution in \cite{niu2011testing}. When $G > d$, we have $r=d$ and the limiting distribution will be independent of the component number $G$. Generally, it is computationally difficult to calculate the limiting distribution in Theorem \ref{thm:limitdist}. When $G>d$, the feasible domain 
$\mathcal{V} = \left\{ {\rm vech }(\V): \V \in \mathbb{R}^{d \times d}, \V \succeq 0\right\}$
is a positive semi-definite matrix cone. Computation of the limiting distribution in Theorem \ref{thm:limitdist} becomes a classic cone quadratic program and can be solved using the algorithms reviewed in \cite{vandenberghe2010cvxopt}, but these algorithms are still computationally expensive. However, it can be easily shown that $\sup_{{\bf v} \in \mathcal{V}} 2{\bf v}\trans {\bf{w}}  - {\bf v}\trans\widetilde{\B}_{22} {\bf v} \leq {\bf w}\trans\widetilde{\B}_{22}^{-1} {\bf w}$, and thus ${\rm EM}_n^{(K)}$ is stochastically less than  $\chi^2({d(d+1)/2})$. Therefore, we can always use $\chi^2({d(d+1)/2})$ as the limiting distribution. The test will be conservative, but our empirical studies show that the test still has a high power.

\section{Simulation}\label{sec:simu}
In this section, we use simulation to assess the performance of the EM-test statistic for feature screening and clustering of high dimensional count data. We compare the EM-test with feature screening and feature selection methods.  The feature screening methods include Dip-test \citep{chan2010using}, KS-test \citep{jin2016influential}, COSCI \citep{banerjee2017feature}, SC-FS \citep{liu2022clustering} and a baseline method based on the goodness-of-fit test. Dip-test screens features by investigating the unimodality of the data distribution. For the baseline method, we use the Chi-square test to test the fit of the data to the null distribution. The feature selection method is the Sparse kmeans method (abbreviated as Skmeans) proposed in \citet{witten2010framework}.   Dip-test, KS-test and COSCI are methods for continuous data. When applying these to the simulated count data, we first $\log$ transform the data ($\log(x+1)$) to make them more like continuous data.

\subsection{Simulation Setup}\label{subsec:SimGen}

In the simulations, we set the number of clusters as $G=5$ and the proportions of the clusters as $\ba =(\alpha_1,\dots,\alpha_5) = (0.5,0.125,0.125,0.125,0.125)$. The sample size is set as $n=1000$. The dimension is set as $p=500$, $5000$ or $p=20,000$. Skmeans  and COSCI are not evaluated for $p=20,000$ because they are computationally too expensive for this ultra-high dimensional setting. The number of cluster-relevant features is fixed at $s=20$. We always set the first $20$ features as the cluster-relevant and all other features as cluster-irrelevant.

%\textbf{\colred In the Supplementary, we give many distribution examples that satisfy  Condition (\Ca)--(\Ch) and here we consider the distribution family $\mathcal{P}$ as the negative binomial distributions ${\rm NB} (\mu, r)$, where $\mu$ and $r$ are the mean and over-dispersion parameters, respectively. }

More specifically, for the $i$th sample, we first randomly assign it to a cluster $g$ with the probability $\alpha_g$. Then, if the $j$th feature is cluster-relevant ($j=1,\dots,20$), we randomly sample $x_{ij}$ from ${\rm NB} (\mu_{gj}, r_{j})$; If it is cluster-irrelevant ($j=21,\dots,p$), we randomly sample $x_{ij}$ from ${\rm NB} (\mu_j, r_j)$. The mean and over-dispersion parameters of the negative binomial distributions are randomly generated (see below). 

We consider simulation setups of two noise levels (low or high) and three cluster signal strength levels (low, medium and high).  
The over-dispersion parameters $r_j$ represent the noise level of the data and the differences of the mean parameters $\mu_{gj}$ between clusters represent the cluster signal strength.  The details of generating $r_j$ and $\mu_{gj}$ are given in the Supplementary material.
%For the low-noise and high-noise scenarios, we independently generate $r_j$ from the uniform distributions ${\rm U}(10,11)$ and ${\rm U}(5,6)$, respectively. For the cluster-relevant features ($j=1,\dots,20$), the mean parameters $\mu_{gj}$ are either set as $\exp(u_j)$ or $\exp(u_j)+D_j$, where $u_j$ is generated from $ {\rm U} (\log \ 2, \log \ 5)$, and $D_j$ is to control the signal strength (the differences between clusters). We generate $D_j$ from ${\rm U}(5,6)$, ${\rm U}(7,8)$ or ${\rm U}(9,10)$ for the low, medium and high signal strength settings, respectively. For the first 5 features ($1\leq j\leq 5$), we set $\mu_{2j}=\exp(u_j)+D_j$ and $\mu_{gj} = \exp(u_j) (g\neq 2)$. Similarly, for $5k+1\leq j\leq 5k+5 (k=1,2,3)$, we set $\mu_{k+2,j}=\exp(u_j)+D_j$ and $\mu_{gj} = \exp(u_j) (g\neq k+2)$. For all cluster-irrelevant features ($j=21,\dots,p$), we set $\mu_j=\exp(u_j)$, where $u_j$ is generated from $ {\rm U} (\log \ 2, \log \ 5)$. 
Thus, in total, we have 18 different simulation setups (3 dimension setups $\times$ 2 noise levels $\times$ 3 signal levels). In each simulation setup, we generate 100 datasets.

 To evaluate the performance of EM-test on continuous data, we also generate simulation data based on normal distributions. The simulation setups are  detailed in Supplementary material. EM-test for normal models are used for these continuous data. To investigate the robustness of EM-test to model mis-specification, we further generate count data based on Poisson-truncated-normal and the binomial-Gamma distributions. EM-test for negative binomial model is used for these count data (Supplementary Section D). In the following, we only discuss the negative binomial simulations. For the continuous data simulation, we find that EM-test performs similar to other available methods under easier simulation setups and outperforms other methods under more difficult setups (Supplementary Section D). From the mis-specified count data simulations, we find that EM-test is robust to model mis-specification and outperforms other methods (Supplementary Section D).

\subsection{Performance on feature screening}\label{subsec:FS}
We first evaluate the accuracy of feature screening. For different screening methods, we first rank the features by their  corresponding test statistics / p-values or feature weights. Following previous researches about feature screening \citep{zhu2011model,li2012feature}, let $\mathcal{S}$ be the minimum number of features needed to include all cluster-relevant features in a rank.  Table \ref{tab:S} shows  the mean and the standard deviation of $\mathcal{S}$ over the 100 replications. Table 1 does not include COSCI  because it only reports a selected feature index but does not provide an order of all features.  In the low dimensional cases ($p=500$), all count-data methods work well. EM-test only needs 20 or slightly more than 20 features to include all cluster-relevant features. In higher dimensions ($p=5000$ or $20,000$), the EM-test outperforms other methods, often by a large amount. For example, in the medium signal, high noise and $p=20,000$ case, the EM-test needs around 21 features to include all cluster-relevant features, while other methods need over a thousand features. SC-FS works well in lower dimensional cases, but its performance deteriorates in higher dimensional cases, especially in higher dimensional cases with lower signal to noise ratios. This is because SC-FS needs a pre-cluster label for feature screening. When $p$ is large, the pre-cluster results can be very inaccurate, leading to the inferior performance of SC-FS in these settings.   The continuous methods (KS-test and Dip-test) do not perform well for these count data. 
%because almost all of the p-values for the KS-test and Dip test are zero in this simulation.
% Table generated by Excel2LaTeX from sheet 'S'
\begin{table}
	\centering
	\caption{The  mean of the minimum model size $\mathcal{S}$ over 100 replications. The numbers in the parenthesis are the standard deviation of $\mathcal{S}$ over 100 replications.}
\renewcommand\arraystretch {1.2}
\footnotesize
\setlength{\tabcolsep}{1.2mm}
	\begin{tabular}{ccccccc}
			\hline
		\hline
	$p$	& EM-test & Chi-square  & SC-FS & Skmeans & KS-test & Dip-test \\
	\hline 
	\multicolumn{7}{c}{Case 1:  High signal  and low noise} \\  
	\hline  
		500   & 20.1 (0.4) & 20.9 (3.0) & 20.0 (0.0) & 20.0 (0.0) & 499.4 (3.1) & 499.2 (2.7) \\
		5000  & 20.7 (0.8) & 28.0 (24.8) & 212.0 (622.7) & 1048.0 (1226.5) & 4994.6 (12.6) & 4988.0 (31.3) \\
		20,000 & 22.7 (1.7) & 47.8 (47.9) & 14675.7 (4233.6) & NA  & 19968.2 (90.6) & 19968.2 (80.6) \\
				\hline 
		\multicolumn{7}{c}{Case 2:  High signal  and high noise} \\  
		\hline  
		500   & 20.0 (0.1) & 24.6 (11.4) & 20.0 (0.0) & 20.0 (0.0) & 499.1 (3.1) & 498.4 (4.1) \\
		5000  & 20.0 (0.2) & 75.2 (86.5) & 871.7 (1276.0) & 554.8 (1047.3) & 4994.4 (12.8) & 4981.0 (48.8) \\
		20,000 & 20.2 (0.4) & 335.5 (689.1) & 16921.6 (3076.0) & NA  & 19977.3 (43.2) & 19955.4 (104.7) \\
		\hline 
		\multicolumn{7}{c}{Case 3:  Medium signal  and low noise} \\  
		\hline  
		500   & 20.2 (0.6) & 41.9 (29.4) & 20.0 (0.0) & 20.0 (0.0) & 498.8 (3.7) & 499.0 (2.8) \\
		5000  & 22.1 (13.2) & 276.6 (348.1) & 1140.1 (1319.5) & 762.5 (1204.2) & 4990.4 (23.4) & 4985.3 (37.0) \\
		20,000 & 23.1 (2.6) & 901.5 (1200.0) & 16747.7 (2854.5) & NA  & 19956.4 (96.6) & 19952.5 (119.7) \\
		\hline 
		\multicolumn{7}{c}{Case 4:  Medium signal  and high noise} \\  
		\hline
		500   & 20.4 (2.1) & 80.2 (70.1) & 20.0 (0.0) & 20.0 (0.0) & 498.1 (5.9) & 497.5 (6.9) \\
		5000  & 20.8 (2.5) & 581.9 (581.6) & 2411.3 (1474.8) & 1814.1 (1438.5) & 4990.4 (22.9) & 4976.9 (48.8) \\
		20,000 & 45.2 (147.1) & 2103.1 (1799.7) & 17306.9 (2630.9) & NA  & 19938.5 (103.2) & 19917.7 (193.7) \\
			\hline 
		\multicolumn{7}{c}{Case 5:  Low signal  and low noise} \\  
		\hline  
		500   & 33.3 (27.0) & 202.3 (100.1) & 20.0 (0.1) & 20.0 (0.0) & 497.5 (5.6) & 498.4 (3.6) \\
		5000  & 129.8 (225.3) & 1795.7 (844.8) & 3349.3 (1203.5) & 2589.0 (1042.7) & 4971.4 (56.5) & 4974.6 (62.8) \\
		20,000 & 755.2 (1635.7) & 7964.7 (4051.5) & 18283.2 (1422.1) & NA  & 19885.6 (197.8) & 19915.5 (182.9) \\
			\hline 
		\multicolumn{7}{c}{Case 6:  Low signal  and high noise} \\  
		\hline
		500   & 58.9 (56.7) & 251.9 (97.1) & 20.1 (0.8) & 20.0 (0.0) & 496.3 (7.6) & 495.8 (8.3) \\
		5000  & 342.7 (496.9) & 2221.6 (965.0) & 4182.7 (783.2) & 3175.8 (1029.2) & 4979.0 (42.4) & 4971.7 (54.0) \\
		20,000 & 1425.6 (2429.9) & 9993.0 (3870.1) & 18696.1 (1263.5) & NA  & 19847.3 (250.0) & 19903.6 (187.2) \\
		\hline
	\end{tabular}%
	\label{tab:S}%
\end{table}%

The minimum model size  $\mathcal{S}$ measures the feature ranks given by different methods. However, in clustering analysis, simply having $\mathcal{S}$ close to $s$ is inadequate because we need a criterion to determine which features to retain. Therefore, we further compare the number of correctly retained cluster-relevant features (denoted as $\mathcal{R}$) and falsely retained cluster-irrelevant features (denoted as $\mathcal{F}$) by different methods. For SC-FS, COSCI and Skmeans, we use their default parameters to select the cluster-relevant features. For the EM-test, we select the features by the adjusted p-values (EM-adjust, adjusted p-value $<$  0.01) and by the threshold $n^{0.35}$ (EM-0.35). The p-value is calculated using the $\chi^2(3)$-distribution, because compared with the limiting distribution in Theorem \ref{thm:limitdist}, the $\chi^2(3)$-distribution is computationally more efficient, achieves good sensitivity and false discovery rate (FDR) control (Supplementary Section D). The Benjamini-Hochberg (BH) procedure \citep{benjamini1995controlling} is used to adjust the p-values. We choose the threshold $n^{0.35}$ because the EM-test has a good balance between retaining cluster-relevant features and excluding cluster-irrelevant features at these cutoffs (Supplementary Section D). For the Chi-square goodness-of-fit test, KS-test and Dip-test, we use the BH-adjusted p-values ($<0.01$) to screen the cluster-relevant features. 

Table \ref{tab:PN} shows the numbers of correctly retained and falsely retained features by different methods. Similarly, we find that the EM-test methods (EM-adjust,  EM-0.35) outperform other methods in most settings especially when $p$ is larger and different clusters are more similar to each other. In most cases, Skmeans is able to select all cluster-relevant features, but also falsely select many cluster-irrelevant features. SC-FS is conservative. In low dimensional settings ($p=500$), SC-FS could correctly select all cluster-relevant features with almost no false positives. However, in higher dimensions, SC-FS also has almost zero false positives, but its power is low. For example, in the $p=5000$, medium signal and high noise case, SC-FS only reports 2 cluster-relevant features. In the same case, EM-adjust reports all 20 features with almost zero false positives. The performance of two versions of EM-test are slightly different. EM-adjust is more conservative than  EM-0.35. In the more difficult settings (with large $p$ and low signal to noise ratio), EM-adjust is still able to control the false positives, but detects less cluster-relevant features. In most cases, EM-0.35 can detect most cluster-relevant features, but also report some cluster-irrelevant features in the more difficult simulation settings.    KS-test and Dip-test report many false positives and select almost all features as clustering-relevant features. COCSI  is very conservative  in this simulation and could not select any features.

\begin{table}
	\centering
	\caption{The mean of numbers of correctly retained features ($\mathcal{R}$) and falsely retained features ($\mathcal{F}$) by different methods over 100 replications. The numbers in the parenthesis are the standard deviation of $\mathcal{R}$ and  $\mathcal{F}$ over 100 replications.}
	\renewcommand\arraystretch {1.2}
	\footnotesize
	\setlength{\tabcolsep}{1.3mm}
	\begin{tabular}{cccccccccc}
		\hline
		\hline
		$p$	&       & EM-adjust & EM-0.35 & Chi-square & SC-FS & Skmeans & KS-test & Dip-test & COSCI \\
		\hline 
		\multicolumn{10}{c}{Case 1:  High signal  and low noise} \\  
		\hline  
		\multirow{2}[0]{*}{500} & $\mathcal{R}$ & 20 (0.1) & 20 (0.1) & 20 (0.7) & 20 (0.9) & 20 (0.0) & 20.0 (0.0) & 20.0 (0.0) & 0.0 (0.0) \\
		& $\mathcal{F}$ & 0 (0.1) & 1 (1.1) & 2 (1.6) & 0 (0.0) & 480 (0.0) & 480.0 (0.0) & 480.0 (0.0) & 0.0 (0.0) \\
		\multirow{2}[0]{*}{5000} & $\mathcal{R}$ & 20 (0.0) & 20 (0.0) & 19 (1.2) & 16 (4.2) & 19 (0.8) & 20.0 (0.0) & 20.0 (0.0) & 0.0 (0.0) \\
		& $\mathcal{F}$ & 0 (0.1) & 10 (3.3) & 2 (1.4) & 0 (0.0) & 239 (700.9) & 4980.0 (0.0) & 4980.0 (0.0) & 0.0 (0.0) \\
		\multirow{2}[0]{*}{20,000} & $\mathcal{R}$ & 20 (0.0) & 20 (0.0) & 18 (1.3) & 0 (0.4) & NA  & 20.0 (0.0) & 20.0 (0.0) & NA \\
		& $\mathcal{F}$ & 0 (0.2) & 40 (5.6) & 2 (1.7) & 0 (0.0) & NA  & 19980.0 (0.0) & 19980.0 (0.0) & NA \\
		\hline 
		\multicolumn{10}{c}{Case 2:  High signal  and high noise} \\  
		\hline  
		\multirow{2}[0]{*}{500} & $\mathcal{R}$ & 20 (0.1) & 20 (0.0) & 18 (1.2) & 20 (0.7) & 20 (0.0) & 20.0 (0.0) & 20.0 (0.0) & 0.0 (0.0) \\
		& $\mathcal{F}$ & 0 (0.2) & 2 (1.4) & 2 (1.5) & 0 (0.0) & 480 (0.0) & 480.0 (0.0) & 480.0 (0.0) & 0.0 (0.0) \\
		\multirow{2}[0]{*}{5000} & $\mathcal{R}$ & 20 (0.0) & 20 (0.0) & 17 (1.9) & 11 (5.4) & 19 (2.0) & 20.0 (0.0) & 20.0 (0.0) & 0.0 (0.0) \\
		& $\mathcal{F}$ & 0 (0.1) & 18 (4.1) & 2 (1.4) & 0 (0.1) & 367 (739.3) & 4980.0 (0.0) & 4980.0 (0.0) & 0.0 (0.0) \\
		\multirow{2}[0]{*}{20,000} & $\mathcal{R}$ & 20 (0.2) & 20 (0.0) & 15 (1.9) & 0 (0.1) & NA  & 20.0 (0.0) & 20.0 (0.0) & NA \\
		& $\mathcal{F}$ & 0 (0.3) & 75 (7.7) & 2 (1.9) & 0 (0.0) & NA  & 19980.0 (0.0) & 19980.0 (0.0) & NA \\
		\hline 
		\multicolumn{10}{c}{Case 3:  Medium signal  and low noise} \\  
		\hline  
		\multirow{2}[0]{*}{500} & $\mathcal{R}$ & 20 (0.3) & 20 (0.1) & 16 (2.4) & 20 (0.5) & 20 (0.0) & 20.0 (0.0) & 20.0 (0.0) & 0.0 (0.0) \\
		& $\mathcal{F}$ & 0 (0.1) & 1 (1.1) & 2 (1.6) & 0 (0.0) & 480 (0.0) & 480.0 (0.0) & 480.0 (0.0) & 0.0 (0.0) \\
		\multirow{2}[0]{*}{5000} & $\mathcal{R}$ & 20 (0.4) & 20 (0.1) & 12 (2.3) & 8 (5.0) & 19 (3.0) & 20.0 (0.0) & 20.0 (0.0) & 0.0 (0.0) \\
		& $\mathcal{F}$ & 0 (0.1) & 10 (3.2) & 2 (1.4) & 0 (0.0) & 447 (852.1) & 4980.0 (0.0) & 4980.0 (0.0) & 0.0 (0.0) \\
		\multirow{2}[0]{*}{20,000} & $\mathcal{R}$ & 20 (0.6) & 20 (0.0) & 10 (2.5) & 0 (0.0) & NA  & 20.0 (0.0) & 20.0 (0.0) & NA \\
		& $\mathcal{F}$ & 0 (0.2) & 41 (5.7) & 1 (1.3) & 0 (0.0) & NA  & 19980.0 (0.0) & 19980.0 (0.0) & NA \\
		\hline 
		\multicolumn{10}{c}{Case 4:  Medium signal  and high noise} \\  
		\hline  
		\multirow{2}[0]{*}{500} & $\mathcal{R}$ & 20 (0.5) & 20 (0.1) & 13 (2.3) & 20 (0.4) & 20 (0.0) & 20.0 (0.0) & 20.0 (0.0) & 0.0 (0.0) \\
		& $\mathcal{F}$ & 0 (0.1) & 2 (1.4) & 2 (1.4) & 0 (0.0) & 480 (0.0) & 480.0 (0.0) & 480.0 (0.0) & 0.0 (0.0) \\
		\multirow{2}[0]{*}{5000} & $\mathcal{R}$ & 19 (0.9) & 20 (0.1) & 9 (2.5) & 3 (3.2) & 16 (6.5) & 20.0 (0.0) & 20.0 (0.0) & 0.0 (0.0) \\
		& $\mathcal{F}$ & 0 (0.2) & 18 (4.1) & 1 (1.1) & 0 (0.1) & 1004 (1417.9) & 4980.0 (0.0) & 4980.0 (0.0) & 0.0 (0.0) \\
		\multirow{2}[0]{*}{20,000} & $\mathcal{R}$ & 19 (1.2) & 20 (0.2) & 7 (2.7) & 0 (0.0) & NA  & 20.0 (0.0) & 20.0 (0.0) & NA \\
		& $\mathcal{F}$ & 0 (0.3) & 75 (7.5) & 1 (1.3) & 0 (0.0) & NA  & 19980.0 (0.0) & 19980.0 (0.0) & NA \\
		\hline 
		\multicolumn{10}{c}{Case 5:  Low signal  and low noise} \\  
		\hline  
		\multirow{2}[0]{*}{500} & $\mathcal{R}$ & 16 (2.0) & 19 (1.0) & 4 (2.5) & 20 (0.6) & 20 (0.0) & 20.0 (0.0) & 20.0 (0.0) & 0.0 (0.0) \\
		& $\mathcal{F}$ & 0 (0.1) & 1 (1.1) & 1 (1.2) & 0 (0.0) & 470 (67.5) & 480.0 (0.0) & 480.0 (0.0) & 0.0 (0.0) \\
		\multirow{2}[0]{*}{5000} & $\mathcal{R}$ & 14 (2.0) & 19 (1.0) & 2 (1.6) & 1 (1.6) & 11 (8.0) & 20.0 (0.0) & 20.0 (0.0) & 0.0 (0.0) \\
		& $\mathcal{F}$ & 0 (0.1) & 10 (3.3) & 1 (0.7) & 0 (0.0) & 1458 (1811.1) & 4980.0 (0.0) & 4980.0 (0.0) & 0.0 (0.0) \\
		\multirow{2}[0]{*}{20,000} & $\mathcal{R}$ & 12 (2.3) & 19 (1.1) & 1 (1.0) & 0 (0.0) & NA  & 20.0 (0.0) & 20.0 (0.0) & NA \\
		& $\mathcal{F}$ & 0 (0.1) & 41 (5.6) & 0 (0.6) & 0 (0.0) & NA  & 19980.0 (0.0) & 19980.0 (0.0) & NA \\
		\hline 
		\multicolumn{10}{c}{Case 6:  Low signal  and high noise} \\  
		\hline  
		\multirow{2}[0]{*}{500} & $\mathcal{R}$ & 15 (2.1) & 18 (1.2) & 3 (2.0) & 20 (0.6) & 20 (0.0) & 20.0 (0.0) & 20.0 (0.0) & 0.0 (0.0) \\
		& $\mathcal{F}$ & 0 (0.1) & 2 (1.4) & 1 (1.0) & 0 (0.1) & 480 (0.0) & 480.0 (0.0) & 480.0 (0.0) & 0.0 (0.0) \\
		\multirow{2}[0]{*}{5000} & $\mathcal{R}$ & 12 (2.3) & 18 (1.3) & 1 (1.1) & 0 (0.3) & 10 (8.7) & 20.0 (0.0) & 20.0 (0.0) & 0.0 (0.0) \\
		& $\mathcal{F}$ & 0 (0.1) & 18 (4.2) & 0 (0.4) & 0 (0.0) & 2150 (2158.9) & 4980.0 (0.0) & 4980.0 (0.0) & 0.0 (0.0) \\
		\multirow{2}[0]{*}{20,000} & $\mathcal{R}$ & 10 (2.5) & 18 (1.4) & 1 (1.0) & 0 (0.0) & NA  & 20.0 (0.0) & 20.0 (0.0) & NA \\
		& $\mathcal{F}$ & 0 (0.3) & 75 (7.4) & 0 (0.7) & 0 (0.0) & NA  & 19980.0 (0.0) & 19980.0 (0.0) & NA \\
		\hline
	\end{tabular}%
	\label{tab:PN}%
\end{table}%

\subsection{Feature screening improves clustering analysis}\label{subsec:clus}
In this subsection, we assess the influence of feature screening on clustering analyses. For each simulation, we first use the feature screening methods to select potential cluster-relevant features and then use the k-means algorithm to cluster the samples.   For the feature selection method Skmeans, we directly use its clustering results.  The number of clusters in the k-means and Skmeans algorithms is set as 5. The parameters of the feature screening/selection methods are set as in Section \ref{subsec:FS}. For comparison, we also include k-means clustering results using all features (called No-Screening) and the oracle clustering results using only the cluster-relevant features. 

Table \ref{tab:ari} shows the adjusted Rand index (ARI) between the clustering results given by different methods and the true clusters. Generally speaking, methods that can accurately select more cluster-relevant features while excluding more cluster-irrelevant features (Table \ref{tab:PN}) tend to perform better in the clustering. All count-data feature screening methods or the feature selection method  can help to improve, often by a large amount, the clustering accuracy in comparison with the baseline method No-Screening, indicating that feature screening is an essential step for clustering analysis of high dimensional data. The two versions of the EM-test method have similar performances and consistently perform better than other methods. When the dimension is small ($p=500$) or the difference between clusters is large (high signal and low noise), EM-test, Chi-square, SC-FS and Skmeans have similar performance.  When the difference between clusters is smaller and the dimension $p$ is larger, the advantage of the EM-test over other methods is more apparent. In addition, clustering based on the EM-test screening can achieve an accuracy similar to that of the oracle clustering in most settings.   The performance of  KS-test and Dip test  are similar to  No-Screening because they select almost all features. The ARIs of COSCI are zero because COSCI could not select any features for count data.  

\begin{table}
	\centering
	\caption{The means and standard deviations (in parenthesis) of ARIs over 100 replications. The values in the table are shown as the actual values $\times ~ 100$.}
	\renewcommand\arraystretch {1.2}
	\footnotesize
	\setlength{\tabcolsep}{0.9mm}
	\begin{tabular}{ccccccccccc}
		\hline
		\hline
		$p$	& No-Screening & Oracle & EM-pvalue & EM-0.35 & Chi-square & SC-FS & Skmeans & KS-test & Dip-test & COSCI \\
		\hline 
		\multicolumn{11}{c}{Case 1:  High signal  and low noise} \\  
		\hline  
		$500$ & 94 (11.9) & 98 (1.1) & 98 (1.5) & 98 (2.7) & 98 (2.9) & 98 (4.1) & 98 (0.8) & 94 (1.4) & 94 (1.4) & 0 (0.0) \\
		$5000$ & 10 (3.6) & 98 (2.8) & 98 (0.9) & 98 (1.5) & 97 (1.4) & 95 (22.7) & 96 (21.5) & 12 (3.2) & 12 (3.1) & 0 (0.0) \\
		$20,000$ & 0 (0.3) & 98 (4.5) & 98 (2.9) & 98 (1.3) & 97 (3.0) & 0 (1.1) & NA  & 1 (0.3) & 1 (0.3) & NA \\
		\hline 
		\multicolumn{11}{c}{Case 2:  High signal  and high noise} \\  
		\hline  
		$500$ & 88 (11.2) & 94 (2.5) & 94 (2.8) & 94 (3.0) & 93 (2.1) & 94 (4.7) & 94 (1.4) & 88 (2.5) & 88 (2.5) & 0 (0.0) \\
		$5000$ & 4 (2.0) & 94 (1.7) & 94 (1.7) & 94 (1.7) & 91 (3.9) & 57 (28.5) & 49 (26.3) & 6 (2.0) & 6 (1.9) & 0 (0.0) \\
		$20,000$ & 0 (0.2) & 94 (2.9) & 94 (1.9) & 93 (1.5) & 89 (4.2) & 0 (1.0) & NA  & 1 (0.2) & 1 (0.2) & NA \\
		\hline 
		\multicolumn{11}{c}{Case 3:  Medium signal  and low noise} \\  
		\hline  
		$500$ & 88 (6.8) & 96 (1.7) & 96 (1.8) & 96 (1.9) & 92 (5.1) & 96 (3.0) & 96 (1.3) & 88 (2.5) & 88 (2.5) & 0 (0.0) \\
		$5000$ & 3 (1.7) & 96 (1.1) & 96 (1.2) & 96 (1.1) & 86 (8.0) & 33 (27.7) & 30 (24.4) & 5 (1.7) & 5 (1.8) & 0 (0.0) \\
		$20,000$ & 0 (0.2) & 96 (3.2) & 96 (1.9) & 95 (1.3) & 80 (11.4) & 0 (0.0) & NA  & 1 (0.2) & 1 (0.2) & NA \\
		\hline 
		\multicolumn{11}{c}{Case 4:  Medium signal  and high noise} \\  
		\hline  
		$500$ & 78 (5.5) & 90 (1.8) & 90 (2.0) & 90 (1.7) & 80 (7.4) & 90 (2.1) & 90 (1.7) & 77 (4.0) & 77 (4.1) & 0 (0.0) \\
		$5000$ & 2 (1.0) & 91 (2.1) & 90 (2.5) & 90 (2.1) & 68 (13.0) & 12 (14.5) & 13 (16.2) & 3 (1.1) & 3 (1.0) & 0 (0.0) \\
		$20,000$ & 0 (0.2) & 90 (1.9) & 89 (3.0) & 89 (1.8) & 54 (17.2) & 0 (0.0) & NA  & 0 (0.2) & 0 (0.2) & NA \\
		\hline 
		\multicolumn{11}{c}{Case 5:  Low signal  and low noise} \\  
		\hline  
		$500$ & 59 (11.1) & 90 (1.8) & 85 (7.6) & 89 (2.6) & 32 (19.4) & 90 (2.8) & 89 (2.1) & 60 (9.6) & 60 (9.7) & 0 (0.0) \\
		$5000$ & 1 (0.7) & 90 (2.1) & 79 (7.1) & 88 (2.6) & 12 (13.7) & 0 (8.9) & 9 (9.9) & 2 (0.6) & 2 (0.7) & 0 (0.0) \\
		$20,000$ & 0 (0.2) & 89 (1.8) & 75 (8.7) & 87 (3.0) & 0 (8.9) & 0 (0.0) & NA  & 0 (0.2) & 0 (0.2) & NA \\
		\hline 
		\multicolumn{11}{c}{Case 6:  Low signal  and high noise} \\  
		\hline  
		$500$ & 38 (8.7) & 82 (2.5) & 73 (6.5) & 80 (3.6) & 17 (14.9) & 82 (3.1) & 81 (3.1) & 42 (7.7) & 42 (7.7) & 0 (0.0) \\
		$5000$ & 1 (0.4) & 83 (2.4) & 67 (8.3) & 79 (3.9) & 10 (8.9) & 0 (1.8) & 0 (6.5) & 1 (0.4) & 1 (0.4) & 0 (0.0) \\
		$20,000$ & 0 (0.2) & 82 (2.6) & 60 (11.7) & 75 (8.9) & 0 (7.6) & 0 (0.0) & NA  & 0 (0.2) & 0 (0.2) & NA \\
		\hline
	\end{tabular}%
	\label{tab:ari}%
\end{table}%

We also compare the computational time of the feature screening/selection methods (Supplementary Section D). SC-FS, KS-test and Dip-test  are the computationally most efficient method. The EM-test is also computationally efficient and can allow analyzing tens of thousands of features using a typical desktop computer. Skmeans is computationally very demanding, partly because it has to select the best tuning parameter using  permutation.

%\begin{table}[htbp]
%	\centering
%	\caption{The computation time of different methods.}
%		\renewcommand\arraystretch {1.2}
%	\footnotesize
%	\setlength{\tabcolsep}{4mm}
%	\begin{tabular}{cccccccc}
%				\hline
%		\hline
%		Time (s) & EM-test & Chi-square & SC-FS & Skmeans & KS-test & Dip-test & COSCI \\
%		\hline
%		$p=500$   & 14.59  & 25.57  & 0.98  & 265.36  & 1.83  & 1.86  & 53.73  \\
%		$p=5000$  & 127.21  & 246.97  & 7.60  & 3074.06  & 7.01  & 7.09  & 470.50  \\
%		$p=20,000$ & 505.21  & 986.82  & 30.53  & NA    & 25.90  & 25.84  & NA \\
%		\hline
%	\end{tabular}%
%	\label{tab:time}%
%\end{table}%

%\begin{table}
%	\centering
%	\caption{The computation time of different methods.}
%	\renewcommand\arraystretch {1.2}
%\footnotesize
%\setlength{\tabcolsep}{8mm}
%	\begin{tabular}{ccccc}
%		\hline
%		\hline
%		Time (s) & EM-test & Chi-square & SC-FS & Skmeans \\
%		\hline
%		$p=500$ & 8.71  & 12.4 & 1.02  & 265 \\
%		$p=5000$ & 72.17 & 115.42 & 7.51  & 3074.06 \\
%		$p=20,000$ & 283.94  &  462.3 & 30.56 & NA \\
%		\hline
%	\end{tabular}%
%	\label{tab:time}%
%\end{table}%

\section{Application on scRNA-seq data}\label{sec:app}
In this section, we consider an application to the scRNA-seq data from \citet{Heming2021NeurologicalMO}. The scRNA-seq data contain single cells from 31 patients, including eight patients of coronavirus disease 2019 with acute or long-term neurological sequelae (Neuro-COVID), five viral encephalitis (VE) patients, nine multiple sclerosis (MS) patients and nine idiopathic intracranial hypertension (IIH) patients. Here, we focus on monocytes, granulocytes and dendritic cells. After quality control, in total, we have 11,697 cells and 33,538 candidate genes. The scRNA-seq data are usually modeled by the negative binomial distribution\citep{Chen2018UMIcountMA}. We thus apply the EM-test of the negative binomial distribution to screen for important genes. At the FDR threshold 0.01, the EM-test selects 2754 genes. With these genes, we perform clustering and annotation analysis and identify 9 cell subtypes. Details of feature screening, clustering and annotation are shown in the Supplementary material.

We also apply the Chi-square test and the KS-test and use their selected genes to cluster the single cells. The Chi-square test is the goodness-of-fit test of the negative binomial distribution. At the FDR threshold 0.01, it reports 158 genes. The KS-test is applied to the normalized data using a normalization procedure that are commonly used in scRNA-seq data analyses \citep{Butler2018IntegratingST}. The normalization can make the data better approximated by normal distributions. Following IF-PCA \citep{jin2016influential}, the threshold of the KS-test is chosen as the higher-criticism threshold, which gives 15,732 important genes. For comparison, we also consider the baseline No-Screening method (all genes are included). Then, we use the same clustering procedures to cluster the single cells using the Chi-square or KS-test selected genes.

We first evaluate the clustering results using the Calinski-Harabasz index \citep{Caliski1974ADM} and Silhouette index \citep{rousseeuw1987silhouettes}. More specifically, we perform dimension reduction with the genes selected by different methods. Then we calculate the Calinski-Harabasz and Silhouette index of the clustering results given by different methods in their respective lower dimension spaces. Principle component analysis (PCA) and uniform manifold approximation and projection (UMAP) \citep{McInnes2018UMAPUM} are used for dimension reduction (to 40 dimensions for PCA and 2 dimensions for UMAP). As shown in Table \ref{table::si}, the EM-test has the largest Calinski-Harabasz and Silhouette indexes, indicating that genes selected by EM-test provide the most distinct clustering results on the reduced feature spaces. Also, we can see that the EM-test is the only method that has the Calinski-Harabasz and Silhouette larger than the No-Screening method, indicating that the EM-test is more effective in selecting cluster-relevant features.

\begin{table}
	\centering
	\caption{The Silhouette and Calinski-Harabasz index of the clustering results based on genes selected by different methods on their respective lower dimensional spaces.}
	\renewcommand\arraystretch {1.2}
\footnotesize
\setlength{\tabcolsep}{6mm}
	\begin{tabular}{ccccc}
		\hline
		\hline
		Method                      & EM-test    & Chi-square & KS-test     & No-Screening \\ \hline
		Number of selected features & 2754 & 158      & 15,732 & 33,538               \\ \hline
		\multicolumn{5}{c}{Dimension Reduction by UMAP}                                                       \\ \hline
		Silhouette index            & 0.27  & 0.04     & 0.21   & 0.23                 \\
		Calinski-Harabasz index     & 8713 & 2981    & 7674  & 7895                \\ \hline
		\multicolumn{5}{c}{Dimension Reduction by PCA}                                                        \\ \hline
		Silhouette index            & 0.11  & 0.02     & 0.08   & 0.10                 \\
		Calinski-Harabasz index     & 1048 & 326      & 1033  & 988                  \\ \hline
	\end{tabular}
	\label{table::si}
\end{table}

Clustering of the EM-test selected genes gives 9 single cell clusters, including the 6 cell subtypes reported in \citep{Heming2021NeurologicalMO}. The additional three cell subtypes are two subtypes of monocytes, which we named as mono\_IFN monocyte and IL7$\textrm{R}^+$ monocyte1 and IL7$\textrm{R}^+$ monocyte2 (Fig.\ref{fig:pbmc2figure}). The two clusters of IL7$\textrm{R}^+$ monocytes are often observed at inflammation sites \citep{AlMossawi2019ContextspecificRO}. The mono\_IFN monocytes highly express many interferon-related genes (Fig.\ref{fig:pbmc2figure}E), suggesting  that these cells might play important roles in immune responses to viral infection \citep{Heming2021NeurologicalMO}. Most of these mono\_IFN monocytes are from VE patients (89\%), and are depleted in Neuro-COVID patients (1\%) compared with VE patients (15\%) (Fig.\ref{fig:pbmc2figure}F). These indicate that there might be an attenuated interferon response in Neuro-COVID patients. This attenuated interferon response in Neuro-COVID patients was discovered in \citet{Heming2021NeurologicalMO} by differential gene expression analysis. However, \citet{Heming2021NeurologicalMO} did not find the mono\_IFN monocyte possibly because of its feature screening. Here, we successfully identify the mono\_IFN monocyte and its marker genes, which can facilitate further downstream analysis of {these types of cells}. All other methods cannot detect these mono\_IFN monocytes. These mono\_IFN monocytes either scatter widely in the method's UMAP plot or are only a small portion of larger cell clusters detected by other methods  (Figure \ref{fig:pbmc2figure} B-D). Therefore, we conclude that the EM-test enables more accurate cell type identification via its precise cluster-relevant gene screening and lead to the discovery of the potential novel cell subtype mono\_IFN monocyte.

\begin{figure}[htb]
	\centering
	\includegraphics[width=\linewidth]{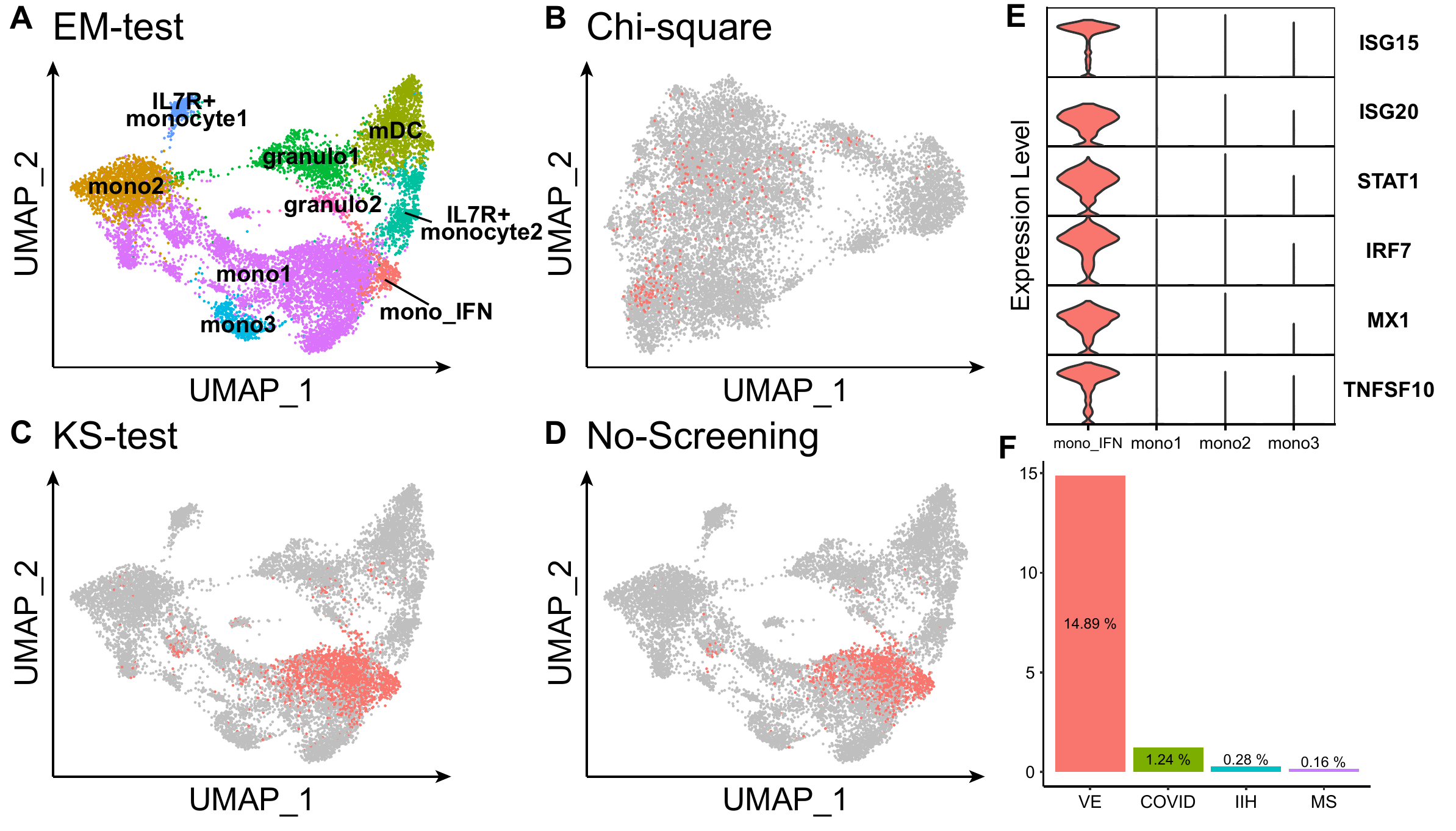}
	\caption{Analysis of cerebrospinal fluid (CSF) samples from  Neuro-COVID, viral encephalitis, multiple sclerosis and nine idiopathic intracranial hypertension patients. (A) Clustering and UMAP based on genes selected by the EM-test. (B) The location of mono\_IFN cells on the UMAP derived from genes selected by Chi-square. (C) The location of the cluster derived from genes selected by KS-test containing the mono\_IFN cells on the UMAP plot by EM-test selected genes. (D) The location of the cluster derived from all candidate genes containing the mono\_IFN cells on the UMAP plot by EM-test selected genes. (E) Expression of several interferon-related genes markers of different cell types.  (F) Percentages of mono\_IFN cells in VE, Neuro-COVID, IIH and MS patients. }
	\label{fig:pbmc2figure}
\end{figure}

\section{Discussion}\label{sec:diss}
%Screening for important features is a ubiquitous problem in high dimensional data analyses. Many innovative supervised feature screening methods have been proposed and deep statistical theories have been developed. In contrast, partly because of the lack of the response variables, feature screening of clustering analysis is much less studied. 
In this paper, we propose a general parametric clustering feature screening method using the EM-test. We establish the tail probability bounds of the EM-test statistic and show that the proposed screening method can achieve the sure independent screening property and consistency in feature selection when $p$ goes to infinity not too fast. Limiting distribution of the EM-test statistic under general settings is also obtained.
Conditions in this paper are generally mild and many commonly used parametric families satisfy all these conditions. Thus, our method can be widely applied. The most stringent condition is the strong identifiability condition (\Cd). Although many exponential family distributions satisfy this condition, normal distributions with unknown means and variances cannot satisfy this condition (but normal distributions with known variances can). However, we find that this problem is closely related to a well-known truncated moment problem and we actually can establish the tail probability bound for normal distributions without Condition (\Cd). This is out of the scope of this paper and we will discuss in future research.

 One limitation of the proposed method is that EM-test is a marginal screening method. Jointly important features may be marginally unimportant and thus could be missed by EM-test. This problem will not occur if the features are independent. In clustering analysis, this problem can also be avoided under conditions other than independence. For example, clustering methods like k-means rely on variables' means for clustering analysis. If clustering-relevant features are assumed to have different means in different clusters, a scenario considered in \citet{Cai2019CHIMECO} and many other clustering works \citep{jin2016influential, Lffler2019OptimalityOS}, jointly important clustering-relevant features will always be marginally important and the problem will not occur. 
On the other hand, marginally important features may be jointly unimportant and and could be falsely retained by marginal screening methods like EM-test. However, if most of important features are retained, inclusion of a few false positives may not have significant impact on clustering accuracy. For example,  for the simulation scenario with low signal and low noise and $p=20,000$, EM-0.35 retains almost all of 20 important features, but also report around 40 false positives. In comparison, EM-adjust has almost no false positives but only retains around 12 important features (Table 2). However, in terms of clustering accuracy, the ARI of EM-0.35 is considerably higher than EM-adjust (0.87 versus 0.74).

The current method can be improved in several aspects. One important type of data is binary data. Since a mixture of binary distributions is still a binary distribution, the current method is not able to screen for cluster-relevant binary data. Further studies on feature screening for binary data are needed. A potential way to address this problem is to first aggregate the binary variables and then perform screening on the aggregated variables. Another important direction to improve over the current methods is to develop non-parametric or semi-parametric screening methods for clustering analyses. Non-parametric or semi-parametric screening methods can allow more robust feature screening and thus potentially have wider applications.

\section{Supplementary material}
All proofs of the theoretical results are given in the Supplementary material. 
Additional simulation results and details for the application are also shown in the Supplementary material.

\section{Acknowledgments}
\label{sec:ackn}
This work was supported by the National Key Basic Research Project of China (2020YFE0204000), the National Natural Science Foundation of China (11971039), and Sino-Russian Mathematics Center.

\appendix
	\section{Proofs of the non-asymptotic results}\label{appnonasymptotic}
In this section, we aim to prove Theorem \thmmain--\thmoneRes.  
\subsection{Sketch of the proof of Theorem \thmzeroRes} \label{sec:sketch}
In this section, we sketch the proof of Theorem \thmzeroRes. We first recall some important definitions   in the manuscript. 
Let 
\begin{eqnarray}\label{eq:b}
	Y_{ih} = \frac{1}{f(x_i,  \btheta_0)} \frac{\partial f(x_i, \btheta_0)}{\partial \theta_h}, ~ Z_{ih} = \frac{1}{2f(x_i,  \btheta_0)}\frac{\partial^2 f(x_i,  \btheta_0)}{\partial \theta_h^2},\nonumber \\ 
	U_{ih\ell} = \frac{1}{f(x_i,  \btheta_0)}\frac{\partial^2 f(x_i,  \btheta_0)}{\partial \theta_h \partial \theta_\ell} (h< \ell),~\b_{1i} = (Y_{i1}, \ldots, Y_{id})\trans, \\
	\b_{2i}= (Z_{i1},\ldots,Z_{id},U_{i12},\ldots,U_{i(d-1)d})\trans, ~\mbox{and}~\b_i = \left(\b_{1i}\trans,\b_{2i}\trans \right)\trans. \nonumber
\end{eqnarray}
Given $h,\ell \in \{1,\dots,d\}$, let 
$${m}_{1h}(\ba, \bxi_1, \bxi_2) = \sum_{g=1}^G {\alpha}_g\left({\theta}_{1gh} - {\theta}_{2gh}\right),~ {m}_{2h}(\ba, \bxi_1, \bxi_2) =  \sum_{g=1}^G {\alpha}_g\left({\theta}_{1gh} - {\theta}_{2gh}\right)^2,$$
$$\mbox{ and }{s}_{h\ell}(\ba, \bxi_1, \bxi_2) =\sum_{g=1}^G {\alpha}_g\left({\theta}_{1gh} - {\theta}_{2gh}\right)\left({\theta}_{1g\ell} - {\theta}_{2g\ell}\right) ~(h<\ell),$$
where $\bxi_j=(\btheta_{j1},\dots,\btheta_{jG})$ and $\btheta_{jg}=(\theta_{jg1},\dots,\theta_{jgd})$ ($j=1,2$, $g=1,\dots,G$).   
We define two vector functions as  
\begin{equation}\label{eq:m1}
	{\m}_1(\ba, \bxi_1, \bxi_2) = ({m}_{11}, \dots, {m}_{1d}) \trans, 
\end{equation}
\begin{equation}\label{eq:m}
	{\m}(\ba, \bxi_1, \bxi_2) = ({m}_{11}, \dots, {m}_{1d},{m}_{21},  \dots, {m}_{2d}, {s}_{12},  \dots, {s}_{(d-1)d} ) \trans, 
\end{equation}
and simplify ${\m}_1(\ba, \bxi, \bxi_0)$ and ${\m}(\ba, \bxi, \bxi_0)$ as $\m_1(\ba, \bxi)$ and $\m(\ba, \bxi)$, respectively.

The basic idea of the proof is to alternatively bound $\ba^{(k)}$ and ${\bxi}^{(k)}$ ($k=0,\dots,K$), and use Taylor's expansion to bound the EM-test statistic. More specifically, given an initial value $\ba^{(0)}$, the one step EM update ${\bxi}^{(0)}$ maximizes the log-likelihood $l_n(\bxi, \ba^{(0)}) = \sum_{i=1}^{n} \log \ \varphi(x_i;\bxi, \ba^{(0)})$. Observe that the homogeneous distribution $f(x;\btheta_0)$ can also be written as $\varphi(x;\bxi_0, \ba^{(0)})$, and all elements of $\ba^{(0)}$ are bounded away from zero, i.e. $ \min_{g=1,\ldots,G} {\alpha}_g^{(0)} > \delta>0$. The one step update ${\bxi}^{(0)}$ will be a consistent estimate of the true parameter $\bxi_0$, and we can bound the tail probability of $\left\|{\bxi}^{(0)}-\bxi_0\right\|_2^2$. Alternatively, since ${\bxi}^{(0)}$ is close to $\bxi_0$, the EM update $\ba^{(1)}$ will also be bounded away from zero, i.e. $ \min_{g=1,\ldots,G} {\alpha}_g^{(1)} > \delta$. We can repeat this process $K$ times and give a tail probability bound for $\left\|{\bxi}^{(K)}-\bxi_0\right\|_2^2$. Finally, we can use Taylor's expansion to represent the EM-test statistic in terms of ${\bxi}^{(K)}-\bxi_0$, and obtain the tail probability bound for the EM-test statistic. In the following, we present the critical lemmas needed in this proof process.

The following Lemma \ref{lem:tailEudis} guarantees that if $\ba^{(k)}$ is bounded away from zero, the EM update ${\bxi}^{(k)}$ will be close to $\bxi_0$ and we can obtain a tail probability bound for $\left\|{\bxi}^{(k)}-\bxi_0\right\|_2^2$. More specifically, we define
\begin{equation}\label{eq:alphaA}
	\mathbb{S}_{\delta} = \left\{ \balpha: \ba \in \mathbb{S}^{G-1}, \min_{g=1,\dots,G} \alpha_g \geq\delta >0 \right\}.
\end{equation}
Clearly, we have $\ba^{(0)} \in \mathbb{S}_{\delta}$. In the proof process of Theorem \thmzeroRes, we can show that $\ba^{(k)}\in \mathbb{S}_{\delta}$ with a high probability. Denote 
\begin{equation}\label{eq:calS}
	\mathcal{S}^{(k)}_{\epsilon} =  \left\{\left\|{\bf {m}}\left(\ba^{(k)}, {\bxi}^{(k)}\right)\right\|_2 < \frac{\epsilon}{L_1}, \left\|{\bxi}^{(k)}-\bxi_0\right\|_2^2 < \frac{\epsilon}{L_2}\right\},
\end{equation}
where $L_1,L_2 > 0$ are two constants that will be specified in Lemma \ref{lem:lower} and $ {\bf {m}}\left(\ba^{(k)}, {\bxi}^{(k)}\right)$ is as defined in (\ref{eq:m}). We have the following lemma. 

\begin{lemma} \label{lem:tailEudis}
	Assume that $x_1,\dots,x_n$ are independent samples from the homogeneous distribution $f(x;\btheta_0)$. 
	Let $c_1 = 1/24, c_2 = (4/27)(1/1926), p_0 =  \lambda G \log(\delta G)$ and $c>0$ be a constant depending on $\delta, d,G, M$ and ${\rm diam}(\Xi)$.  
	Under Condition (\Ca)--(\Cd),
	for any  $\epsilon>0$ and sufficiently large $n$ such that 
	$\epsilon \geq \max \left(c n^{-1/2} {\rm log} n, c_1^{-1/2} \left({-p_0} \right)^{1/2} n^{-1/2}\right),$
	we have 
	$$\Pr\left( \mathcal{S}^{(k)}_{\epsilon} \cap \left\{ \ba^{(k)} \in\mathbb{S}_{\delta}\right\} \right) \geq \Pr \left(\ba^{(k)} \in \mathbb{S}_{\delta}\right) -5\exp\left(-c_2 n \epsilon^2 \right) .$$
\end{lemma}
Define 
\begin{equation}\label{eq:calE}
	\mathcal{E}^{(k+1)} = \left\{\min_{g=1,\dots,G} {\alpha}_g^{(k+1)} \geq \min_{g=1,\dots,G}  {\alpha}_g^{(k)}\left(1-\frac{2}{K}\right) \right\}.
\end{equation}
The following lemma shows that if ${\bxi}^{(k)}$ is close to the true value $\bxi_0$, $\min_{g} {\alpha}_g^{(k+1)} $ can be bounded by $\min_{g} {\alpha}_g^{(k)} $ up to a fixed factor with a high probability. Let $\Delta_K$ be the constant defined in Lemma \ref{lem:deltaK}.

\begin{lemma} \label{lem:alpha}
	Assume that $x_1,\dots,x_n$ are independent samples from the homogeneous distribution $f(x;\btheta_0)$. 
	When $L_2\Delta_K^2 \geq \epsilon$, for any measurable set $\mathcal{B}$, we have 
	$$
	\Pr\left(\mathcal{E}^{(k+1)} \cap\mathcal{S}^{(k)}_{\epsilon}  \cap \mathcal{B} \right) \geq 	\Pr\left(\mathcal{S}^{(k)}_{\epsilon}  \cap \mathcal{B} \right)-2\exp\left(\frac{-2 n}{K^2} \right).
	$$
\end{lemma}
{\colblue Theorem 2 aims to give an upper bound of
	the EM-test statistic under  $\mathbb{H}_0$. By the likelihood non-decreasing property of the EM algorithm,
	if ${\rm EM}_n^{(K)} (K \geq 3)$ is bounded, then ${\rm EM}_n^{(K)} (K < 3)$ is bounded. Thus, without loss of generality, we can assume that $K\geq3$. In other words, the assumption $K\geq3$ in our manuscript can be relaxed to $K>0$ clearly}.  Since $\ba^{(0)} \in \mathbb{S}_{\delta}$, we can alternatively apply Lemma \ref{lem:tailEudis} and Lemma \ref{lem:alpha}, and get 
\begin{equation}\label{eq:Kgeq3}
	\min_{g=1,\ldots,G} {\alpha}_g^{(K)} \geq  \min_{g=1,\ldots,G}  {\alpha}_g^{(0)}\left(1-\frac{2}{K}\right)^K\geq 27^{-1} \min_{g=1,\ldots,G}  {\alpha}_g^{(0)} \geq \delta ~(K\geq 3),
\end{equation}
with a high probability. Applying Lemma \ref{lem:tailEudis} again, we obtain the tail probability bound for $\left\|{\bxi}^{(K)}-\bxi_0\right\|_2$ and $\left\|{\bf {m}}\left(\ba^{(K)}, {\bxi}^{(K)}\right)\right\|_2 $. Combining these results, we can simultaneously bound $\left\|{\bxi}^{(K)}-\bxi_0\right\|_2$ and $\min_{g=1,\ldots,G} {\alpha}_g^{(K)}$. 
\begin{lemma}\label{lem:tail2} Assume that $x_1,\dots,x_n$ are independent samples from the homogeneous distribution $f(x;\btheta_0)$.  Let $c_1 = 1/24, c_2 = (4/27)(1/1926)$, $p_0 =  \lambda G \log(\delta G)$ and $c>0$ be a constant depending on $\delta, d,G, M$ and ${\rm diam}(\Xi)$.  
	Under Condition (\Ca)--(\Cd),
	for any  $\epsilon>0$ and sufficiently large $n$ such that
	$L_2\Delta_K^2 \geq \epsilon \geq \max \left(c n^{-1/2} {\rm log} n, c_1^{-1/2} \left({-p_0} \right)^{1/2} n^{-1/2}\right),$
	we have 
	$$\Pr\left(\mathcal{S}_{\epsilon}^{(K)} \cap \left\lbrace \ba^{(K)} \in \mathbb{S}_{\delta} \right\rbrace  \right) \geq 1-5(K+1)\exp\left(-c_2 n \epsilon^2 \right)-2K\exp\left(\frac{-2 n}{K^2} \right).$$
\end{lemma}

Lemma \ref{lem:tail2}  shows that  when $n$ is sufficiently large, the tail probability of $\bxi^{(K)}$ away from $\bxi_0$ exponentially decays to zero. 
Besides, the convergence rate of $\bxi^{(K)}$  is $ O_p\left(n^{-1/4}\log^{1/2} n\right) $
Based on this result, we can prove Theorem \thmzeroRes.

It is difficult to directly prove Lemma \ref{lem:tailEudis} and bound the Euclidean distance between  $\bxi^{(k)}$ and $\bxi_0$, because the Fisher information matrix is not positive definite under the homogeneous model $\mathbb{H}_0$. However, results in \cite{wong1995probability} imply that the Hellinger distance between $\bxi^{(k)}$ and $\bxi_0$ can be bounded with a high probability. The following Lemma \ref{lem:lower} shows that the Euclidean distance between  $\bxi^{(k)}$ and $\bxi_0$ is dominated by their Hellinger distance. Thus, to bound $\left\|{\bxi}^{(k)}-\bxi_0\right\|_2$ and prove Lemma \ref{lem:tailEudis}, it suffices to bound the Hellinger distance between $\bxi^{(k)}$ and $\bxi_0$. Before presenting Lemma \ref{lem:lower}, we introduce some notations.

Define
$$\mathcal{P}^G_{\mathbb{S}_{\delta}} = \left\{ \sum_{g=1}^{G}\alpha_gf(x;\btheta_g): \btheta_g \in \Theta, (g=1,\dots,G), \balpha 
\in \mathbb{S}_{\delta} \right\}.$$ 
For any two densities $p_1,p_2$ with respect to a measure $\mu$, we define their Hellinger distance as 
$$H(p_1, p_2) = \left\{2^{-1} \int
\left(p_1^{1/2} - p_2^{1/2} \right)^2 {\rd}\mu \right\}^{1/2} .  $$
When $\varphi_1, \varphi_2 \in \mathcal{P}^G_{\mathbb{S}_{\delta}}$, we use $(\bxi_1, \ba_1),(\bxi_1, \ba_2)$ to represent $\varphi_1, \varphi_2$, respectively, and write their Hellinger distance as 
$$H(\ba_1, \ba_2, \bxi_1, \bxi_2) =  \left[ 2^{-1} \int \left\lbrace   \varphi^{1/2}(x; \bxi_1, \ba_1)   - \varphi^{1/2}(x; \bxi_2, \ba_2)\right\rbrace^2   \mu(\rd x) \right] ^{1/2}.$$
When $\bxi_0 = (\btheta_{0}, \ldots, \btheta_{0})$, the Hellinger distance $H(\ba_1,\ba_2, \bxi_1, \bxi_0)$ can be written as 
$$H(\ba_1,\ba_2, \bxi_1, \bxi_0) =  \left[ 2^{-1} \int \left\lbrace   \varphi^{1/2}(x; \bxi_1, \ba_1)   - f^{1/2}(x; \btheta_{0})\right\rbrace^2    \mu(\rd x) \right] ^{1/2}.$$
Note that $ H(\ba_1,\ba_2, \bxi_1, \bxi_0)$ is independent of $\ba_2$ and we write it $ H(\ba_1, \bxi_1, \bxi_0)$.
Let 
$  {\rm diam}_{\m}(\Xi)  =  \sup_{\substack{\bxi_1, \bxi_2 \in \Xi \\ \ba \in \mathbb{S}_{\delta}}} \left\| {\m}(\ba, \bxi_1, \bxi_2) \right\|_2^2.$
Since $\Xi$ is a compact set, we have  $ {\rm diam}_{\m}(\Xi) < \infty$.   
For any ${\delta}' > 0$, 	let 
\begin{equation}\label{eq:Dkappa}
	\mathcal{D}({\delta}') = \left\{  (\ba,\btheta_{0}, \bxi) : \ba \in \mathbb{S}_{\delta},   \btheta_0 \in \Theta, \bxi \in \Xi ,  \sum_{g=1}^G \Vert\btheta_g -\btheta_0 \Vert_2  \geq {\delta'} \right\}.
\end{equation}
We have the following lemma that provides the connection between the Hellinger distance and Euclidean distance.

\begin{lemma} \label{lem:lower}
	Under Condition (\Ca)--(\Cd) and $\mathbb{H}_0$, there exists  $\Delta_1 > 0$  such that for any $\btheta_0\in \Theta, \ba \in \mathbb{S}_{\delta}$, when $\sum_{g=1}^G \Vert\btheta_g -\btheta_0 \Vert_2 < \Delta_1 $, we have 
	$$ H(\ba, \bxi, \bxi_0) \geq 32^{-1}\lambda_{\rm min} \Vert\m(\ba,\bxi)\Vert_2, 
	$$
	where $\bxi=(\btheta_1,\cdots,\btheta_G)$ and $\bxi_0=(\btheta_0,\cdots,\btheta_0)$. Furthermore,  if
	$\omega = \inf _{	\mathcal{D}(\Delta_1)}   H(\ba, \bxi, \bxi_0) > 0,$
	then for any $\btheta_0\in \Theta, \bxi\in \Xi, \ba \in \mathbb{S}_{\delta}$, we have 
	$$ H(\ba, \bxi, \bxi_0) \geq   L_1 \Vert\m(\ba,\bxi)\Vert_2\geq L_2 \Vert\bxi - \bxi_0\Vert_2^2,$$
	where $L_1 = \sqrt{\min\left({ 32^{-1}\lambda_{\rm min}}, \frac{\omega^2}{{\rm diam}_{\m}(\Xi)}\right)}$ and $L_2 =d^{-1/2}{L_1 {\delta }}$.
\end{lemma}

%\begin{remark}
Lemma \ref{lem:lower} shows that there exists a constant $L_2$ such that $ H(\ba, \bxi, \bxi_0)\geq L_2 \Vert\bxi - \bxi_0\Vert_2^2,$ provided that $\min_{g} \alpha_g \geq \delta >0$. It demonstrates that, to bound the Euclidean distance between $\bxi^{(k)} $ and $\bxi_0$, we only need to bound their Hellinger distance. 
Note that this lemma depends on an additional condition 	$\omega  > 0$, which can be guaranteed by the assumption that $\mathcal{P}^G$ is an identifiable finite mixture and the compactness of $\Theta$. In fact, by the identifiability,  if $ \bxi \neq \bxi_0$, then $H(\ba, \bxi, \bxi_0) > 0$. Since $H(\ba, \bxi, \bxi_0)$ is uniformly continuous on the compact set $\mathcal{D}(\Delta_1)$, we have $\omega  > 0$.  The  continuity of $H(\ba, \bxi, \bxi_0)$ is clear from Condition (\Ca). 
\subsection{Proofs of Lemma \lemtailEudis--\lemlower} \label{subsec:proofofbound}
In this section, we give the proofs of Lemma \lemlower, \lemtailEudis, \lemalpha  {} and  \lemtailtwo.  

\subsubsection{Proofs of Lemma \lemlower}

Before proving Lemma \lemlower, we state the following lemma that is often used in the proof. 
\begin{lemma} \label{lem:normtrans}
	Let $  {\bxi}_1, \bxi_2 \in \Xi $, where ${\bxi}_1 = \left({\btheta}_{11}, \dots, {\btheta}_{1G}  \right)$ and ${\bxi}_2 = \left({\btheta}_{21}, \dots, {\btheta}_{2G}  \right)$.
	Then, for any integer $k \geq 1$, $1\leq g \leq G$ and $j_1,\dots, j_k \in \{1,\dots,d \}$, we have 
	$$\prod_{s=1}^k \left|{\theta}_{1gj_s} - {\theta}_{2gj_s}\right| \leq 
	\left(\sqrt{d}\right)^k \left\|{\btheta}_{1g} -{\btheta}_{2g}\right\|_2^k,$$
	where $ {\btheta}_{ig} = ({\theta}_{ig1}, \dots, {\theta}_{igd})$ $(i = 1,2, g =1,\dots, G)$. 
\end{lemma}
\begin{proof}[Proof of Lemma \ref{lem:normtrans}]
	It is clear that 
	$$\prod_{s=1}^k \left|{\theta}_{1gj_s} - {\theta}_{2gj_s}\right| \leq 
	\left(\sum_{j=1}^d \left|{\theta}_{1gj} - {\theta}_{2gj} \right|\right)^k\leq \left(\sqrt{d}\right)^k \left\|{\btheta}_{1g} -{\btheta}_{2g}\right\|_2^k,$$
	where the last inequality is from Cauchy's inequality.
\end{proof}
\begin{proof}[Proof of Lemma \lemlower]
	For notation simplicity, we abbreviate  $\m(\ba,\bxi,\bxi_0)$ defined in (\ref{eq:m}) as $\m$, where $\bxi_0=(\btheta_0,\dots,\btheta_0)\in \Xi$. Note that 
	\begin{align*}
		H^2(\ba, \bxi, \bxi_0) &= 1-  \int \left\lbrace \sum_{g=1}^G \alpha_g f\left(x,\btheta_g \right) f\left(x,\btheta_0 \right) \right\rbrace^{1/2}  \mu({\rm d}x)  \\
		&= 1- \int \left\lbrace \frac{\sum_{g=1}^G \alpha_g f\left(x,\btheta_g \right)}{f\left(x,\btheta_0 \right)}  \right\rbrace^{1/2} f\left(x,\btheta_0 \right)  \mu(\rd x) \\
		&=  \int \left( 1- \left\lbrace \frac{\sum_{g=1}^G \alpha_g (f\left(x,\btheta_g \right)-f\left(x,\btheta_0 \right))}{f\left(x,\btheta_0 \right)}  + 1\right\rbrace^{1/2}\right)  f\left(x,\btheta_0 \right)  \mu(\rd x).
	\end{align*}
	Let $\delta(x) = \frac{1}{f\left(x,\btheta_0 \right)}\sum_{g=1}^G \alpha_g\left (f\left(x,\btheta_g \right)-f\left(x,\btheta_0 \right)\right)$. We can rewrite $ H^2(\ba, \bxi, \bxi_0)$ as 
	$$H^2(\ba, \bxi, \bxi_0) = \int \left(1-\sqrt{\delta(x)+1}\right)f\left(x,\btheta_0 \right) \mu(\rd x).$$
	Applying the inequality
	$$\sqrt{x+1} - 1 \leq x/2 - x^2/8 + x^3/16, $$
	and $\Ex [\delta(x)] = 0$, we have
	\begin{align}\label{H.ineq1}
		H^2(\ba, \bxi, \bxi_0) &\geq -\Ex [\delta(x)/2] + \Ex \left[\delta^2(x)/8\right] - \Ex \left[\delta^3(x)/16\right] \nonumber \\ 
		&= \Ex \left[\delta^2(x)/8\right] - \Ex \left[\delta^3(x)/16\right]. 
	\end{align} 
	\noindent
	{\bf Step 1. } We first consider the quadratic term. Define
	$${\bf b}(x) = (Y_{1}(x), \dots, Y_{d}(x), Z_{1}(x),\dots, Z_{d}(x), U_{12}(x), \dots , U_{(d-1)d}(x)),$$ 
	where $Y_{h}(x)$, $Z_{h}(x)$ ($h=1,\dots,d$) and $U_{h\ell}(x)$ ($1\leq h < \ell \leq d$) are defined as in (\ref{eq:b}) without the subscript $i$.	
	By Taylor's expansion, we have 
	\begin{align}\label{eq:deltax}
		\delta(x)  &= \sum_{g=1}^G {\alpha}_g\frac{f(x;   {\btheta}_g) - f(x;   {\btheta}_0)}{f(x;   {\btheta}_0)} \notag\\
		&= \sum_{h=1}^{d} \sum_{g=1}^G {\alpha}_g({\theta}_{gh} - {\theta}_{0h}) Y_{h}(x) +
		\sum_{h=1}^{d} \sum_{g=1}^G {\alpha}_g({\theta}_{gh} - {\theta}_{0h})^2 Z_{h}(x)  \notag \\
		&+ \sum_{h<\ell}^{d} \sum_{g=1}^G {\alpha}_g({\theta}_{gh} - {\theta}_{0h})({\theta}_{g\ell} - {\theta}_{0\ell}) U_{h\ell}(x)  + \varepsilon(x) \notag\\
		&= {\bf m}\trans {\bf b}(x) + \varepsilon(x). \notag
	\end{align}
	Here $\varepsilon(x)$ is the remainder term and can be accurately represented as 
	$$
	\varepsilon(x) = \sum_{j_{1}=1}^d \cdots \sum_{j_{3}=1}^d \sum_{g=1}^G {\alpha}_g \prod_{s=1}^3 \left({\theta}_{gj_s} - {\theta}_{0j_s}\right) \left(\frac{\partial^3 f(x,{\bm \zeta}_g(x))}{\partial \theta_{j_1} \partial \theta_{j_2}\partial \theta_{j_3}}\right) \bigg /  (3!f(x,\btheta_0)), 
	$$
	where ${\bm \zeta}_g(x)$ lies between ${\btheta}_g$ and $\btheta_0$. 
	Since $\delta(x) = {\bf m}\trans {\bf b}(x) + \varepsilon(x)$,  we have
	$$\delta^2(x) = \overbrace{{\bf m}\trans {\bf b}(x)  {\bf b}(x)\trans {\bf m}}^{\text {denote as I}} + \overbrace{2{\bf m}\trans {\bf b}(x) \varepsilon(x)}^{\text {denote as II}} + \overbrace{\varepsilon^2(x)}^{\text {denote as III}}. $$ 
	For the first term I, we have
	$$\Ex({\bf m}\trans {\bf b}(x)  {\bf b}(x)\trans {\bf m}) = {\bf m}\trans \B(\btheta_0) {\bf m} \geq \lambda_{\rm min} \Vert{\bf m}\Vert_2^2,$$
	where $ \B(\btheta_{0}) = \Ex\left({\bf b}(x)  {\bf b}(x)\trans\right)$.
	For the second term II,  by Cauchy's inequality, we have 
	$$2\Ex \left({\bf m}\trans {\bf b}(x) \varepsilon(x)\right) \geq -2 \Ex\left(\left|{\bf m}\trans {\bf b}(x)\Vert \varepsilon(x)\right|\right) \geq  -2\Vert{\bf m}\Vert_2 \Ex \left(\Vert{\bf b}(x)\Vert_2| \varepsilon(x)|\right).$$
	Hence, we aim to bound $\Ex(\Vert{\bf b}(x)\Vert_2| \varepsilon(x)|)$. For any fixed $j_1, j_2, j_3$ and $g$, by Lemma \ref{lem:normtrans}, we have
	\begin{equation}\label{eq:mm}
		{\alpha}_g \prod_{s=1}^3 ({\theta}_{gj_s} - {\theta}_{0j_s}) \leq {\alpha}_g d^{3/2} \Vert\btheta_g - \btheta_0 \Vert_2^3 \leq d^2 \Vert\btheta_g - \btheta_0 \Vert_2 \Vert{\bf m}\Vert_2.
	\end{equation}
	The second inequality of (\ref{eq:mm}) is from $\sqrt{d}\Vert{\bf m}\Vert_2 \geq {\alpha}_g \Vert\btheta_g - \btheta_0 \Vert_2^2  $, because
	\begin{equation}\label{eq:cauchy:m}
		\sqrt{d}\Vert{\bf m}\Vert_2 \geq \sqrt{d} \left({\sum_{h=1}^d \left\{\sum_{g=1}^G \alpha_g (\theta_{gh} - \theta_{0h})^2\right\}^2} \right)^{1/2}\geq \sum_{h=1}^d \sum_{g=1}^G \alpha_g (\theta_{gh} - \theta_{0h})^2.
	\end{equation}
	Remember that $g(x,\btheta_0)$ is the function defined in Condition (\Cc). 
	Then, we have 
	\begin{align*}
		& \quad \Ex \left(\Vert{\bf b}(x)\Vert_2 \left|\sum_{g=1}^G {\alpha}_g \prod_{s=1}^3 ({\theta}_{gj_s} - {\theta}_{0j_s}) \left(\frac{\partial^3 f(x,{\bm \zeta}_g(x,\btheta_0))}{\partial \theta_{j_1} \partial \theta_{j_2}\partial \theta_{j_3}}\right) \bigg /  (3!f(x,\btheta_0))\right|\right) \\
		\leq&\quad  	\Ex ( g(x,\btheta_0) \Vert{\bf b}(x)\Vert_2) d^2 \sum_{g=1}^G \Vert\btheta_g - \btheta_0 \Vert_2 \Vert{\bf m}\Vert_2  \ \ \ \ (\mbox{by (\ref{eq:mm}}))\\
		\leq&\quad  d^2 \sum_{g=1}^G \Vert\btheta_g - \btheta_0 \Vert_2 \Vert{\bf m}\Vert_2 (d+1)M^2. 
	\end{align*}
	where 
	$$\Ex ( g(x,\btheta_0) \Vert{\bf b}(x)\Vert_2) \leq \left( \Ex ( g^2(x,\btheta_0) ) \Ex ( \Vert{\bf b}(x)\Vert_2^2 ) \right)^{1/2} \leq (d+1)M^2$$
	is from Cauchy's inequality.
	It follows that
	\begin{align*}
		\Ex(\Vert{\bf b}(x)\Vert_2| \varepsilon(x)|) \leq  d^5(d+1)\sum_{g=1}^G \Vert\btheta_g - \btheta_0 \Vert_2 \Vert{\bf m}\Vert_2 M^2.
	\end{align*}	
	Therefore, there exists a constant $\Delta_{11}>0$ such that when $\sum_{g=1}^G \Vert\btheta_g -\btheta_0 \Vert_2 \leq \Delta_{11}$, 
	$$2 \Ex \left({\bf m}\trans {\bf b}(x) \varepsilon(x)\right) \geq -2^{-1}  \lambda_{\rm min} \Vert{\bf m}\Vert_2^2.$$
	Since $\Ex \left( \varepsilon^2(x)\right) > 0$,  when $\sum_{g=1}^G \Vert\btheta_g -\btheta_0 \Vert_2 \leq \Delta_{11}$,  we have
	$$\Ex (\delta^2(x)) \geq 2^{-1}  \lambda_{\rm min} \Vert{\bf m}\Vert_2^2.$$  
	
	\noindent 
	{\bf Step 2. } Next, we aim to prove that
	there exists $\Delta_{12}>0$ such that when $\sum_{g=1}^G \Vert\btheta_g -\btheta_0 \Vert_2 \leq \Delta_{12}$, 
	$\Ex \left( \left|\delta^3(x) \right|\right) \leq 2^{-1}  \lambda_{\rm min} \Vert{\bf m}\Vert_2^2$.
	We have
	$$\delta^3(x) = \left({\bf m}\trans {\bf b}(x)\right)^3  + 3{\bf m}\trans {\bf b}(x){\bf m}\trans {\bf b}(x) \varepsilon(x) + 3{\bf m}\trans {\bf b}(x) \varepsilon^2(x) + \varepsilon^3(x).$$
	Note that 
	\begin{align*}
		|\varepsilon(x)| &\leq \left|\sum_{j_{1}=1}^d \cdots \sum_{j_{3}=1}^d \sum_{g=1}^G {\alpha}_g \prod_{s=1}^3 ({\theta}_{gj_s} - {\theta}_{0j_s}) g(x;\btheta_0)\right| \\
		& \leq  d^5\Vert{\bf m}\Vert_2\sum_{g=1}^G \Vert\btheta_g -\btheta_0 \Vert_2  g(x;\btheta_0)  \ \ \ \ (\mbox{applying (\ref{eq:mm}})) \\ 
		& = \Vert{\bf m}\Vert_2 \widetilde{\varepsilon}(x),
	\end{align*}
	where  
	$$ \widetilde{\varepsilon}(x) = d^5 \sum_{g=1}^G \Vert\btheta_g -\btheta_0 \Vert_2 g(x,\btheta_0).$$
	Then, we have
	\begin{align*}
		\left|\delta(x)^3\right| &= \left| \left({\bf m}\trans {\bf b}(x)\right)^3 + 3{\bf m}\trans {\bf b}(x){\bf m}\trans {\bf b}(x) \varepsilon(x) + 3{\bf m}\trans {\bf b}(x) \varepsilon^2(x) + \varepsilon^3(x)\right| \\
		& \leq \Vert{\bf m}\Vert_2^3 \left( \Vert {\bf b}(x)\Vert_2^3 + 3 \Vert {\bf b}(x)\Vert_2^2|\widetilde{\varepsilon}(x)| + 3 \Vert {\bf b}(x)\Vert_2|\widetilde{\varepsilon}^2(x)| +  |\widetilde{\varepsilon}^3(x)|  \right). 
	\end{align*}
	By Condition (\Cc),  obviously,  the random variable 
	$$\Vert {\bf b}(x)\Vert_2^3 + 3 \Vert {\bf b}(x)\Vert_2^2|\widetilde{\varepsilon}(x)| + 3 \Vert {\bf b}(x)\Vert_2|\widetilde{\varepsilon}(x)|^2 +  |\widetilde{\varepsilon}(x)|^3  $$ 
	is  integrable.
	Then, similarly to the proof in Step 1, 
	there exists a constant $\Delta_{12}>0$ such that when $\sum_{g=1}^G \Vert\btheta_g -\btheta_0 \Vert_2 \leq \Delta_{12}$, we have
	$\Ex \left[\delta^3(x)\right] \leq  2^{-1} \lambda_{\rm min} \Vert{\bf m}\Vert_2^2$.
	Taking $\Delta_{1} = \min(\Delta_{11},\Delta_{12})$, by (\ref{H.ineq1}), for any $\btheta_0 \in \Theta$, when $\sum_{g=1}^G \Vert\btheta_g -\btheta_0 \Vert_2 \leq \Delta_1$, we have 
	$$ H^2(\ba, \bxi, \bxi_0) \geq 32^{-1}\lambda_{\rm min} \Vert{\bf m}\Vert_2^2.$$
	By the definition of $\omega$ and ${\rm diam}_m(\Xi)$, for any  $\bxi \in \Xi$, $\btheta_0 \in \Theta$ and $\ba \in \mathbb{S}_{\delta}$, we have
	$$ H^2(\ba, \bxi, \bxi_0) \geq \min\left\{\frac{\omega^2}{{\rm diam}_m(\Xi)},  32^{-1}\lambda_{\rm min} \right\}\Vert{\bf m}\Vert_2^2.$$
	Write $L_1 = \sqrt{\min\left\{\frac{\omega^2}{{\rm diam}_m(\Xi)},   32^{-1}\lambda_{\rm min} \right\}}$. Finally, by (\ref{eq:cauchy:m}), we have $ \sqrt{d} \Vert{\bf m}\Vert_2 \geq \alpha_{\rm min}  \sum_{g=1}^{G} \Vert\btheta_g -\btheta_0 \Vert_2^2 \geq \alpha_{\rm min} \Vert\bxi - \bxi_0\Vert_2^2$. Together with  $\ba \in \mathbb{S}_{\delta}$ yields $$L_1 \Vert{\bf m}\Vert_2 \geq \frac{L_1 {\alpha_{\rm min} }}{\sqrt{d}} \Vert\bxi - \bxi_0\Vert_2^4 \geq \frac{L_1 {\delta }}{\sqrt{d}} \Vert\bxi - \bxi_0\Vert_2^2,$$
	which finishes the proof of Lemma \lemlower. 
\end{proof}
\subsubsection{Proofs of Lemma \lemtailEudis}
To prove Lemma  \lemtailEudis, it remains to construct a link between the log-likelihood ratio and the Hellinger distance between the estimator and the true value.  
The following lemma shows that ${\bxi}^{(k)}$ is concentrated on $\bxi_0$ in the sense of the Hellinger distance.  In other words,  Lemma \ref{lem:tailH} constructs a link between the log-likelihood ratio and the Hellinger distance between $\left({\bxi}^{(k)}, \ba^{(k)}\right)$ and $\left(\bxi_0, \ba_0\right)$. 

\begin{lemma} \label{lem:tailH}
	Let $c_1 = 1/24, c_2 = (4/27)(1/1926), p_0 =  \lambda G \log(\delta G)$ and $c>0$ be a constant depending on $\delta, d,G, M$ and ${\rm diam}(\Xi)$.  Under Condition (\Ca)--(\Cd) and $\mathbb{H}_0$, 
	for any $\epsilon \geq c n^{-1/2} {\rm log} \ n$ and 
	$-n^{-1}p_0 \leq c_1 \epsilon^2$, we have 
	$$\Pr\left(\left\{H\left({\ba}^{(k)}, {\bxi}^{(k)}, \bxi_0\right) \leq \epsilon\right\} \cap \left\{ \ba^{(k)} \in \mathbb{S}_{\delta}  \right\}\right) \geq \Pr \left(\ba^{(k)} \in \mathbb{S}_{\delta}\right) - 5\exp\left(-c_2 n \epsilon^2\right).$$
\end{lemma}
Based on Lemma \ref{lem:tailH} and  Lemma \lemlower, we can prove Lemma \lemtailEudis.  
\begin{proof}[Proof of Lemma \lemtailEudis]
	By Lemma \ref{lem:tailH}, when 
	$\epsilon \geq  c n^{-1/2} \log n$ and $c_1\epsilon^2 \geq  \left(-p_0 \right) n^{-1} $, i.e.
	$$\epsilon \geq \max \left(c n^{-1/2} {\rm log} n, c_1^{-1/2} \left(-p_0 \right)^{1/2} n^{-1/2}\right),$$ 
	we have
	$$\Pr\left( \left\{H\left({\ba}^{(k)}, {\bxi}^{(k)}, \bxi_0\right) \leq \epsilon\right\} \cap \left\{ \ba^{(k)} \in \mathbb{S}_{\delta}  \right\} \right)  \geq \Pr \left(\ba^{(k)} \in \mathbb{S}_{\delta}\right) -5\exp\left(-c_2 n \epsilon^2\right).$$
	By Lemma \lemlower,  we have
	$$
	\left\{H\left({\ba}^{(k)}, {\bxi}^{(k)}, \bxi_0\right) \leq \epsilon\right\} \cap \left\{ \ba^{(k)} \in \mathbb{S}_{\delta}  \right\} \subset 	\mathcal{S}^{(k)}_{\epsilon} \cap  \left\{ \ba^{(k)} \in \mathbb{S}_{\delta}  \right\}. 
	$$
	It follows that
	$$\Pr\left(\mathcal{S}^{(k)}_{\epsilon} \cap  \left\{ \ba^{(k)} \in \mathbb{S}_{\delta}  \right\} \right)  \geq\Pr \left(\ba^{(k)} \in \mathbb{S}_{\delta}\right) -5\exp\left(-c_2 n \epsilon^2\right),$$
	and thus we complete the proof.
\end{proof}

We now  aim to prove Lemma \ref{lem:tailH}. 
We use the Hellinger distance entropy to measure the size of the parameter space $\Xi$. For any $u >0$, we call a finite set $\{ (f_j^L, f_j^U), j = 1, \dots, N \}$ a (Hellinger) $u$-bracketing of a distribution family $\mathcal{F}$, if $H(f_j^L , f_j^U) \leq u$,
and for any $p \in \mathcal{F}$, there is a $j$ such that $f_j^L \leq p \leq f_j^U $. Define the Hellinger distance entropy of $\mathcal{P}^G_{\mathbb{S}_{\delta}}$ as $\mathcal{H}(u, \mathcal{P}^G_{\mathbb{S}_{\delta}}) = $ logarithm of the cardinality of the $u$-bracketing of the smallest size. To bound the Hellinger distance entropy $\mathcal{H}(u, \mathcal{P}^G_{\mathbb{S}_{\delta}})$, we need the following Lipschitz property of  $\varphi^{1/2}(x; \bxi, \balpha)$.

\begin{lemma} \label{lem:upper}
	Under Condition (\Cc), if $\bxi_1, \bxi_2 \in \Xi, \ba_1, \ba_2 \in \mathbb{S}_{\delta}$,  then  
	$$|\varphi^{1/2}(x; \bxi_1, \balpha_1) - \varphi^{1/2}(x; \bxi_2, \balpha_2) | \leq ar(x) \left\|(\bxi_1, \ba_1) - (\bxi_2, \ba_2) \right\|_2    , $$
	where $r(x)$ is the  function as in Condition (\Cc) and  $a =  \sqrt{\frac{G}{4\delta}}$.
\end{lemma}

With Lemma \ref{lem:upper}, computing the Hellinger distance entropy can be converted to computing the Euclidean distance entropy. 
The following lemma gives an upper bound of  $\mathcal{H}(u, 
\mathcal{P}^G_{\mathbb{S}_{\delta}})$ based on   Lemma \ref{lem:upper}.

\begin{lemma}\label{lem:entropy}
	Under Condition (\Ca)--(\Cc), we have 
	$$\mathcal{H}(u, 
	\mathcal{P}^G_{\mathbb{S}_{\delta}}) \leq G(d+1) 
	{\rm log}\left(1+ \frac{2aM {\rm diam}(\Xi \times \mathbb{S}_{\delta} )}{u}\right),$$
	where $a =  \sqrt{\frac{G}{4\delta}}$ and ${\rm diam}(\Xi \times\mathbb{S}_{\delta})$ and  are  the Euclidean diameter of $\Xi\times\mathbb{S}_{\delta}$. 
\end{lemma}

We remark here that ${\rm diam}(\Xi \times \mathbb{S}_{\delta} ) $  is only depending on ${\rm diam}(\Xi), G$ and $d$ because the elements of $\ba$ with $\ba \in \mathbb{S}_{\delta}$  are bounded by 1. 
The following lemma from \cite{wong1995probability} gives a uniform exponential bound for the likelihood ratio.   
\begin{lemma} \label{lem:wingHeltail} 
	Taking $c_1 = 1/24, c_2 = (4/27)(1/1926), c_3 = 10, c_4 = (2/3)^{5/2}/512$,  for any $\epsilon > 0$, if 
	\begin{equation}\label{eps}
		\int_{\epsilon^{2} / 2^{8}}^{\sqrt{2} \epsilon} \mathcal{H}^{1 / 2}\left(u / c_{3}, \mathcal{P}^G_{\mathbb{S}_{\delta}} \right) d u \leq c_{4} n^{1 / 2} \epsilon^{2},
	\end{equation}
	then
	$$
	\Pr^{*}\left(\sup_{\substack{H(\ba, \bxi_1, \bxi_0) \geq \epsilon \\ \bxi_1 \in \Xi \\  \ba\in\mathbb{S}_{\delta}}} \prod_{i=1}^{n} \varphi\left(x_{i};  {\bxi_1}, \ba\right) / \varphi\left(x_{i}; {\bxi_0}, \ba\right) \geq \exp \left(-c_{1} n \epsilon^{2}\right)\right) \leq 4 \exp \left(-c_{2} n \epsilon^{2}\right),
	$$
	where $\Pr^{*} $ is understood to be the outer probability measure corresponding to the measure at $(\bxi_0, \ba_0)$.  
\end{lemma}
The following lemma claims that when $n\geq 2$, for any $ \epsilon \geq c n^{-1/2} \log n $, where $c$ is a constant only depending on $\delta,d,G, M$ and ${\rm diam}(\Xi)$, (\ref{eps}) holds. 
\begin{lemma}\label{lem:epsilon}
	Under Condition (\Ca)--(\Cc), there exists a constant $c$ depending on $\delta,d,G, M$ and ${\rm diam}(\Xi)$ such that 
	when  $n\geq 2$,  for any $ \epsilon \geq n^{-1/2} {\rm log} n $,   (\ref{eps}) holds. 
\end{lemma}
In fact, if we use the local Hellinger distance entropy, we can remove the $\log n$ factor and obtain a  stronger result.  However, in this paper, $ c n^{-1/2} \log n $ is sufficient. 
Thus,  based on Lemma \ref{lem:wingHeltail} and Lemma \ref{lem:epsilon},  we prove Lemma \ref{lem:tailH}. 
\begin{proof}[Proof of Lemma \ref{lem:tailH}]
	Since $\bxi^{(0)} = \arg \max_{\bxi\in \Xi} pl_n(\bxi, \balpha^{(0)})$, we have
	$$ pl_n\left(\bxi^{(0)}, \balpha^{(0)}\right) \geq  pl_n\left(\bxi_0, \balpha^{(0)}\right),$$
	and thus
	$$ l_n\left(\bxi^{(0)}, \balpha^{(0)}\right) \geq  l_n\left(\bxi_0, \balpha^{(0)}\right) = l_n\left(\bxi_0, \balpha_0\right).$$
	By  the property of the EM algorithm, for any $1\leq k \leq K$, we have 
	$$ pl_n\left(\bxi^{(k)}, \balpha^{(k)}\right) \geq pl_n\left(\bxi^{(0)}, \balpha^{(0)}\right).$$
	Since  $ p( \balpha^{(0)}) \geq p_0$,   we conclude that 
	\begin{equation}\label{eq:eta}
		l_n\left(\bxi^{(k)}, \balpha^{(k)}\right) - l_n\left(\bxi^{(0)}, \balpha^{(0)}\right) \geq p(\balpha^{(0)}) - p(\balpha^{(k)})\geq p(\balpha^{(0)}) \geq p_0. 
	\end{equation} 
	
	Next, by Lemma \ref{lem:epsilon} and \ref{lem:wingHeltail}, for any $\epsilon \geq c n^{-1/2} {\rm log} \ n$ and $-n^{-1}p_0\leq c_1 \epsilon^2$,  we have 
	$$
	\Pr^{*}\left(\sup_{\substack{H(\ba, \bxi_1, \bxi_0) \geq \epsilon \\ \bxi_1 \in \Xi \\  \ba\in\mathbb{S}_{\delta}}} \prod_{i=1}^{n} \varphi\left(x_{i};  {\bxi_1}, \ba\right) / \varphi\left(x_{i}; {\bxi_0}, \ba\right) \geq \exp \left(-c_{1} n \epsilon^{2}\right)\right) \leq 4 \exp \left(-c_{2} n \epsilon^{2}\right),
	$$
	or equivalently, 
	$$
	\Pr^{*}\left(\sup_{\substack{H(\ba, \bxi_1, \bxi_0) \geq \epsilon \\ \bxi_1 \in \Xi \\  \ba\in\mathbb{S}_{\delta}}}  l_n({\bxi_1}, \ba) -  l_n({\bxi_0}, \ba_0) \geq p_0 \right) \leq 4 \exp \left(-c_{2} n \epsilon^{2}\right).
	$$
	Write 
	$$
	\mathcal{G} =   \left\{\sup_{\substack{H(\ba, \bxi_1, \bxi_0) \geq \epsilon \\ \bxi_1 \in \Xi \\  \ba\in\mathbb{S}_{\delta}}}  l_n({\bxi_1}, \ba) -  l_n({\bxi_0}, \ba_0) \geq p_0 \right\}.
	$$
	By (\ref{eq:eta}) and the fact $\bxi^{(k)} \in \Xi$,  we have
	$$
	\left\{ H(\ba^{(k)}, \bxi^{(k)}, \bxi_0) \geq \epsilon \right\} \cup \left\{\ba^{(k)} \not \in \mathbb{S}_{\delta}\right\} \subset \mathcal{G} 
	\cup \left\{\ba^{(k)} \not \in \mathbb{S}_{\delta}\right\},
	$$
	because if $\ba^{(k)}  \in \mathbb{S}_{\delta}$, combining (\ref{eq:eta}) with $\left\{ H(\ba^{(k)}, \bxi^{(k)}, \bxi_0) \geq \epsilon \right\}$ implies $\mathcal{G}$.
	Thus, we conclude that 
	\begin{align*}
		\Pr \left( \left\{H\left({\ba}^{(k)}, {\bxi}^{(k)}, \bxi_0\right) \geq \epsilon\right\} \cup \left\{ \ba^{(k)} \not \in \mathbb{S}_{\delta}  \right\}\right)  &\leq  \Pr^{*} (\mathcal{G} )
		+ \Pr \left(\ba^{(k)} \not \in \mathbb{S}_{\delta}\right)\\
		&\leq 5 \exp \left(-c_{2} n \epsilon^{2}\right) + 1-\Pr \left(\ba^{(k)} \in \mathbb{S}_{\delta}\right),
	\end{align*}
	which proves this lemma.
\end{proof}

At the end of this section, we  give the proofs of Lemma \ref{lem:upper},  \ref{lem:entropy} and \ref{lem:epsilon}.  Before presenting their proofs, we give a bound of covering numbers of the Euclidean ball which can be founded in \cite{vershynin2018high} (Corollary 4.2.13). Let $\mathcal{N}(\varepsilon, K)$ be the smallest number of closed Euclidean balls with centers in $K$ and radius $\varepsilon$ whose union covers $K$.

\begin{lemma}\label{lem:boundconverball}
	The covering numbers of the unit Euclidean ball $B_2^p$ satisfy the following for any $\varepsilon>0$ :
	$$
	\left(\frac{1}{\varepsilon}\right)^p \leq \mathcal{N}\left( \varepsilon,B_2^p\right) \leq\left(\frac{2}{\varepsilon}+1\right)^p .
	$$
\end{lemma}

\begin{proof}[ Proof of Lemma \ref{lem:upper}]
	Since $\varphi^{1/2}(x; \bxi, \balpha) = \left(  \sum_{g=1}^{G} \alpha_{g} f(x; \btheta_g) \right)^{1/2}$,  the gradient of $\varphi^{1/2}$ can be written as 
	\begin{align*}
		& \quad	\nabla\varphi^{1/2}(x; \bxi, \balpha) \\
		&=  2^{-1}\varphi^{-1/2}(x; \bxi, \balpha)\left( \alpha_1 \nabla_{\btheta_1} f(x; \btheta_1), \ldots, \alpha_G \nabla_{\btheta_G} f(x; \btheta_G),f(x; \btheta_1), \ldots, f(x; \btheta_G)	\right). 
	\end{align*}
	By Lagrange's theorem and Cauchy's inequality, we have
	\begin{equation*}\label{eq:upper1}
		|\varphi^{1/2}(x; \bxi_1, \balpha_1) - \varphi^{1/2}(x; \bxi_2, \balpha_2) | \leq \left\|\nabla\varphi^{1/2}(x; \check{\bxi}(x), \check{\balpha}(x))\right\|_2 \left\|(\bxi_1, \ba_1) - (\bxi_2, \ba_2) \right\|_2, 
	\end{equation*}
	where $\check{\bxi}(x) = (\check{\btheta}_1(x) , \dots, \check{\btheta}_G(x) )$ lies between $\bxi_1$  and $\bxi_2$ and $\check{\balpha}(x)$ lies between $\ba_1$ and $\ba_2$.  
	Since $\check{\btheta}_g(x) \in \Theta \ (g = 1,\dots, G)$ and $\check{\balpha}(x) 
	\in \mathbb{S}_{\delta}$, it follows that
	\begin{align*}
		\left\|\nabla\varphi^{1/2}(x; \check{\bxi}(x), \check{\balpha}(x))\right\|_2^2 &= \sum_{g=1}^G \frac{\check{\alpha}_g^2 \left\| \nabla_{\btheta_g} f(x; \check{\btheta}_g(x)) \right\|_2^2 + f^2(x; \check{\btheta}_g(x)) }{4\varphi(x, \check{\bxi}(x), \check{\balpha}(x))}   \\
		&\leq  \sum_{g=1}^G \frac{\left\| \nabla_{\btheta_g} f(x; \check{\btheta}_g(x)) \right\|_2^2  + f^2(x; \check{\btheta}_g(x))}{4\check{\alpha}_g f(x;\check{\btheta}_g(x) )} \\
		&\leq \frac{G}{4\delta} r^{2}(x),
	\end{align*}
	where $r(x)$ is defined in Condition (\Cd).
	Thus,  
	we have
	$$|\varphi^{1/2}(x; \bxi_1, \balpha_1) - \varphi^{1/2}(x; \bxi_2, \balpha_2) | \leq \sqrt{\frac{G}{4\delta}}r(x) \left\|(\bxi_1, \ba_1) - (\bxi_2, \ba_2) \right\|_2    , $$
	which proves this lemma.
\end{proof}

\begin{proof}[ Proof of Lemma \ref{lem:entropy}]
	Let $a =  \sqrt{\frac{G}{4\delta}}$. 
	We use brackets of the type 
	$$\left[\left\{\left(\varphi^{1/2}(x; \bxi, \balpha)-ar(x)\epsilon\right)_{+}\right\}^2,  \left(\varphi^{1/2}(x; \bxi, \balpha)+ ar(x)\epsilon\right)^2 \right],$$
	for $(\bxi, \ba)$ ranging over a suitable chosen subset of $\Xi \times \mathbb{S}_{\delta}$. 
	Firstly, these brackets are of size no greater than $aM \epsilon $, because 
	\begin{align*}
		&\quad \left[2^{-1} \int
		\left(\varphi^{1/2}(x; \bxi, \balpha)+ar(x)\epsilon - \left(\varphi^{1/2}(x; \bxi, \balpha)- ar(x)\epsilon\right)_{+} \right)^2 {\rd}x \right]^{1/2} \\
		&\leq \left[2^{-1} \int
		\left(\varphi^{1/2}(x; \bxi, \balpha)+ ar(x)\epsilon - \left(\varphi^{1/2}(x; \bxi, \balpha)- ar(x)\epsilon\right) \right)^2 {\rd}x \right]^{1/2} \\
		&\leq \left[2^{-1} \int
		2a^2r^2(x)\epsilon^2 {\rd}x \right]^{1/2} \leq a \epsilon M, 
	\end{align*}
	where $\left[ \int r^2(x) {\rd}x \right]^{1/2} \leq M$ is from Condition (\Cc). 
	If $(\bxi, \ba)$ ranges over a grid of mesh-width $\epsilon$ over $\Xi \times \mathbb{S}_{\delta}$, then the brackets cover $\mathcal{P}^G_{\mathbb{S}_{\delta}}$.  It is because that by Lemma \ref{lem:upper}, 	
	$$\left\{\left(\varphi^{1/2}(x; \bxi, \balpha)- ar(x)\epsilon\right)_{+}\right\}^2 \leq  \varphi(x; \bxi_1, \balpha_1) \leq\left(\varphi^{1/2}(x; \bxi, \balpha)+ ar(x)\epsilon\right)^2,$$
	provided that $\left\|(\bxi_1, \ba_1) - (\bxi, \ba) \right\|_2 \leq \epsilon $. 
	Therefore, the smallest number of brackets with size $\epsilon$ whose union cover $\mathcal{P}^G_{\mathbb{S}_{\delta}}$ is less than the smallest number of  balls with radius $(aM)^{-1} \epsilon$ whose union cover $\Xi \times \mathbb{S}_{\delta}$.
	Since $\Xi \times \mathbb{S}_{\delta}$ is a compact set, by Lemma \ref{lem:boundconverball}, we have
	$$\mathcal{H}(u, 
	\mathcal{P}^G_{\mathbb{S}_{\delta}}) \leq G(d+1)
	{\rm log}\left(1+ \frac{2aM {\rm diam}(\Xi \times \mathbb{S}_{\delta}) }{u}\right),$$
	which proves the lemma.
\end{proof}

\begin{proof}[Proof of Lemma \ref{lem:epsilon}]
	Clearly, when $\sqrt{2}\epsilon \leq \epsilon^2/2^8$, i.e., $\epsilon \geq 2^8 \sqrt{2}$, (\ref{eps}) holds.  We now assume $ \epsilon \leq 2^8 \sqrt{2}$ and thus $\sqrt{2}\epsilon \leq  2^9$.  Let $a_1 = 2aM {\rm diam}(\Xi \times \mathbb{S}_{\delta} )$. Then, by Lemma \ref{lem:entropy}, we have  
	$$\mathcal{H}(u, 
	\mathcal{P}^G_{\mathbb{S}_{\delta}}) \leq G(d+1) 
	{\rm log}\left(1+ \frac{a_1 \vee (e-1)2^9 }{u}\right) \leq G(d+1) 
	{\rm log}\left( \frac{2(a_1 \vee (e-1)2^9) }{u}\right), $$
	when $u \leq 2^9$.  Let $a_2=2(a_1 \vee (e-1)2^9)$ and $a_3 = G(d+1)$.   Thus, we have 
	\begin{align*}
		\int_{\epsilon^{2} / 2^{8}}^{\sqrt{2} \epsilon} \mathcal{H}^{1 / 2}\left(u / c_{3}, \mathcal{P}^G_{\mathbb{S}_{\delta}} \right) d u &\leq 	 a_3 \int_{\epsilon^{2} / 2^{8}}^{\sqrt{2} \epsilon}  \left\{ {\rm log}\left( \frac{a_2 }{u}\right)  \right\}^{1 / 2}  d u \\
		&\leq a_3\int_{\epsilon^{2} / 2^{8}}^{\sqrt{2} \epsilon}  {\rm log}\left( \frac{a_2 }{u}\right)   d u, \ \ \ \ (\mbox{since }  {\rm log}\left( {a_2 }/{u}\right) \geq 1)  \\
		& =  a_3a_2\int_{\epsilon^{2} / (a_2 2^{8})}^{\sqrt{2} \epsilon/a_2} -{\rm log}\left(t\right)  d t, \ \ \ \ (\mbox{let } t = u/a_2) \\
		& = \left( t-t\log t \bigg |_{\epsilon^{2} / (a_2 2^{8})}^{\sqrt{2} \epsilon/a_2}\right) a_3a_2. 
	\end{align*}
	When $\epsilon \leq 2^8 \sqrt{2}$, we have 
	$$
	\frac{\epsilon^{2}}{a_2 2^{8}}  \leq \frac{2^{17}}{(e-1) 2^{18}} < e.
	$$
	Write 
	$\phi(t) =  t - t \log t $. 
	Using the fact that $\phi(t) \geq 0$ when $t \leq e$, we have $ \phi\left(\epsilon^{2} / (a_2 2^{8})\right) > 0 $.
	It follows that 
	$$\left( t-t\log t \bigg |_{\epsilon^{2} / (a_2 2^{8})}^{\sqrt{2} \epsilon/a_2}\right) a_3a_2 \leq {a_3\sqrt{2} \epsilon} - \left({a_3\sqrt{2} \epsilon}\right) \log \{{\sqrt{2} \epsilon/a_2}\}.$$
	Therefore,  we conclude that there exists two constants $c'>0$ and $c''$ such that 
	$$
	\int_{\epsilon^{2} / 2^{8}}^{\sqrt{2} \epsilon} \mathcal{H}^{1 / 2}\left(u / c_{3}, \mathcal{P}^G_{\mathbb{S}_{\delta}} \right) d u \leq \epsilon \left(c'' - c' \log \ \epsilon \right).
	$$
	In order to ensure that (\ref{eps}) holds, we only need that 
	\begin{equation}\label{eq:A6}
		c'' - c' \log \ \epsilon  \leq c_4 n^{1/2} \epsilon.
	\end{equation}
	It is clear that we can choose a sufficiently large $c >0$ such that when  $n\geq 2$,  for any $ \epsilon \geq n^{-1/2} {\rm log} n $,  (\ref{eq:A6}) holds.  The proof is complete.
\end{proof}

\subsubsection{Proofs of Lemma \lemtailtwo}
To prove Lemma \lemtailtwo, we first prove Lemma \lemalpha. 
Recall that 
%	\begin{equation}\label{eq:alphaA}
	%	\mathbb{S}_{\delta} = \left\{ \balpha: \ba \in \mathbb{S}^{G-1}, \min_{g=1,\dots,G} \alpha_g \geq\delta >0 \right\},
	%	\end{equation} 
%	\begin{equation}\label{eq:calS}
	%	\mathcal{S}^{(k)}_{\epsilon} =  \left\{\left\|{\bf {m}}\left(\ba^{(k)}, {\bxi}^{(k)}\right)\right\|_2 < \frac{\epsilon}{L_1}, \left\|{\bxi}^{(k)}-\bxi_0\right\|_2^2 < \frac{\epsilon}{L_2}\right\},
	%	\end{equation}
%	and 
\begin{equation}\label{eq:calE}
	\mathcal{E}^{(k+1)} = \left\{\min_{g=1,\dots,G} {\alpha}_g^{(k+1)} \geq \min_{g=1,\dots,G}  {\alpha}_g^{(k)}\left(1-\frac{2}{K}\right) \right\}.
\end{equation}  

The following lemma gives the definition of $\Delta_K$. 
\begin{lemma}\label{lem:deltaK}
	For all  $\btheta_0 \in \Theta$,  there exists a constant $\Delta_K>0$ such that when $\Vert\btheta - \btheta_0\Vert_2 \leq \Delta_K$, we have
	$$\Ex \left(\frac{\inf_{\Vert\btheta - \btheta_0\Vert_2\leq\Delta_K } f(x;\btheta)}{\sup_{\Vert\btheta - \btheta_0\Vert_2\leq\Delta_K} f(x;\btheta)}\right)\geq 1-\frac{1}{K}.$$
\end{lemma}
\begin{proof}[Proof of Lemma \ref{lem:deltaK}]
	Observe that  for any $h>0$, we have
	$$
	0\leq \inf_{\btheta_{0} \in \Theta}\frac{\inf_{\Vert\btheta - \btheta_0\Vert_2\leq h } f(x;\btheta)}{\sup_{\Vert\btheta - \btheta_0\Vert_2\leq h} f(x;\btheta)} \leq 1.
	$$
	By the  dominated convergence theorem, the compactness of $\Theta$ and the continuity of $f(x;\btheta) $, we have 
	$$ \lim_{h \rightarrow 0}\Ex\left(\inf_{\btheta_{0} \in \Theta}\frac{\inf_{\Vert\btheta - \btheta_0\Vert_2\leq h } f(x;\btheta)}{\sup_{\Vert\btheta - \btheta_0\Vert_2\leq h} f(x;\btheta)}\right) = \Ex\left(\lim_{h \rightarrow 0}\inf_{\btheta_{0} \in \Theta}\frac{\inf_{\Vert\btheta - \btheta_0\Vert_2\leq h } f(x;\btheta)}{\sup_{\Vert\btheta - \btheta_0\Vert_2\leq h} f(x;\btheta)}\right) = 1.$$ 
	Therefore,  there is  $\Delta_K>0$ such that the inequality in Lemma  \ref{lem:deltaK} holds. 
\end{proof}
Before proving Lemma \lemalpha, we state the following Hoeffding's  inequality which can be found in \cite{vershynin2018high} (Theorem 2.2.6).

\begin{lemma}\label{lem:hoeff}
	Let $X_1, \ldots, X_N$ be independent random variables. Assume that $X_i \in\left[m_i, M_i\right]$ for every $i$. Then, for any $t>0$, we have
	$$
	\mathbb{P}\left\{\sum_{i=1}^N\left(X_i-\mathbb{E} X_i\right) \geq t\right\} \leq \exp \left(-\frac{2 t^2}{\sum_{i=1}^N\left(M_i-m_i\right)^2}\right).
	$$
\end{lemma}

\begin{proof}[Proof of Lemma \lemalpha]
	Recall that  $p(\balpha) = \lambda \left(\sum_{g=1}^G\log(\alpha_g) + G\log G\right)$ and 
	\begin{equation}\label{eq:defw}
		w_{gi}^{(k)} = \frac{\alpha^{(k)}_{g} f\left(x_i; \btheta_g^{(k)}\right)}{\varphi\left(x_i; \bxi^{(k)}, \balpha^{(k)}\right)}.	
	\end{equation} Then, the update of $\balpha$ can be written as $${\alpha}_g^{(k+1)} = \frac{\sum_{i=1}^n {w}_{gi}^{(k)} + \lambda }{n+G\lambda},$$
	which is a weighted sum of $n^{-1}\sum_{i=1}^n w_{gi}^{(k)}$ and $G^{-1}$ and shrinks $n^{-1}\sum_{i=1}^n{w_{gi}^{(k)}}$ towards $G^{-1}$.
	Thus,  we conclude that 
	\begin{equation}\label{eq:boundalpha}
		\min_{g=1,\ldots,G} {\alpha}_g^{(k+1)} \geq \min_{g=1,\ldots,G} n^{-1}{\sum_{i=1}^n {w}_{gi}^{(k)}}.	
	\end{equation}
	Thus, we only need to bound $\min_{g=1,\dots,G} {n^{-1}\sum_{i=1}^n {w}_{gi}^{(k)}}$.
	Let 
	$$T_i = \frac{\inf_{\Vert\btheta - \btheta_0\Vert_2\leq\Delta_K } f(x_i;\btheta)}{\sup_{\Vert\btheta - \btheta_0\Vert_2\leq\Delta_K} f(x_i;\btheta)},$$
	and  $\mathcal{K} = \mathcal{S}^{(k)}_{\epsilon}  \cap \mathcal{B}$, where $\Delta_K$ is as defined in Lemma \ref{lem:deltaK}, and $S_\epsilon^{(k)}$ is defined in (\ref{eq:calS}). 
	Since  $L_2\Delta_K^2 \geq \epsilon$,  on $\mathcal{K} $, we have $\left\|{\bxi}^{(k)}-\bxi_0\right\|_2 \leq \Delta_K$.
	It follows that on $\mathcal{K} $, 
	$ {w}_{gi}^{(k)} \geq {\alpha}_{g}^{(k)}T_i$.  
	Thus, 
	we conclude that 
	\begin{align}\label{eq:boundalpha2}
		&\quad	\left\{  n^{-1}{\sum_{i=1}^n T_i} \geq 1-2K^{-1}\right\} \cap \mathcal{K} \notag\\ 
		&\subset	\left\{ \min_{g=1,\ldots,G}  n^{-1}{\sum_{i=1}^n {w}_{gi}^{(k)}} \geq \min_{g=1,\ldots,G} {\alpha}_g^{(k)}\left( 1- 2K^{-1} \right) \right\} \cap \mathcal{K}  \notag \\
		&\subset \mathcal{E}^{(k+1)} \cap \mathcal{K},
	\end{align}
	where $  \mathcal{E}^{(k+1)}$ is defined in (\ref{eq:calE}).
	
	Then, it suffices to bound the probability
	$\Pr \left(n^{-1}{\sum_{i=1}^n T_i} \geq 1-2K^{-1} \right).$
	Note that $0\leq T_i\leq 1$. Hence, by Lemma \ref{lem:hoeff}, we have 
	$$\Pr \left(n^{-1}\left|{\sum_{i=1}^n T_i- \Ex(T_i)} \right| \geq K^{-1} \right) \leq 2\exp\left({-2 n}/{K^2} \right).$$
	By Lemma \ref{lem:deltaK}, we have
	$1 \geq \Ex (T_i) \geq 1-K^{-1},$  and thus
	$$\Pr \left(n^{-1}{\sum_{i=1}^n T_i} \geq 1-2K^{-1} \right) \geq 1-2\exp\left({-2 n}/{K^2} \right).$$
	Applying (\ref{eq:boundalpha}) and  (\ref{eq:boundalpha2}), we have 
	\begin{align*}
		\Pr\left( \mathcal{E}^{(k+1)} \cap \mathcal{K} \right) & \geq 	\pr \left(n^{-1}{\sum_{i=1}^n T_i} \geq 1-2K^{-1}\right) +  \pr \left(\mathcal{K} \right) - 1
		\\ &\geq \pr \left(\mathcal{K}\right)-2\exp\left({-2 n}/{K^2} \right),
	\end{align*}
	and Lemma \lemalpha {} is proved.
\end{proof}

Finally, combining Lemma \lemtailEudis {} with Lemma \lemalpha, we can prove  Lemma \lemtailtwo.
\begin{proof}[Proof of Lemma \lemtailtwo]
	Recall the definition of $\mathbb{S}_{\delta}$, $\mathcal{S}_{\epsilon}^{(k)}  $ and $\mathcal{E}^{(k+1)} $ defined in (\ref{eq:alphaA}), (\ref{eq:calS}) and (\ref{eq:calE}). 
	For $0\leq k \leq K$, we define 
	\begin{equation}\label{eq:alphak}
		\mathcal{B}^{(k)}=  \left\{  \min_{g=1,\ldots,G} \alpha^{(k)}_g \geq \min_{g=1,\ldots,G}  \left(1-\frac{2}{K}\right)^k{\alpha}_g^{(0)}  \right\}.
	\end{equation}
	It is clear that for any $0\leq k \leq K$,  $\mathcal{B}^{(k)} \subset \left\{ \ba^{(k)} \in \mathbb{S}_{\delta}\right\}$  because $(1-2/K)^k \geq 27^{-1}$ for $K\geq 3$. 
	We aim to prove a stronger result
	\begin{equation}\label{eq:strongerresult}
		\Pr\left( 	\mathcal{S}_{\epsilon}^{(k)} \cap 	\mathcal{B}^{(k)}  \right) \geq 1-5(k+1)\exp\left(-c_2 n \epsilon^2 \right)-2k\exp\left(\frac{-2 n}{K^2} \right), 0\leq k\leq K.
	\end{equation}
	We use mathematical induction to prove (\ref{eq:strongerresult}). 
	We first give the proof for the case $k = 0$.  It is clear that $ \mathcal{S}_{\epsilon}^{(0)} \cap 	\mathcal{B}^{(0)}  = \mathcal{S}_{\epsilon}^{(0)}. $
	Since $\ba^{(0)} \in \mathbb{S}_{\delta} $,  by Lemma \lemtailEudis, we have 
	\begin{equation}\label{eq:conind0}
		\Pr\left( 	\mathcal{S}_{\epsilon}^{(0)}  \cap 	\mathcal{B}^{(0)} \right) \geq 1-5\exp\left(-c_2 n \epsilon^2 \right).  
	\end{equation}
	
	Assume the result holds for $k < K$,  we will prove it for $k+1$.  On $	\mathcal{S}_{\epsilon}^{(k)} \cap 	\mathcal{B}^{(k)}$, 
	since $L_2\Delta_K^2 \geq \epsilon$, 
	we have $\left\|{\bxi}^{(k)}-\bxi_0\right\|_2 \leq \Delta_K$, i.e. $ \mathcal{S}_{\epsilon}^{(k)} \cap 	\mathcal{B}^{(k)} \subset \left\{ \left\|{\bxi}^{(k)}-\bxi_0\right\|_2 \leq \Delta_K \right\}$. 
	By the inductive hypothesis, we have
	$$
	\Pr\left( 	\mathcal{S}_{\epsilon}^{(k)} \cap 	\mathcal{B}^{(k)}  \right) \geq 1-5(k+1)\exp\left(-c_2 n \epsilon^2 \right)-2k\exp\left(\frac{-2 n}{K^2} \right).
	$$
	Thus, by Lemma \lemalpha, we have
	$$
	\Pr\left(\mathcal{E}^{(k+1)} \cap 	\mathcal{S}_{\epsilon}^{(k)} \cap 	\mathcal{B}^{(k)} \right) \geq 
	1-5(k+1)\exp\left(-c_2 n \epsilon^2 \right)-2(k+1)\exp\left(\frac{-2 n}{K^2} \right).
	$$
	Note that 
	$$ \mathcal{E}^{(k+1)} \cap	\mathcal{B}^{(k)} \subset 	\mathcal{B}^{(k+1)},$$
	because 
	$$
	\min_{g=1,\dots,G} {\alpha}_g^{(k+1)} \geq \min_{g=1,\dots,G}  {\alpha}_g^{(k)}\left(1-\frac{2}{K}\right) \geq
	\min_{g=1,\ldots,G}  \left(1-\frac{2}{K}\right)^{k+1}{\alpha}_g^{(0)}. 
	$$
	Thus, we conclude that
	$$
	\Pr\left(  \left\{\ba^{(k+1)} \in \mathbb{S}_{\delta}\right\}  \right) \geq \Pr\left(	\mathcal{B}^{(k+1)}  \right) \geq \Pr\left(\mathcal{E}^{(k+1)} \cap 	\mathcal{S}_{\epsilon}^{(k)} \cap 	\mathcal{B}^{(k)} \right). 
	$$
	By Lemma \lemtailEudis, 	we have 
	\begin{align*}
		\Pr\left( \mathcal{S}^{(k+1)}_{\epsilon} \cap \left\{ \ba^{(k+1)} \in \mathbb{S}_{\delta}\right\} \right) &\geq\Pr\left(  \left\{\ba^{(k+1)} \in \mathbb{S}_{\delta}\right\}  \right) -5\exp\left(-c_2 n \epsilon^2 \right) \\ 
		&\geq  
		1-5(k+2)\exp\left(-c_2 n \epsilon^2 \right)-2(k+1)\exp\left(\frac{-2 n}{K^2} \right),
	\end{align*}
	and thus we complete the proof.
\end{proof}
\subsection{Proofs of Theorem \thmzeroRes} \label{subsec:proofofH01}
In order to derive a tail probability bound for the EM-test statistic, we need the following lemmas.
\begin{lemma}[Rosenthal's  inequality] \label{lem:Rosen}
	Suppose that $\{X_i\}_{i=1}^n$ are mean-zero  and independent random variables and satisfy the moment bound $\Vert X_i\Vert_{L^{2m}} \leq C , 1 \leq i \leq n $ with some fixed integer $m\geq 1$. Then, we have 
	$$\Pr \left\{\left|\sum_{i=1}^n X_i\right| \geq nt   \right\}\leq  2R_m \left(\frac{C}{\sqrt{n}t}\right)^{2m}, \ \ \mbox{for all } t>0,$$
	where $R_m$ is a universal constant only depending on $m$. 
	Further, if  $\Ex(X_i) \neq 0$, then we have
	$$\Pr \left\{\left|\sum_{i=1}^n \left[X_i - \Ex(X_i)\right]\right| \geq nt   \right\}\leq  2R_m \left(\frac{2C}{\sqrt{n}t}\right)^{2m}, \ \ \mbox{for all } t>0.$$
\end{lemma}

\begin{lemma} \label{lem:normcontrol}
	Let $p,q \in \mathbb{N}$ and fulfill $p+q \leq 3$.  Let $g(x; \btheta_{0})$ and $m$ be the same as in Condition (\Cc).
	Write $R_{pq}(x_i) =  g^p(x_i; \btheta_{0})  \Vert \b_i \Vert_2^q $,
	where  
	$$  \left\|\b_i\right\|_2 = \left( \sum_{j=1}^d Y_{ij}^2 + \sum_{j=1}^d Z_{ij}^2 + \sum_{j_1=1}^d\sum_{j_2> j_1}^d U_{ij_1j_2}^2\right)^{1/2}.$$
	Then, under $\mathbb{H}_0$ and Condition (\Cc), we have 
	$$\Pr \left\{\left|\sum_{i=1}^n R_{pq}(x_i)\right| \geq n(1+(d+1)^{q} M^{p+q})   \right\}\leq  2R_m \left(\frac{2(d+1)^{q} M^{p+q}}{\sqrt{n}}\right)^{2m}.$$
\end{lemma}

\begin{lemma} \label{lem:cont-deri}
	Let $k \in \{3,4\}$ and  $j_1,\dots, j_k \in \{1,\dots,d \}$.  Define 
	$$D_i(j_1, \dots, j_k) = \left(\frac{\partial^k f(x_i; \btheta_0)}{\partial \theta_{j_1}  \cdots\partial \theta_{j_k}}\right) \bigg/  (k!f(x_i; \btheta_0)).$$
	Then, under $\mathbb{H}_0$ and Condition (\Cc), we have 
	$$\Pr\left\{\sum_{j_{1}=1}^d \cdots \sum_{j_{k}=1}^d\left|\sum_{i=1}^n D_i(j_1, \dots, j_k) \right| < d^kn^{5/8}  \right\}\geq 1- 2d^kR_m  \left(\frac{M}{ n^{1/8} }\right)^{2m},$$
\end{lemma}

\begin{proof}[Proof of Theorem \thmzeroRes]
	Recall that ${\rm EM}_n^{(K)} = \max\left\{ M_n^{(K)}(\ba_t),  t=1,\dots,T\right\}.$ Without loss of generality, we assume $T=1$ and $ {\rm EM}_n^{(K)}  =M_n^{(K)}\left(\ba^{(0)}\right) .$ 
	%{\bf[check.]}
	Considering that
	$$M_n^{(K)}\left(\ba^{(0)}\right) = 2\left[pl_n\left( {\bxi}^{(K)},{\balpha}^{(K)}\right) - pl_n( \bxi_0,\ba_0) + pl_n( \bxi_0,\ba_0) - pl_n\left( \hat{\bxi}_0,\ba_0\right) \right],$$
	we let $R_{1n} =2\left[pl_n\left( {\bxi}^{(K)},{\balpha}^{(K)}\right) - pl_n( \bxi_0,\ba_0) \right] $ and $R_{0n} =2\left[ pl_n( \bxi_0,\ba_0) - pl_n\left( \hat{\bxi}_0,\ba_0\right) \right] $.  
	Since $R_{0n} \leq 0$, we have $R_{1n} \geq M_n^{(K)}\left(\ba^{(0)}\right) $. 
	Hence, we only consider the $R_{1n}$ term. For notation simplicity, we write $\bar{\ba}, \bar{\bxi} $ in replacement of ${\balpha}^{(K)}, {\bxi}^{(K)}$.
	
	Next, we focus on the $R_{1n}$ term. Since $p(\ba)$ is maximized at $\ba_0$, we have
	\begin{align}\label{R1-inequality}
		R_{1n} &\leq 2 \left\{l_n( \bar{\bxi},\bar{\balpha} ) - l_n( \bxi_0, \ba_0)\right\}  \nonumber\\
		&= 2\sum_{i=1}^n \log\left(1+\sum_{g=1}^G \bar{\alpha}_g\left(\frac{f(x_i;   \bar{\btheta}_g)}{f(x_i;   {\btheta}_0)} -1 \right)\right) \nonumber\\
		&= \sum_{i=1}^n 2\log(1+\delta_i), \nonumber
	\end{align}
	where $\delta_i = \sum_{g=1}^G \bar{\alpha}_g\left(\frac{f(x_i;   \bar{\btheta}_g)}{f(x_i;   {\btheta}_0)} -1 \right)$. 
	Applying the inequality $\log(1+x) \leq x - x^2/2 + x^3/3$, we have
	
	\begin{equation}\label{eq:R1n}
		R_{1n} \leq \sum_{i=1}^n 2\log\left(1+\delta_i\right) \leq 2\sum_{i=1}^{n} \delta_i - \sum_{i=1}^{n} \delta_i^2 + (2/3) \sum_{i=1}^{n} \delta_i^3,	
	\end{equation}
	where $\delta_i = \sum_{g=1}^G \bar{\alpha}_g\left(\frac{f\left(x_i;  \bar{\btheta}_g\right)}{f\left(x_i;  {\btheta}_0\right)} -1 \right)$.
	Let 
	$$\bar{\m} = \m(\bar{\ba}, \bar{\bxi}, \bxi_0),$$
	where $\m$ is defined in (\ref{eq:m}). 
	Define  
	\begin{equation}\label{eq:epsilon}
		\varepsilon_{in} =  \delta_i - \bar{\m}^{ {\rm T}} \b_i.
	\end{equation}
	Plugging (\ref{eq:epsilon}) into (\ref{eq:R1n}), we have 
	$$2\sum_{i=1}^{n} \delta_i  =  2\sum_{i=1}^{n} \bar{\m}^{ {\rm T}} \b_i + 2\sum_{i=1}^{n} \varepsilon_{in}. $$
	and 
	$$-\sum_{i=1}^{n} \delta^2_i  \leq  - \sum_{i=1}^{n} \bar{\m}^{ {\rm T}} \b_i \b_i\trans \bar{\m} - 2\sum_{i=1}^{n} \bar{\m}^{ {\rm T}} \b_i\varepsilon_{in}, $$
	because $ \sum_{i=1}^{n}  \varepsilon_{in}^2 \geq 0$.  Therefore, (\ref{eq:R1n}) can be rewritten as
	\begin{align}\label{eq:R1nexpan}
		R_{1n} &\leq 2\sum_{i=1}^{n} \bar{\m}^{ {\rm T}} \b_i - \sum_{i=1}^{n} \bar{\m}^{ {\rm T}} \b_i \b_i\trans \bar{\m} + 2\sum_{i=1}^{n} \varepsilon_{in} \\
		&\quad - 2\sum_{i=1}^{n} \bar{\m}^{ {\rm T}} \b_i\varepsilon_{in} + (2/3) \sum_{i=1}^{n} \delta_i^3 \notag.
	\end{align}
	Our next goal is to control the $  2\sum_{i=1}^{n} \varepsilon_{in} - 2\sum_{i=1}^{n} \bar{\m}^{ {\rm T}} \b_i\varepsilon_{in} + (2/3) \sum_{i=1}^{n} \delta_i^3$ term, and it will be divided into three steps.  
	
	\noindent
	{\bf Step 1: } In the first step, we bound the $2\sum_{i=1}^{n}\varepsilon_{in} $ term.
	By Taylor's expansion to the fifth order,  $\varepsilon_{in}$ can be accurately represented as 
	\begin{align*}
		\varepsilon_{in} &= 	\sum_{j_{1}=1}^d \cdots \sum_{j_{3}=1}^d \sum_{g=1}^G \bar{\alpha}_g \prod_{s=1}^3 \left(\bar{\theta}_{gj_s} - {\theta}_{0j_s}\right) \left(\frac{\partial^3 f(x_i; \btheta_0)}{\partial \theta_{j_1} \partial \theta_{j_2}\partial \theta_{j_3}}\right) \bigg/  (3!f(x_i; \btheta_0)) \\
		&+  \sum_{j_{1}=1}^d \cdots \sum_{j_{4}=1}^d \sum_{g=1}^G \bar{\alpha}_g \prod_{s=1}^4 \left(\bar{\theta}_{gj_s} - {\theta}_{0j_s}\right) \left(\frac{\partial^4 f(x_i; \btheta_0)}{\partial \theta_{j_1} \partial \theta_{j_2}\partial \theta_{j_3}\partial \theta_{j_4}}\right) \bigg/  (4!f(x_i; \btheta_0)) \\
		&+  \sum_{j_{1}=1}^d \cdots \sum_{j_{5}=1}^d  \sum_{g=1}^G \bar{\alpha}_g \prod_{s=1}^5 \left(\bar{\theta}_{gj_s} - {\theta}_{0j_s}\right) \left(\frac{\partial^5 f(x_i; {\bm \zeta}_g(x_i))}{\partial \theta_{j_1} \partial \theta_{j_2}\partial \theta_{j_3}\partial \theta_{j_4}\partial \theta_{j_5}}\right)\bigg /  (5!f(x_i; \btheta_0)) \\
		&= {\rm I + II + III},
	\end{align*}
	where ${\bm \zeta}_g (x_i)$ lies between $\bar{\btheta}_g$ and $\btheta_0$. 
	Take $\epsilon = cn^{-11/24} \log n \wedge L_2 \Delta^2$, where $\Delta = \Delta_K \wedge  \tau$. 
	Let 
	$$\mathcal{A}_1 = \left\{ \left \|\bar{\bxi} - \bxi_0\right \|_2^2 < \frac{\epsilon}{L_2}, \left\|
	\bar{\bf{m}}\right\|_2 < \frac{\epsilon}{L_1} \right\}.$$ 
	By Lemma \lemtailtwo, when $n$ is large enough such that  $$L_2\Delta^2 \geq \max \left(c n^{-1/2} {\rm log} n, c_1^{-1/2} \sqrt{-p_0} n^{-1/2}\right),$$ we have 
	\begin{equation}\label{eq:H0A1prob}
		\Pr(\mathcal{A}_1) \geq 1-5(K+1)\exp\left(-c_2 n \epsilon^2 \right)-2K\exp\left(\frac{-2 n}{K^2} \right).
	\end{equation}
	On $\mathcal{A}_1$, we have 
	$$\left\|\bar{\bxi}- \bxi_0\right\|_2 < \left(\sqrt{\frac{{c}}{L_2}}\right)n^{-11/48}\log^{1/2}n \mbox{ and } \left\|\bar{\bf{m}}\right\|_2 < \frac{c}{L_1}n^{-11/24} \log n.$$ 
	For fixed $j_1,j_2,j_3$, by Lemma \ref{lem:normtrans}, we have
	$$\sum_{g=1}^G \bar{\alpha}_g \prod_{s=1}^3 \left|\left(\bar{\theta}_{gj_s} - {\theta}_{0j_s}\right)\right| \leq \sum_{g=1}^G  \bar{\alpha}_g \left(\sqrt{d}\right)^3 \left\|\bar{\btheta}_g -{\btheta}_{0}\right\|_2^3 \leq (\sqrt{d})^3 \left \|\bar{\bxi} - \bxi_0\right \|_2^3.$$
	Let 
	$$\mathcal{A}_2 = \left\{ \sum_{j_{1}=1}^d \cdots \sum_{j_{3}=1}^d\left|\sum_{i=1}^{n} \left(\frac{\partial^3 f(x_i; \btheta_0)}{\partial \theta_{j_1} \partial \theta_{j_2}\partial \theta_{j_3}}\right) \bigg/  (3!f(x_i; \btheta_0))\right| < d^{3} n^{5/8} \right\}.$$
	Then, on $\mathcal{A}_1 \cap \mathcal{A}_2$, we have 
	$${\rm |I|} < d^{9/2}\left(\sqrt{\frac{{c}}{L_2}}\right)^3 n^{-1/16}\log^{3/2}n.$$
	By Lemma  \ref{lem:cont-deri}, it follows that 
	$$P(\mathcal{A}_2) \geq 1-2d^3 R_m \left(\frac{M}{ n^{1/8} }\right)^{2m}.$$ 
	Similarly, we let
	$$\mathcal{A}_3 = \sum_{j_{1}=1}^d \cdots \sum_{j_{4}=1}^d \left\{ \left|\sum_{i=1}^{n} \left(\frac{\partial^4 f(x_i; \btheta_0)}{\partial \theta_{j_1} \partial \theta_{j_2}\partial \theta_{j_3}\partial \theta_{j_4}}\right) \bigg/  (4!f(x_i; \btheta_0))\right| < d^4 n^{5/8} \right\}.$$
	On $\mathcal{A}_1 \cap \mathcal{A}_3$ we have 
	$${\rm |II|} < d^{6}\left(\sqrt{\frac{{c}}{L_2}}\right)^4 n^{-7/24}\log^{2}n.$$
	By Lemma  \ref{lem:cont-deri}, it follows that 
	$$P(\mathcal{A}_3) \geq 1-2d^4  R_m \left(\frac{M}{ n^{1/8} }\right)^{2m}.$$ 
	Finally, let $g(x; \btheta_0)$ be  the function defined in Condition (\Cc). 
	Then, for fixed $j_1,\dots,j_5$,  
	$$\sup_{\Vert\btheta - \btheta_{0}\Vert_2 \leq \tau} \left| \left(\frac{\partial^5 f(x; {\bm \theta})}{\partial \theta_{j_1} \partial \theta_{j_2}\partial \theta_{j_3}\partial \theta_{j_4}\partial \theta_{j_5}}\right) /  (5!f(x; \btheta_0)) \right| \leq g(x; \btheta_{0}).$$
	By Lemma \ref{lem:normcontrol}, we have 
	$$\Pr\left(\left|\sum_{i=1}^n g(x_i)\right| \geq n(1+M)\right) \leq 2R_m \left( \frac{4M^2}{n}\right)^{m}.$$
	Let  
	$$\mathcal{A}_4 = \left\{ \left|\sum_{i=1}^{n}d^5 g(x_i) \right| < d^5 n(1+M) \right\}.$$
	Then, we have on $\mathcal{A}_1\cap \mathcal{A}_4$
	$${\rm |III|} < d^{15/2}\left(\sqrt{\frac{{c}}{L_2}}\right)^5 (1+M) n^{-7/48} \log^{5/2}n,$$
	and the probability is at least 
	$$\Pr( \mathcal{A}_4) \geq 1- 2R_m \left( \frac{4M^2}{n}\right)^{m}.$$ 
	In summary, on $\mathcal{A}_1 \cap \mathcal{A}_2 \cap \mathcal{A}_3 \cap \mathcal{A}_4$, we have 
	\begin{align}\label{eq:H0step1}
		\left|\sum_{i=1}^{n}\varepsilon_{in}\right|   &< d^{9/2}\left(\sqrt{\frac{{c}}{L_2}}\right)^3 n^{-1/16}\log^{3/2}n + d^{6}\left(\sqrt{\frac{{c}}{L_2}}\right)^4 n^{-7/24}\log^{2}n \\
		&\quad + d^{15/2}\left(\sqrt{\frac{{c}}{L_2}}\right)^5 (1+M) n^{-7/48} \log^{5/2}n, \notag
	\end{align}
	and 
	\begin{equation}\label{eq:H0step1prob}
		\Pr(\mathcal{A}_2 \cap \mathcal{A}_3 \cap \mathcal{A}_4) \geq 1- 2(d^3+d^4 ) R_m \left(\frac{M}{ n^{1/8} }\right)^{2m}  - 2R_m \left( \frac{4M^2}{n}\right)^{m}. 
	\end{equation}
	\noindent
	{\bf Step 2: } Next, we aim to bound 
	$\left|2\sum_{i=1}^{n} \bar{\m}^{ {\rm T}} \b_i\varepsilon_{in}\right|. $
	By Taylor's expansion to the third order,  $\varepsilon_{in}$ can be accurately represented as 
	$$
	\varepsilon_{in} =  \sum_{j_{1}=1}^d \cdots \sum_{j_{3}=1}^d \sum_{g=1}^G \bar{\alpha}_g \prod_{s=1}^3 \left(\bar{\theta}_{gj_s} - {\theta}_{0j_s}\right) \left(\frac{\partial^3 f(x_i;{\bm \zeta}_g(x_i))}{\partial \theta_{j_1} \partial \theta_{j_2}\partial \theta_{j_3}}\right) /  (3!f(x_i;\btheta_0)), 
	$$
	where ${\bm \zeta}_g(x_i)$ lies between ${\theta}_g$ and $\theta_0$.  
	Then we have 
	\begin{align}\label{eq:epsBound}
		\left|\varepsilon_{in}\right| &\leq \left|\sum_{j_{1}=1}^d \cdots \sum_{j_{3}=1}^d \sum_{g=1}^G \bar{\alpha}_g \prod_{s=1}^3 \left(\bar{\theta}_{gj_s} - {\theta}_{0j_s}\right)\right|  g(x_i; \btheta_{0})  \notag\\
		&\leq (\sqrt{d})^3 \left \|\bar{\bxi} - \bxi_0\right\|_2^3 d^3 g(x_i; \btheta_{0}).
	\end{align}
	Applying the inequality (\ref{eq:epsBound}) and Cauchy's inequality, we have
	$$\left|2 \bar{\m}^{{\rm T}} \b_i  \varepsilon_{in}\right|\leq 2 \left\|\bar{\m}\right\|_2  \left\|\b_i\right\|_2  \left(\sqrt{d}\right)^3 \left\|\bar{\bxi} - \bxi_0\right\|_2^3  d^3g(x_i; \btheta_{0}),$$
	where 
	$$  \left\|\b_i\right\|_2 = \left( \sum_{j=1}^d Y_{ij}^2 + \sum_{j=1}^d Z_{ij}^2 + \sum_{j_1=1}^d\sum_{j_2> j_1}^d U_{ij_1j_2}^2\right)^{1/2}.$$
	Therefore,  we only need to consider the term 
	$ \sum_{i=1}^n \left\|\b_i\right\|_2    g(x_i; \btheta_{0}).$
	By Lemma \ref{lem:normcontrol},   we have 
	$$\Pr\left(\left|\sum_{i=1}^n \left\|  \b_i\right\|_2     g(x_i; \btheta_{0}) \right| \geq n\left(1+(d+1)M^2\right)\right) \leq 2R_m \left( \frac{ 4(d+1)^2M^4}{n}\right)^{m}.$$
	Let 
	$$\mathcal{A}_{5}= \left\{\left|\sum_{i=1}^n \left\|  \b_i\right\|_2     g(x_i; \btheta_{0}) \right|  < n\left(1+(d+1)M^2\right) \right\}.$$
	Therefore, on $\mathcal{A}_1 \cap \mathcal{A}_5$,  using the fact that $L_2 \leq L_1 $, we have 
	\begin{equation} \label{eq:H0step2}
		\left|\sum_{i=1}^n2 \bar{\m}^{{\rm T}} \b_i  \varepsilon_{in}\right|< 2 d^{9/2} \left(1+(d+1)M^2\right) \left({\frac{{c}}{L_2}}\right)^{5/2}n^{-7/48}\log^{5/2}n,	
	\end{equation}
	and 
	\begin{equation}\label{eq:H0step2prob}
		\Pr(\mathcal{A}_5) \geq 1 - 2R_m \left( \frac{ 4(d+1)^2M^4}{n}\right)^{m}.
	\end{equation}
	\noindent
	{\bf Step 3: } Finally, we aim to bound 
	\begin{equation}\label{eq:H0step3all}
		\sum_{i=1}^{n} \delta_i^3 =  \sum_{i=1}^{n} \left\{  \left(\bar{\m}^{ {\rm T}} \b_i\right)^3  + 3 \left(\bar{\m}^{ {\rm T}} \b_i\right)^2\varepsilon_{in} + 3 \left(\bar{\m}^{ {\rm T}} \b_i\right)\varepsilon_{in}^2 +  \varepsilon_{in}^3\right\}.
	\end{equation}
	We first deal with 
	$ \sum_{i=1}^{n} \left|\left(\bar{\m}^{ {\rm T}} \b_i\right)^q\varepsilon_{in}^p\right|,$ where $p,q \in \mathbb{N}$ and $p+q = 3$. 
	Similar to the proof in Step 2,  we have 
	$$\left|\varepsilon_{in}\right|
	\leq d^{9/2}\left\|\bar{\bxi} - \bxi_0\right\|_2^3 g(x_i;\btheta_{0}).$$
	Then, we have 
	\begin{equation*}
		\sum_{i=1}^{n} \left|\left(\bar{\m}^{ {\rm T}} \b_i\right)^q\varepsilon_{in}^p\right| \leq d^{9p/2} \left\|\bar{\bxi} - \bxi_0\right\|_2^{3p}\Vert\bar{\m}\Vert_2^q\sum_{i=1}^{n}  \Vert\b_i\Vert_2^q  g^p(x_i; \btheta_{0}). 
	\end{equation*}
	Let 
	$$\mathcal{B}_{pq}= \left\{ \left|\sum_{i=1}^{n}  \Vert\b_i\Vert_2^q  g^p(x_i; \btheta_{0})\right|  < n(1+(d+1)^{q} M^{3})\right\}.$$
	Then, on $\mathcal{A}_1 \cap \mathcal{B}_{pq}$, we have 
	$$  \sum_{i=1}^{n} \left|\left(\bar{\m}^{ {\rm T}} \b_i\right)^q\varepsilon_{in}^p\right| <  d^{9p/2} \left(1+(d+1)^qM^3\right) \left({\frac{{c}}{L_2}}\right)^{\frac{3p+2q}{2}}n^{\frac{48-11(3p+2q)}{48}}\log^{\frac{3p+2q}{2}}n, $$
	and  by Lemma \ref{lem:normcontrol}, 
	$$\Pr(\mathcal{B}_{pq}) \geq 1 - 2R_m \left(\frac{2(d+1)^{q} M^{3}}{\sqrt{n}}\right)^{2m}.$$
	Hence, we let 
	\begin{equation*}
		\mathcal{A}_6 = \bigcap_{p+q = 3}  \mathcal{B}_{pq}.
	\end{equation*}
	Then, on $\mathcal{A}_1 \cap \mathcal{A}_6$,
	\begin{align} \label{eq:H0step3}
		\left|\sum_{i=1}^{n} \delta_i^3\right| &<  \sum_{p+q = 3} \tbinom{3}{p}   d^{9p/2} \left(1+(d+1)^qM^3\right) \left({\frac{{c}}{L_2}}\right)^{\frac{3p+2q}{2}}n^{\frac{48-11(3p+2q)}{48}}\log^{\frac{3p+2q}{2}}n \notag \\
		&\leq \tilde{C} n^{-\frac{3}{8}}\log^{3}n,
	\end{align}
	where $\tilde{C}>0$ is a constant and  by Lemma \ref{lem:normcontrol}, 
	\begin{equation}\label{eq:H0step3prob}
		\Pr(\mathcal{A}_6) \geq 1-  \left( \sum_{q = 0}^3    2R_m \left(\frac{2(d+1)^{q} M^{3}}{\sqrt{n}}\right)^{2m} \right).
	\end{equation}
	By (\ref{eq:H0step1}, \ref{eq:H0step1prob}, \ref{eq:H0step2}, \ref{eq:H0step2prob}, \ref{eq:H0step3}, \ref{eq:H0step3prob}), there are two constants $C$ and $C_1$ depending on $m,M,d,L_2,c$  such that 
	\begin{equation}\label{eq:H0all1}
		\Pr\left(\left|2\sum_{i=1}^{n} \varepsilon_{in} - 2\sum_{i=1}^{n} \bar{\m}^{ {\rm T}} \b_i\varepsilon_{in} + \frac{2}{3} \sum_{i=1}^{n} \delta_i^3\right| \geq C n^{-1/16}\log^{3/2}n \right) \leq { (C_1n)^{-m/4} }. 
	\end{equation}
	The above inequality (\ref{eq:H0all1}) yields that $\left|2\sum_{i=1}^{n} \varepsilon_{in} - \sum_{i=1}^{n} \bar{\m}^{ {\rm T}} \b_i\varepsilon_{in} + \frac{2}{3} \sum_{i=1}^{n} \delta_i^3\right| $ is sufficiently small with high probability.  
	
	We next  bound
	$$\sum_{i=1}^{n} \bar{\m}^{ {\rm T}} \b_i  \b_i\trans \bar{\m}  =  \bar{\m}^{ {\rm T}}\sum_{i=1}^{n} \left[  \b_i  \b_i \trans\right]  \bar{\m}  =  n \bar{\m}^{\rm T}\sum_{i=1}^{n} \frac{\left[  \b_i  \b_i\trans \right]}{n}  \bar{\m}.$$
	To this end, we only need to bound 
	the set 
	$$\mathcal{A}_7 = \left\{ \left \|\sum_{i=1}^{n} \frac{\left[  \b_i \b_i\trans \right]}{n}-\B\right\|_F < \lambda_{\rm min}/2\right\}.$$
	By the matrix inequality $\Vert\A-\B\Vert_F \geq |\lambda_{\rm min}(\A) - \lambda_{\rm min}(\B)|$, on $\mathcal{A}_7 $, we have 
	$$\lambda_{\rm min}\left(\sum_{i=1}^{n} \frac{\left[  \b_i \b_i\trans \right]}{n}\right) > \lambda_{\rm min}/2.$$
	It follows that  on $\mathcal{A}_7 $,  we have 
	\begin{equation}\label{eq:H0B}
		n \bar{\m}^{{\rm T}}\sum_{i=1}^{n} \frac{\left[  \b_i  \b_i\trans \right]}{n}  \bar{\m} \geq \frac{\lambda_{\rm min}}{2} n \bar{\m}^{{\rm T}}\bar{\m}.
	\end{equation}
	Next we bound the probability of $\mathcal{A}_7$. 
	Since $$\mathcal{A}_7^c \subset \left\{\mbox{There exists one pair } k,l \mbox{ such that } \left|\left( \sum_{i=1}^{n} \frac{\left[  \b_i  \b_i\trans \right]}{n}\right)_{kl} - B_{kl}\right| \geq \lambda_{\rm min}/(2d^2)  \right\}, $$
	we have 
	$$\Pr(\mathcal{A}_7^c) \leq \sum_{k=1}^{d}\sum_{l=1}^{d} \Pr\left(\left|\left( \sum_{i=1}^{n} \frac{\left[  \b_i  \b_i\trans \right]}{n}\right)_{kl} - B_{kl}\right| \geq \lambda_{\rm min}/(2d^2) \right).$$
	For fixed $k,l$ let 
	$T_i =   {\left[  \b_i  \b_i\trans \right]_{kl}}.$
	Therefore, we have 
	$\left\| T_i  \right\|_{L^{2m}} \leq M^2$.
	Since $\Ex \left( {\left[  \b_i  \b_i\trans \right]}\right)_{kl} =  B_{kl}$,  by Lemma \ref{lem:Rosen}, we have 
	$$\Pr\left(\left|\sum_{i=1}^n \left( T_i - B_{kl}\right)\right| \geq n\lambda_{\rm min}/\left(2d^2\right)\right) \leq 2R_m  \left( \frac{4M^4}{n}\right)^{m}  \left( \frac{2d^2}{{\lambda_{\rm min}} }\right)^{2m} .$$
	Thus, we have 
	\begin{equation}\label{eq:H0Bprob}
		\Pr(\mathcal{A}_7) \geq 1-2d^2R_m  \left( \frac{4M^4}{n}\right)^{m}  \left( \frac{2d^2}{{\lambda_{\rm min}} }\right)^{2m} .	
	\end{equation}
	Combining (\ref{eq:R1nexpan}, \ref{eq:H0all1}) with (\ref{eq:H0B}, \ref{eq:H0Bprob}),  we have
	\begin{equation}\label{eq:H0final}
		\Pr\left( R_{1n} \leq 2\sum_{i=1}^{n} \bar{\m}^{ {\rm T}} \b_i -\frac{\lambda_{\rm min}}{2} n \bar{\m}^{{\rm T}}\bar{\m} +C n^{-1/16}\log^{3/2}n   \right) \geq 1-  { (C_1n)^{-m/4} },
	\end{equation}
	where $C, C_1$ are another two constants.
	It is clear that  
	\begin{align*}
		R_{1n} &\leq  2\bar{\m}^{ {\rm T}} \sum_{i=1}^n\b_i - \frac{\lambda_{\rm min}}{2} n \bar{\m}^{ {\rm T}}\bar{\m} +C n^{-1/16}\log^{3/2}n  \\
		&\leq \frac{2(\sum_{i=1}^n \b_i)\trans(\sum_{i=1}^n \b_i)}{\lambda_{\rm min} n} + C n^{-1/16}\log^{3/2}n \\
		&= \frac{2}{\lambda_{\rm min} n} \sum_{j=1}^{d(d+3)/2} \left(\sum_{i=1}^nb_{ij} \right) ^2 + C n^{-1/16}\log^{3/2}n  . 
	\end{align*}
	Therefore, let 
	$$T_n = \frac{2}{\lambda_{\rm min} n} \sum_{j=1}^{d(d+3)/2} \left(\sum_{i=1}^nb_{ij} \right) ^2.$$ For any $t>0$, we have 
	\begin{align*}
		\{ T_n  \geq t \} &= \left\{  \sum_{j=1}^{d(d+3)/2} \left(\sum_{i=1}^nb_{ij} \right) ^2 \bigg/n \geq \frac{\lambda_{\rm min}}{2}t \right\} \\ 
		&\subset \left\{ \mbox{There exists } j \mbox{ such that } \left(\sum_{i=1}^nb_{ij} \right) ^2\bigg/n \geq    \frac{\lambda_{\rm min}}{d(d+3)} t \right\}.
	\end{align*}
	It follows that  
	$$\Pr(T_n  \geq t) \leq  \sum_{j=1}^{d(d+3)/2} \Pr \left(\frac{\sum_{i=1}^nb_{ij}}{\sqrt{n}} \geq \left(\frac{\lambda_{\rm min}}{d(d+3)} t\right)^{1/2}\right).$$
	For any fixed $j$, by Lemma \ref{lem:Rosen}, we have 
	$$\Pr\left(\frac{\sum_{i=1}^nb_{ij}}{\sqrt{n}} \geq \left(\frac{\lambda_{\rm min}}{d(d+3)} t\right)^{1/2}\right) \leq 2R_m \left(\frac{M^2d(d+3)}{\lambda_{\rm min}} \right)^{m} t^{-m}.$$
	Thus, we have 
	\begin{equation}\label{eq:H0Tn}
		\Pr(T_n  \geq t) \leq  d(d+3) R_m \left(\frac{M^2d(d+3)}{\lambda_{\rm min}} \right)^{m} t^{-m}.	
	\end{equation}
	Combing (\ref{eq:H0final})  with (\ref{eq:H0Tn}), we conclude that 
	$$\Pr(R_{1n} \leq t+C n^{-1/16}\log^{3/2}n) \geq  1-  { (C_1n)^{-m/4} } -  {(C_2t)}^{-m} ,$$
	where $C, C_1,C_2 $ are three constants.
	It follows that 
	$$\Pr({\rm EM}_n^{(K)}  \leq t+C n^{-1/16}\log^{3/2}n) \geq  1-  { (C_1n)^{-m/4} } -  {(C_2t)}^{-m} ,$$
	and thus we prove the theorem.
\end{proof}

\begin{proof}[Proof of Lemma \ref{lem:Rosen}]
	By Exercise 2.20 in \cite{wainwright2019high}, under the stated conditions, there is a universal constant $R_m$ such that  
	$$
	\mathbb{E}\left[\left(\sum_{i=1}^{n} X_{i}\right)^{2 m}\right] \leq R_{m}\left\{\sum_{i=1}^{n} \mathbb{E}\left[X_{i}^{2 m}\right]+\left(\sum_{i=1}^{n} \mathbb{E}\left[X_{i}^{2}\right]\right)^{m}\right\}.
	$$
	By Lyapunov's inequality, we have $\Vert X_i\Vert_{L^2} \leq \Vert X_i\Vert_{L^{2m}} \leq C$.
	By Markov's inequality, we have
	\begin{align*}
		\Pr \left\{\left|\sum_{i=1}^n X_i \right| \geq n\delta  \right \} &\leq  \frac{\Ex \left(  |\sum_{i=1}^n X_i|^{2m}  \right)}{\left( n\delta\right)^{2m}} \\
		&\leq\frac{  R_{m}}{\left( n\delta\right)^{2m}} \left\{\sum_{i=1}^{n} \mathbb{E}\left[X_{i}^{2 m}\right]+\left(\sum_{i=1}^{n} \mathbb{E}\left[X_{i}^{2}\right]\right)^{m}\right\} \\
		&\leq \frac{  R_{m}}{\left( n\delta\right)^{2m}} \left( n C^{2m} + n^m C^{2m}\right) \\
		&\leq  \frac{  2R_{m}} {\left( n\delta\right)^{2m}}  n^m C^{2m} =  2R_m \left(\frac{C}{\sqrt{n}\delta}\right)^{2m}.
	\end{align*}
	The second conclusion is  a direct corollary of the first   conclusion, and thus we complete the proof.
\end{proof}
\begin{proof}[Proof of Lemma \ref{lem:normcontrol}]
	We first prove that when $p+q  \leq 3 $,  we have
	\begin{equation}\label{eq:normcontrolcon}
		\left\|  \left\|\b_i\right\|_2^q  g^p(x_i; \btheta_{0}) \right\|_{L^{2m}} \leq (d+1)^{q} M^{p+q} .
	\end{equation}
	When $p,q < 3$, by Cauchy's inequality,  we have
	$$\left\|  \left\|\b_i\right\|_2^q  g^p(x_i; \btheta_{0}) \right\|_{L^{2m}} \leq  \left\|\left\|\b_i\right\|_2^q \right\|_{L^{4m}}  \|g^p(x_i; \btheta_{0}) \|_{L^{4m}}.$$
	By the triangle inequality and Condition (\Cc), we conclude that 
	\begin{align}\label{eq:normcontrol1}
		\left\|\left\|\b_i\right\|_2^q \right\|_{L^{4m}}   &=    \left\| \left( \sum_{j=1}^d Y_{ij}^2 + \sum_{j=1}^d Z_{ij}^2 + \sum_{j_1=1}^d\sum_{j_2>j_1}^d U_{ij_1j_2}^2\right)^{q/2}\right\|_{L^{4m}} \notag \\
		&=   \left(  \left\| \sum_{j=1}^d Y_{ij}^2 + \sum_{j=1}^d Z_{ij}^2 + \sum_{j_1=1}^d\sum_{j_2>j_1}^d U_{ij_1j_2}^2\right\|_{L^{2mq}}\right)^{q/2} \notag\\
		&\leq \left(  \sum_{j=1}^d \left\| Y_{ij}^2 \right\|_{L^{2mq}}+ \sum_{j=1}^d  \left\|Z_{ij}^2\right\|_{L^{2mq}} + \sum_{j_1=1}^d\sum_{j_2>j_1}^d \left\| U_{ij_1j_2}^2 \right\|_{L^{2mq}}\right)^{q/2}  \notag\\
		&= \left(  \sum_{j=1}^d \left\| Y_{ij} \right\|_{L^{4mq}}^2+ \sum_{j=1}^d  \left\|Z_{ij}\right\|_{L^{4mq}}^2 + \sum_{j_1=1}^d\sum_{j_2>j_1}^d \left\| U_{ij_1j_2}^2 \right\|_{L^{4mq}}^2\right)^{q/2}  \notag \\
		&\leq (d+1)^{q} M^{q},
	\end{align}
	where the last inequality is from the fact $q\leq 2$ and $d(d+3)/2 \leq (d+1)^2$.  
	By $p\leq 2$ and Condition (\Cc), similarly, we conclude that  
	$$
	\|g^p(x_i; \btheta_{0}) \|_{L^{4m}} = (\|g(x_i; \btheta_{0}) \|_{L^{4mp}})^p \leq M^p. 
	$$
	Thus, we prove (\ref{eq:normcontrolcon}) when $p,q < 3$. 
	When $p=0, q=3$, analysis similar  to that in (\ref{eq:normcontrol1}) shows that   
	$$\left\|  \left\|\b_i\right\|_2^q \right\|_{L^{2m}} \leq (d+1)^{q} M^{q}.$$
	Similarly, when $p=3,q=0$ we have 
	$$
	\|g^p(x_i; \btheta_{0}) \|_{L^{2m}} = (\|g(x_i; \btheta_{0}) \|_{L^{2mp}})^p \leq M^p. 
	$$
	Thus, we prove that for any $p+q \leq 3$, (\ref{eq:normcontrolcon}) holds.
	
	Next,  by Lemma \ref{lem:Rosen}, we have
	$$\Pr \left\{\left|\sum_{i=1}^n \left\{R_{pq}(x_i) - \Ex(R_{pq}(x_i))\right\}\right| \geq n   \right\}\leq  2R_m \left(\frac{2(d+1)^{q} M^{p+q}}{\sqrt{n}}\right)^{2m}.$$
	By Lyapunov's inequality, we have 
	$$
	\left\|  \left\|\b_i\right\|_2^q  g^p(x_i; \btheta_{0}) \right\|_{L^{1}}	\leq \left\|  \left\|\b_i\right\|_2^q  g^p(x_i) \right\|_{L^{2m}} \leq (d+1)^{q} M^{p+q} .
	$$
	It follows that 
	$$\Pr \left\{\left|\sum_{i=1}^n R_{pq}(x_i)\right| \geq n(1+(d+1)^{q} M^{p+q})   \right\}\leq  2R_m \left(\frac{2(d+1)^{q} M^{p+q}}{\sqrt{n}}\right)^{2m},$$
	and thus we complete the proof.
\end{proof}
\begin{proof}[Proof of Lemma \ref{lem:cont-deri}]
	Note that $\Ex \left[ D_i(j_1, \dots, j_k) \right] = 0.$ 
	By Condition (\Cc), we have 
	$$\left\| D_i(j_1, \dots, j_k)\right\|_{L^{2m}} \leq M.$$  By Lemma \ref{lem:Rosen}, we have
	$$\Pr\left\{\left|\sum_{i=1}^n  D_i(j_1, \dots, j_k) \right| < n t  \right\}\geq 1-  2R_m  \left(\frac{M}{\sqrt{n}t}\right)^{2m}, \ \ \mbox{for all } t>0.$$ 
	Taking $t = n^{-3/8}$, we obtain the tail probability bound  as
	$$\Pr\left\{\left|\sum_{i=1}^n D_i(j_1, \dots, j_k) \right| \geq n^{5/8}  \right\}\leq 2R_m  \left(\frac{M}{ n^{1/8} }\right)^{2m}.$$
	Observe that 
	$$
	\left\{\sum_{j_{1}=1}^d \cdots \sum_{j_{k}=1}^d\left|\sum_{i=1}^n D_i(j_1, \dots, j_k) \right| \geq d^kn^{5/8}  \right\} \subset \bigcup_{j_{1}=1}^d \cdots \bigcup_{j_{k}=1}^d\left\{\left|\sum_{i=1}^n D_i(j_1, \dots, j_k) \right| \geq n^{5/8}  \right\}
	$$
	It follows that 
	$$\Pr\left\{\sum_{j_{1}=1}^d \cdots \sum_{j_{k}=1}^d\left|\sum_{i=1}^n D_i(j_1, \dots, j_k) \right| < d^kn^{5/8}  \right\}\geq 1- 2d^kR_m  \left(\frac{M}{ n^{1/8} }\right)^{2m},$$
	and thus we complete the proof.
\end{proof}

\subsection{Proofs of Theorem \thmoneRes}
We abbreviate $pl_n\left( \hat{\bxi}_0,\ba_0\right)$ and $l_n\left( \hat{\bxi}_0,\ba_0\right)$ to $pl_n\left(\hat{\bxi}_0\right)$ and $l_n\left(\hat{\bxi}_0\right)$, respectively. 
Recall that 
\begin{equation}\label{eq:theta0dag}
	{\btheta}^{\dagger}_0 = \argmax_{\btheta \in {\Theta}} \Ex_{\balpha^*,\bxi^{*}}  \left[\log f(x;\btheta)\right],
\end{equation}
and
\begin{equation}\label{eq:xidag}
	{\bxi}^{\dagger} = \argmax_{\bxi \in \Xi} \Ex_{\balpha^*,\bxi^{*}}  \left[  \log \ \varphi\left(x;\bxi, \ba^{(0)}\right) \right].
\end{equation}

We briefly describe the proof of Theorem \thmoneRes. Observe that the EM-test statistic is larger than the penalized log-likelihood ratio  $pl_n\left(\bxi^{(0)}, \ba^{(0)}\right) - pl_n\left(\hat{\bxi}_{0}, \ba_0\right)$, which can be decomposed as a summation of three parts, $pl_n\left(\bxi^{(0)}, \ba^{(0)}\right) -  pl_n\left({\bxi}^{\dagger}, \ba^{(0)}\right)$,  
$pl_n\left({\bxi}^{\dagger}, \ba^{(0)}\right)-  pl_n\left({\bxi}_{0}^{\dagger}, \ba_{0}\right)$
and $pl_n\left({\bxi}_{0}^{\dagger}, \ba_{0}\right) - pl_n\left(\hat{\bxi}_{0}, \ba_0\right)$. 
All three parts can be bounded. The first part is non-negative.  
The second part can be written as 
$$pl_n\left({\bxi}^{\dagger}, \ba^{(0)}\right)-  pl_n\left({\bxi}_{0}^{\dagger}, \ba_{0}\right) = \sum_{i=1}^n R(x_i;\bxi^*) + p(\ba^{(0)}) -  p(\ba_{0}),$$
and can be bounded using the Bernstein inequality. 
For the third part, since $ D(\btheta )\leq  D\left({\btheta}^{\dagger}_0\right) $ for all $\btheta \in \Theta$, we have 
$$pl_n\left({\bxi}_{0}^{\dagger}, \ba_{0}\right) - pl_n\left(\hat{\bxi}_{0}, \ba_0\right) =  \sum_{i=1}^n \left\lbrace  \log \ f(x_i;\hat{\btheta}_0) - \log \ f\left(x_i;{\btheta}^{\dagger}_0\right) \right\rbrace  \geq  -n^{1/2}\sup_{\btheta \in \Theta} Z_{\btheta}(\bxi^*). $$
Thus, the third part can be bounded by analyzing the supremum of the empirical process $\{Z_{\btheta}(\bxi^*), \btheta \in \Theta\}$ using the generalized Dudley inequality.

We first give two technical lemmas.
\begin{lemma} \label{lem:H1first}
	Under Condition (\Cf) and (\Cg), for every $t\geq 0$, we have, 
	$$\Pr\left(l_n\left({\bxi}^{\dagger},\balpha^{\left(0\right)}\right) -l_n\left({\bxi}_0^{\dagger}\right) \geq n\varrho - t\right)  \geq 1-2\exp\left[-C' \min\left(\frac{t^2}{n M_{\psi_1}^2} , \frac{t}{M_{\psi_1}}\right) \right],$$
	where $C'$ is a constant and $\bxi^{\dagger}$ and $\bxi_{0}^{\dagger} = \left(\btheta_{0}^{\dagger}, \dots, \btheta_{0}^{\dagger}\right)$ are defined in (\ref{eq:xidag}) and (\ref{eq:theta0dag}).
\end{lemma}

Let $\rho(\btheta, \btheta')  = C_{\rho}\Vert \btheta - \btheta' \Vert_2$, where $C_{\rho}$ is in Condition (\Ch). 
Let $\mathcal{N}(u , \Theta, \rho)$ be the covering number, which is the smallest number of closed balls with centers in $\Theta$ and
radius  $u$ whose union covers $\Theta$.  Next we define the generalized Dudley integral as 
$$J(D) = \int_{0}^{D} \log \left(1 + \mathcal{N}(u , \Theta, \rho)\right) {\rm d}u. $$
where $D = \sup_{\btheta, \btheta' \in \Theta} \Vert \btheta - \btheta' \Vert_2$ is the  Euclidean diameter. 
Note that $\Theta $ is a compact set and $\rho(\btheta, \btheta')  = C_{\rho}\Vert \btheta - \btheta' \Vert_2$. Therefore,  $J(D) < \infty$, and we have the following generalized Dudley inequality by the chaining method. 
The proof of the classic Dudley's inequality can be found in \cite{vershynin2018high}.  
The proof of the following lemma follows the same arguments by a chaining method and can be found in \cite{wainwright2019high} (Theorem 5.36). Thus, we omit the proof.
\begin{lemma}\label{lem:Dudley}
	Under Condition (\Ch), for any $t$,
	$$\Pr \left(\sup_{\btheta, \btheta' \in \Theta} |Z_{\btheta}(\bxi^*) -Z_{\btheta'}(\bxi^*) |\geq C_J[J(D) + t]\right) \leq 2\exp\left(\frac{-t}{D}\right),$$
	where $C_J$ is a constant, $D$ is the diameter and $J(D)$ is the generalized Dudley integral.
\end{lemma}
\begin{proof}[Proof of Theorem \thmoneRes]
	We first aim to bound the probability 
	$$\Pr \left(pl_n \left( \bxi^{(0)},\balpha^{(0)}\right) - pl_n\left(\hat{\bxi}_0\right) \geq 2^{-1}n\varrho  -  n^{1/2}C_J\left[J\left(D\right) +t\right] -   p_0\right).$$
	Since $pl_n \left( \bxi^{(0)},\balpha^{(0)}\right) = l_n\left( \bxi^{(0)},\balpha^{(0)}\right) + p\left(\balpha^{(0)}\right)$ and $ p\left(\balpha^{(0)}\right)\geq p_0$,  we have 
	$$pl_n \left( \bxi^{(0)},\balpha^{(0)}\right) - pl_n\left(\hat{\bxi}_0\right) \geq 
	l_n\left(\bxi^{(0)},\balpha^{(0)}\right) - l_n\left(\hat{\bxi}_0\right) + p_0.$$
	Thus, we only need to control 
	$$\Pr \left(l_n\left(\bxi^{(0)},\balpha^{(0)}\right) - l_n\left(\hat{\bxi}_0\right) \geq2^{-1}n\varrho  -  n^{1/2}C_J\left[J\left(D\right) +t\right] \right).
	$$ 
	Since 
	\begin{align*}
		l_n\left(\bxi^{(0)},\balpha^{(0)}\right) - l_n\left(\hat{\bxi}_0\right) 
		&= l_n\left(\bxi^{(0)},\balpha^{(0)}\right)- l_n\left({\bxi}^{\dagger},\balpha^{(0)}\right) \\ &+ 
		l_n\left({\bxi}^{\dagger},\balpha^{(0)}\right) -l_n\left({\bxi}_0^{\dagger}\right) +l_n
		\left({\bxi}_0^{\dagger}\right) - l_n\left(\hat{\bxi}_0\right)
	\end{align*}
	and 
	$l_n\left(\bxi^{(0)},\balpha^{(0)} \right)- l_n\left({\bxi}^{\dagger},\balpha^{(0)}\right) > 0$, we have 
	\begin{align*}
		&\quad \Pr\left(l_n\left(\bxi^{(0)},\balpha^{\left(0\right)}\right) - l_n\left(\hat{\bxi}_0\right) \geq 2^{-1}n\varrho  -  n^{1/2}C_J\left[J\left(D\right) +t\right] \right)  \\
		&\geq \Pr\left(l_n\left({\bxi}^{\dagger},\balpha^{\left(0\right)}\right) -l_n\left({\bxi}_0^{\dagger}\right) +l_n\left({\bxi}_0^{\dagger}\right) - l_n\left(\hat{\bxi}_0\right) \geq 2^{-1}n\varrho  -  n^{1/2}C_J\left[J\left(D\right) +t\right]\right) \\
		&\geq 1 - \Pr\left(l_n\left({\bxi}^{\dagger},\balpha^{\left(0\right)}\right) -l_n\left({\bxi}_0^{\dagger}\right) < 2^{-1}n\varrho \right) - \Pr\left(l_n\left({\bxi}_0^{\dagger}\right) - l_n\left(\hat{\bxi}_0\right) < -n^{1/2}C_J\left[J\left(D\right) +t\right]\right).
	\end{align*}
	Thus, we divide the remaining proof into two steps. 
	
	\noindent
	{\bf Step 1.} 
	We aim to bound
	$ \Pr\left(l_n\left({\bxi}^{\dagger},\balpha^{\left(0\right)}\right) -l_n\left({\bxi}_0^{\dagger}\right)  \geq 2^{-1}n\varrho\right).$
	Applying Lemma \ref{lem:H1first} and 
	taking $t = n\varrho / 2$ yield 
	\begin{equation}\label{H11}
		\Pr \left(l_n\left({\bxi}^{\dagger},\balpha^{\left(0\right)}\right) -l_n\left({\bxi}_0^{\dagger}\right)\geq 2^{-1}n\varrho \right)  \geq 1-2\exp\left[-c \min\left(\frac{n\varrho^2}{4 M_{\psi_1}^2} , \frac{n\varrho}{2M_{\psi_1}}\right) \right].
	\end{equation}
	
	\noindent
	{\bf Step 2.} We aim to bound
	$\Pr\left(l_n\left({\bxi}_0^{\dagger}\right) - l_n\left(\hat{\bxi}_0\right) < -n^{1/2}C_J\left[J\left(D\right) +t\right]\right).$
	We can write $l_n\left({\bxi}_0^{\dagger}\right)$ as $l_n\left({\btheta}_0^{\dagger} \right)$ because ${\bxi}_0^{\dagger} = \left({\btheta}_0^{\dagger} , \ldots, {\btheta}_0^{\dagger}\right) $. 
	The major difficulty  to control the second term is the randomness of $\hat{\btheta}_0$.  To deal with it, we  note that 
	\begin{align*}
		&\quad	\Pr\left(l_n\left({\btheta}_0^{\dagger}\right) - l_n\left(\hat{\btheta}_0\right) \geq -n^{1/2}C_J\left[J\left(D\right) +t\right]\right)  \\
		&\geq \Pr\left( \inf_{\btheta \in \Theta} \left\lbrace l_n\left({\btheta}_0^{\dagger}\right) - l_n\left(\btheta \right) \right\rbrace \geq -n^{1/2}C_J\left[J\left(D\right) +t\right]\right). 
	\end{align*}
	Thus, we turn to control 
	the  probability of  $\left \lbrace\inf_{\btheta \in \bTheta} \left\lbrace l_n\left({\btheta}_0^{\dagger}\right) - l_n\left(\btheta \right) \right\rbrace \geq -n^{1/2}C_J\left[J\left(D\right) +t\right] \right\rbrace$.  It is equivalent to controlling 
	the  probability of 
	$\left \lbrace\sup_{\btheta \in \bTheta} \left\lbrace l_n\left(\btheta \right)  - l_n\left({\btheta}_0^{\dagger}\right) \right\rbrace \leq n^{1/2}C_J\left[J\left(D\right) +t\right]  \right\rbrace$.  
	
	Let $\btheta' = {\btheta}_0^{\dagger}$. It follows that  $Z_{\btheta'}(\bxi^*) = 0$.  By Lemma \ref{lem:Dudley}, we have 
	$$\Pr \left(\sup_{\btheta \in \Theta} |Z_{\btheta}(\bxi^*) |\geq C_J[J(D) + t]\right) \leq 2\exp\left(\frac{-t}{D}\right).$$ 
	Plugging $Z_{\btheta}$ into the above inequality, we have
	$$\Pr \left(\sup_{\btheta \in \Theta} \left|  n^{-1/2} \left\lbrace l_n(\btheta ) - l_n({\btheta}_0^{\dagger}) - \left(D(\btheta )- D\left({\btheta}_0^{\dagger}\right) \right)  \right\rbrace \right| \geq C_J[J(D) + t]\right) \leq 2\exp\left(\frac{-t}{D}\right),$$
	where $D(\btheta ) = \Ex_{\ba^*, \bxi^*}[ \log f(x; \btheta)]$.
	Since $D(\btheta )- D({\btheta}_0^{\dagger}) \leq 0$, we have
	$$\Pr \left(\sup_{\btheta \in \Theta}   n^{-1/2} \left\lbrace l_n(\btheta ) - l_n\left({\btheta}_0^{\dagger}\right)  \right\rbrace\geq C_J[J(D) + t]\right) \leq 2\exp\left(\frac{-t}{D}\right).$$
	That is 
	\begin{equation}\label{H12}
		\Pr \left(\sup_{\btheta \in \Theta}    \left\lbrace l_n(\btheta ) - l_n\left({\btheta}_0^{\dagger}\right)  \right\rbrace\leq  n^{1/2}C_J[J(D) + t]\right) \geq 1-2\exp\left(\frac{-t}{D}\right).
	\end{equation}
	Combining (\ref{H11}), (\ref{H12}) with the likelihood non-decreasing property of EM yields Theorem  \thmoneRes. 
\end{proof}
From Theorem \thmoneRes, we can prove Corollary \corHone. 

\begin{proof}[Proof of Corollary \corHone]
	Write 
	$$
	t = \frac{2^{-1}n^{1/2}\varrho - n^{\vartheta - 1/2} - p_0 n^{-1/2}}{C_J} - J_D.
	$$
	Then, we have 
	$$2^{-1}n\varrho  -  n^{1/2}C_J\left[J\left(D\right) +t\right] -   p_0 = n^{\vartheta}.$$
	By Theorem \thmoneRes, we have 
	\begin{align*}
		\Pr \left( {\rm EM}_n^{\left(K\right)} \geq n^{\vartheta} \right) &\geq 1 - 2\exp\left[-C' \min\left(\frac{n\varrho^2}{4 M_{\psi_1}^2} , \frac{n\varrho}{2M_{\psi_1}}\right) \right] \\
		& -    2\exp\left(D^{-1}\frac{-2^{-1}n^{1/2}\varrho + n^{\vartheta - 1/2} + p_0 n^{-1/2}}{C_J} + D^{-1}J_D\right).
	\end{align*}
	Therefore, we can find two constants $C_3, C_4 > 0$ such that 
	$$ 
	\Pr\left( {\rm EM}_n^{\left(K\right)} \geq t_n \right) \geq 1 -  \exp\left(-C_3 n^{1/2}  + C_4 n^{\vartheta - 1/2}\right), $$
	and complete the proof.
\end{proof}

\begin{proof}[Proof of Lemma \ref{lem:H1first}]
	Recall that 
	$$R(x_i;\bxi^*) =  \log\left(\varphi\left(x_i;{\bxi}^{\dagger}, \ba^{(0)}\right) \right)  -\log\left(f\left(x_i;{\btheta}_0^{\dagger}\right)\right) .$$
	Since 
	$$l_n\left({\bxi}^{\dagger},\balpha^{\left(0\right)}\right) -l_n\left({\bxi}_0^{\dagger}\right) = \sum_{i=1}^{n}
	\left\lbrace \log\left(\varphi\left(x_i;{\bxi}^{\dagger}, \ba^{(0)}\right) \right) -\log f\left(x_i;{\btheta}_0^{\dagger}\right)\right\rbrace, 
	$$
	we have
	$$l_n\left({\bxi}^{\dagger},\balpha^{\left(0\right)}\right) -l_n\left({\bxi}_0^{\dagger}\right) = \sum_{i=1}^{n} R(x_i;\bxi^*).$$
	By Bernstein's inequality in \cite{vershynin2018high} (Theorem 2.8.1)  and Condition (\Cg),  for every $t\geq0$, we have
	$$\Pr\left(\left|\sum_{i=1}^n R(x_i;\bxi^*) - \Ex\left(R(x_i;\bxi^*)\right) \right|\leq t\right) \geq 1-2\exp\left[-C' \min\left(\frac{t^2}{n M_{\psi_1}^2} , \frac{t}{M_{\psi_1}}\right) \right],$$
	where $C'>0$ is a universal constant.   
	By Condition (\Cf),  we conclude that $\Ex \left(R(x_i;\bxi^*)\right) \geq \varrho$. It follows that  
	$$\Pr \left(\sum_{i=1}^n R(x_i;\bxi^*) \geq n\varrho - t\right) \geq 1-2\exp\left[-C' \min\left(\frac{t^2}{n M_{\psi_1}^2} , \frac{t}{M_{\psi_1}}\right) \right].$$ 
	That is 
	$$\Pr \left(l_n\left({\bxi}^{\dagger},\balpha^{\left(0\right)}\right) -l_n\left({\bxi}_0^{\dagger}\right) \geq n\varrho - t\right)  \geq 1-2\exp\left[-C' \min\left(\frac{t^2}{n M_{\psi_1}^2} , \frac{t}{M_{\psi_1}}\right) \right],$$
	which proves the lemma.
\end{proof}

\subsection{Proofs of Theorem \thmmain} \label{subsec:proofofmain}
\begin{proof}[Proof of Theorem \thmmain]
	Observe that 
	$$\{ S_1  \subset \hat{S}_1(t_n)\}  = \left\{ {\rm EM}_{nj}^{(K)} \geq t_n , \mbox{ for all } j \in S_1 \right \}.$$
	We have 
	$$\Pr\left( S_1  \subset \hat{S}_1(t_n)\right) \geq 1- \sum_{j \in S_1} \Pr \left(  {\rm EM}_{nj}^{(K)} < t_n \right).$$
	By Corollary \corHone, we have 
	\begin{equation}\label{eq:main1}
		\Pr\left( S_1  \subset \hat{S}_1(t_n)\right) \geq 1- s  \exp\left(-C_3 n^{1/2}  + C_4 n^{\vartheta - 1/2}\right),
	\end{equation}
	and the first inequality is proved. Next, observing that 
	$$\left\{ S_1  =  \hat{S}_1(t_n)\right\}  = \left\{ S_1  \subset \hat{S}_1(t_n)\right\}   \cap \left\{ {\rm EM}_{nj}^{(K)} <  t_n , \mbox{ for all } j \in S_0 \right\},$$
	using Theorem \thmzeroRes,  we have
	\begin{equation}\label{eq:main0}
		\Pr\left( \left\{ {\rm EM}_{nj}^{(K)} <  t_n , \mbox{ for all } j \in S_0 \right\} \right) \geq 1- (p-s) \left((C_1 n)^{-m/4} + (C_2n)^{-\vartheta m} \right) .
	\end{equation}
	Combining (\ref{eq:main0}) with (\ref{eq:main1}) yields the second result. 
\end{proof}
\section{Proofs of the asymptotic results}\label{suppA}
We first derive
the upper bound of $\left\|\bxi^{(k)} - \bxi_0\right\|_2$.  Namely, we provide the following results. 
\begin{theorem}\label{rate}
	Assume that $x_1,\dots,x_n$ are independent samples from the homogeneous distribution $f(x;\btheta_0)$. Under Condition (\Ca)--(\Cb) and (\WCc)--(\WCd),  given any initial value $\ba^{(0)} \in \mathbb{S}^{G-1}$, for any fixed $K>0$ and $1 \leq k \leq K$, we have 
	$\left\|\balpha^{(k)} - \balpha^{(0)}\right\|_2 = o_p(1), \  \left\|\bxi^{(k)} - \bxi_0\right\|_2 = O_p\left(n^{-1/4}\right)$
	and 
	$\left\|\sum_{g=1}^G {\alpha}_g^{(k)} \left({\btheta}_g^{(k)} - \btheta_0\right)\right\|_2 = O_p\left(n^{-1/2}\right).$
\end{theorem}
%\begin{remark}
Theorem \ref{rate} says that the convergence rate of $ \bxi^{(k)}$ under the homogeneous model $\mathbb{H}_0$ is only $O_p\left(n^{-1/4}\right)$, but not the common convergence rate $O_p\left(n^{-1/2}\right)$. The reason is that under $\mathbb{H}_0$,  the heterogeneous model  is unidentifiable, and, in consequence, the Fisher information matrix is not positive definite. However, the weighted average of $\bxi^{(k)} $, $\sum_{g=1}^G {\alpha}_g^{(k)} \left({\btheta}_g^{(k)} - \btheta_0\right) $, is a $\sqrt{n}$-consistent estimator. 
%	Before presenting the asymptotic distribution of the EM-test statistic, we first introduce a few notations. 

\subsection{Proofs of Theorem \thmone} 
In this subsection,  we give the proof of Theorem \thmone. 
To prove Theorem \thmone, we only need to prove the following lemmas. 
\begin{lemma}[Consistency] \label{lem:consisasy}
	Assume that $x_1,\dots,x_n$ are independent samples from the homogeneous distribution $f(x;\btheta_0)$. Let $\left(  \bar{\bxi},\bar{\balpha}   \right) $ be an estimator of the parameters in the heterogeneous model $\varphi(x;\bxi, \balpha)  = \sum_{g=1}^{G}\alpha_{g} f(x;\btheta_{g})$ such that $\eta \leq \bar{\alpha}_g \leq 1$ for some  $\eta \in (0,0.5] $. Assume that there exists a constant $c$ such that for any $n$ 
	$$l_n\left(  \bar{\bxi},\bar{\balpha}   \right) - l_n\left(  {\bxi}_0,{\balpha_0}   \right) \geq c > -\infty.$$
	Then, under Condition (\Ca)--(\Cb) and (\WCc)--(\WCd), we have $\left\|\bar{\bxi} -{\bxi}_0\right\|_2 = o_p(1) $.
\end{lemma}
\begin{lemma}[Convergence rate] \label{lem:rateasy}
	Assume that $x_1,\dots,x_n$ are independent samples from the homogeneous distribution $f(x;\btheta_0)$. Let $\left(  \bar{\bxi},\bar{\balpha}   \right) $ be an estimator of the parameters in the heterogeneous model $\varphi(x;\bxi, \balpha)  = \sum_{g=1}^{G}\alpha_{g} f(x;\btheta_{g})$ such that $\eta \leq \bar{\alpha}_g \leq 1$ for some $\eta \in (0,0.5] $. Assume that there exists a constant $c$, such that for any  $n$,  
	$$pl_n\left(  \bar{\bxi},\bar{\balpha}   \right) - pl_n\left(  {\bxi}_0,{\balpha_0}   \right) \geq c > - \infty.$$
	Then, under Condition (\Ca)--(\Cb) and (\WCc)--(\WCd) we have 
	$$\left\|\bar{\bxi} -{\bxi}_0\right\|_2 = O_p(n^{-1/4}) ,  \mbox{ and }  \Vert\bar{\m}_1\Vert_2 = \left\|\sum_{g=1}^G \bar{\alpha}_g (\bar{\btheta}_g - \btheta_0)\right\|_2 = O_p(n^{-1/2}),$$
	where $\bar{\m}_1 = \m_1(\bar{\ba}, \bar{\bxi}, \bxi_0)$ and $\m_1$ is defined in (\ref{eq:m1}).
\end{lemma}
Given an estimator $\left(  \bar{\bxi},\bar{\balpha}   \right) $, define 
$\bar{w}_{gi} = {\bar{\alpha}_g f(x_i;  \bar{\btheta}_g)}/{\varphi\left(x_i;\bar{\balpha}, \bar{\bxi}\right)}$ and $\bar{\ba}^{(1)} = \sum_{i=1}^n \frac{\bar{w}_{gi} + \lambda}{n+G\lambda}$ is the one-step EM update of $\balpha$. 
The following lemma states that under $\mathbb{H}_0$, the EM-update of $\balpha$ does not change much. 
\begin{lemma}\label{lem:asymalpha}
	Assume $\bar{\ba} - \ba^{(0)}= o_p(1)$.  Then, under the same conditions as in Lemma \ref{lem:rateasy}, we have $\bar{\ba}^{(1)} - \ba^{(0)} = o_p(1)$.
\end{lemma}
Combining Lemma \ref{lem:rateasy}, \ref{lem:asymalpha} and the likelihood non-decreasing property of the EM algorithm, we prove Theorem \thmone.  
\begin{proof}[Proof of Lemma \ref{lem:consisasy}]
	Since $\Xi$ is compact, the conclusion can be easily proved using the classical Wald's consistency Theorem \citep{van2000asymptotic}.
\end{proof}

\begin{proof}[Proof of Lemma \ref{lem:rateasy}]
	Let $	R_{1n}(\bar{\bxi}, \bar{\ba}) = 2 \left\{pl_n( \bar{\bxi},\bar{\balpha} ) - pl_n( \bxi_0, \ba_0)\right\}$, where $\bar{\bxi} = (\bar{\btheta}_1, \ldots, \bar{\btheta}_G)$. Since $p(\ba)$ is maximized at $\ba_0$, we have
	\begin{align}\label{R1-inequality}
		R_{1n} &\leq 2 \left\{l_n( \bar{\bxi},\bar{\balpha} ) - l_n( \bxi_0, \ba_0)\right\}  \nonumber\\
		&= 2\sum_{i=1}^n \log\left(1+\sum_{g=1}^G \bar{\alpha}_g\left(\frac{f(x_i;   \bar{\btheta}_g)}{f(x_i;   {\btheta}_0)} -1 \right)\right) \nonumber\\
		&= \sum_{i=1}^n 2\log(1+\delta_i),
	\end{align}
	where $\delta_i = \sum_{g=1}^G \bar{\alpha}_g\left(\frac{f(x_i;   \bar{\btheta}_g)}{f(x_i;   {\btheta}_0)} -1 \right)$. 
	Applying the inequality $\log(1+x) \leq x - x^2/2 + x^3/3$, we have 
	\begin{equation}\label{basic}
		R_{1n} \leq 2\sum_{i=1}^{n} \delta_i - \sum_{i=1}^{n} \delta_i^2 + (2/3) \sum_{i=1}^{n} \delta_i^3.
	\end{equation}
	We first  deal with  $\sum_{i=1}^{n}\delta_i$ in (\ref{basic}).  By Taylor's expansion of $f(x_i;   \bar{\btheta}_g)$ at ${\btheta}_0$, we have 
	\begin{align}\label{eq:deltai}
		\delta_i &= \sum_{g=1}^G \bar{\alpha}_g\frac{f(x_i;   \bar{\btheta}_g) - f(x_i;   {\btheta}_0)}{f(x_i;   {\btheta}_0)} \notag\\
		&= \sum_{h=1}^{d} \sum_{g=1}^G \bar{\alpha}_g(\bar{\theta}_{gh} - {\theta}_{0h}) Y_{ih} +
		\sum_{h=1}^{d} \sum_{g=1}^G \bar{\alpha}_g(\bar{\theta}_{gh} - {\theta}_{0h})^2 Z_{ih} \notag \\
		&+ \sum_{h<\ell}^{d} \sum_{g=1}^G \bar{\alpha}_g(\bar{\theta}_{gh} - {\theta}_{0h})(\bar{\theta}_{g\ell} - {\theta}_{0\ell}) U_{ih\ell} + \varepsilon_{in},
	\end{align}
	where $ Y_{ih}$, $Z_{ih}$ and $U_{ih\ell}$ are defined in (\ref{eq:b}) and $ \varepsilon_{in}$ is the remainder term.
	Let
	$$\bar{\m} = \m(\bar{\ba}, \bar{\bxi}, \bxi_0) \mbox{ and }, \bar{\m}_1 = \m_1(\bar{\ba}, \bar{\bxi}, \bxi_0)$$
	where $\m $ and $\m_1$ are defined in (\ref{eq:m}) and (\ref{eq:m1}).
	Then, we write (\ref{eq:deltai}) as 
	\begin{equation}\label{eq:deltai2}
		\delta_i = \bar{\m}\trans \b_i +  \varepsilon_{in},
	\end{equation}
	where $\b_i$ is defined in (\ref{eq:b}). 
	
	\noindent
	{\bf{Step 1. Controlling the remainder term $\bm{\varepsilon}_{in}$}}.  We aim to prove 
	\begin{equation} \label{eps1}
		\sum_{i=1}^{n}\varepsilon_{in} = o_p(1) + o_p(n)\left(\Vert\bar{\m}\Vert_2^2\right).
	\end{equation}
	In order to show this, we note that $\sum_{i=1}^{n}\varepsilon_{in}$ can be written as 
	\begin{align*}
		\sum_{i=1}^{n}\varepsilon_{in} &= 	\sum_{i=1}^{n} \sum_{j_{1}=1}^d \cdots \sum_{j_{3}=1}^d\sum_{g=1}^G \bar{\alpha}_g \prod_{s=1}^3 (\bar{\theta}_{gj_s} - {\theta}_{0j_s}) \left(\frac{\partial^3 f(x_i;  \btheta_0)}{\partial \theta_{j_1} \partial \theta_{j_2}\partial \theta_{j_3}}\right) \bigg /  (3!f(x_i;  \btheta_0)) \\
		&+ \sum_{i=1}^{n} \sum_{j_{1}=1}^d \cdots \sum_{j_{4}=1}^d\sum_{g=1}^G \bar{\alpha}_g \prod_{s=1}^4 (\bar{\theta}_{gj_s} - {\theta}_{0j_s}) \left(\frac{\partial^4 f(x_i;  \btheta_0)}{\partial \theta_{j_1} \partial \theta_{j_2}\partial \theta_{j_3}\partial \theta_{j_4}}\right)\bigg /  (4!f(x_i;  \btheta_0)) \\
		&+ \sum_{i=1}^{n} \sum_{j_{1}=1}^d \cdots \sum_{j_{5}=1}^d\sum_{g=1}^G  \bar{\alpha}_g \prod_{s=1}^5 (\bar{\theta}_{gj_s} - {\theta}_{0j_s}) \left(\frac{\partial^5 f(x_i;  {\bm \zeta}_{g}(x_i))}{\partial \theta_{j_1} \partial \theta_{j_2}\partial \theta_{j_3}\partial \theta_{j_4}\partial \theta_{j_5}}\right) \bigg/  (5!f(x_i;\btheta_0)) \\
		& = {\rm I + II + III},
	\end{align*}
	where ${\bm \zeta}_{g}(x_i)$ lies between $\bar{\btheta}_g$ and $\btheta_0$. 
	
	For I, note that the production term $\prod_{s=1}^3 (\bar{\theta}_{gj_s} - {\theta}_{0j_s})$ does not involve the index $i$. We can change the summation and production order as 
	$\sum_{j_{1}=1}^d \cdots \sum_{j_{3}=1}^d\sum_{g=1}^G\prod_{s=1}^3 \sum_{i=1}^{n}  $.
	Further, for any fixed $j_1,j_2,j_3$,  by Cauchy's inequality, we have
	\begin{equation} \label{eq:theta}
		\prod_{s=1}^3 \left|\bar{\theta}_{gj_s} - {\theta}_{0j_s}\right| \leq 
		\left(\sum_{j=1}^d \left|\bar{\theta}_{gj} - {\theta}_{0j} \right|\right)^3\leq \left(\sqrt{d}\right)^3 \left\|\bar{\btheta}_g -{\btheta}_{0}\right\|_2^3.
	\end{equation}
	Hence, 
	\begin{align*}
		|{\rm I}| &\leq \sum_{j_{1}=1}^d \cdots \sum_{j_{3}=1}^d \left|\sum_{g=1}^G \bar{\alpha}_g \prod_{s=1}^3 (\bar{\theta}_{gj_s} - {\theta}_{0j_s}) \left(\sum_{i=1}^{n}  \left(\frac{\partial^3 f(x_i;  \btheta_0)}{\partial \theta_{j_1} \partial \theta_{j_2}\partial \theta_{j_3}}\right)\bigg/  (3!f(x_i;  \btheta_0))\right)\right| \\
		&\leq \sum_{j_{1}=1}^d \cdots \sum_{j_{3}=1}^d \sum_{g=1}^G \bar{\alpha}_g \left|\prod_{s=1}^3 (\bar{\theta}_{gj_s} - {\theta}_{0j_s})\right| \left|\left(\sum_{i=1}^{n}  \left(\frac{\partial^3 f(x_i;  \btheta_0)}{\partial \theta_{j_1} \partial \theta_{j_2}\partial \theta_{j_3}}\right)\bigg/  (3!f(x_i;  \btheta_0))\right)\right| \\
		&\leq  \sum_{j_{1}=1}^d \cdots \sum_{j_{3}=1}^d \left| \left(\sum_{g=1}^G \bar{\alpha}_g \left(\sqrt{d}\right)^3  \Vert\bar{\btheta}_g -{\btheta}_{0}\Vert_2^3\right) \left(\sum_{i=1}^{n} \left(\frac{\partial^3 f(x_i;  \btheta_0)}{\partial \theta_{j_1} \partial \theta_{j_2}\partial \theta_{j_3}}\right) \bigg /  (3!f(x_i;  \btheta_0))\right)\right|.
	\end{align*}
	Also for any fixed $j_1,j_2,j_3$, let $D_i(j_1,j_2,j_3) = \left(\frac{\partial^3 f(x_i;  \btheta_0)}{\partial \theta_{j_1} \partial \theta_{j_2}\partial \theta_{j_3}}\right)\big /  (3!f(x_i;  \btheta_0))$. 
	By Condition (WC3), we have $\Ex\left(D_i(j_1,j_2,j_3) \right) = 0$ and $\var\left(D_i(j_1,j_2,j_3) \right) < \infty$. Applying the Central Limit Theorem, we have $\sum_{i=1}^n D_i(j_1,j_2,j_3)  = O_p\left(n^{1/2}\right)$.
	Hence, we conclude that 
	\begin{equation}
		|{\rm I}| = O_p\left(n^{1/2}\right)\left(\sum_{g=1}^G \bar{\alpha}_g \left(\sqrt{d}\right)^3  \Vert\bar{\btheta}_g -{\btheta}_{0}\Vert_2^3\right).
	\end{equation}
	Similarly, for II, we have 
	\begin{equation}
		|{\rm II}| = O_p\left(n^{1/2}\right)\left(\sum_{g=1}^G \bar{\alpha}_g \left(\sqrt{d}\right)^4 \Vert\bar{\btheta}_g -{\btheta}_{0}\Vert_2^4\right).
	\end{equation}
	For III,  from Condition (WC3), for any $j_1,\ldots, j_5$, we have 
	$$\sup_{\Vert\btheta - \btheta_0\Vert \leq \tau} \left|\frac{\partial^5 f(x,  \btheta)}{\partial \theta_{j_1} \cdots \partial \theta_{j_5}} \bigg /  f(x,  \btheta_0)\right| \leq g(x; \btheta_0).$$
	By  the law of large numbers, we have 
	$$ \sum_{i=1}^n g(x_i; \btheta_{0}) = O_p(n). $$
	Using the consistency of $\bar{\btheta}_{g}$ from Lemma \ref{lem:consisasy}, we have 
	\begin{equation}
		|{\rm III }| = O_p(n)\left(\sum_{g=1}^G \bar{\alpha}_g \left(\sqrt{d}\right)^5  \Vert\bar{\btheta}_g -{\btheta}_{0}\Vert_2^5\right).
	\end{equation}
	Since $\Vert\bar{\btheta}_g -{\btheta}_{0}\Vert_2 = o_p(1)$, we get
	$$\left|\sum_{i=1}^{n}\varepsilon_{in}\right| =  \left|o_p\left(n^{1/2}\right)\right|\sum_{g=1}^G \bar{\alpha}_g  \Vert\bar{\btheta}_g -{\btheta}_{0}\Vert_2^2 + |o_p(n)|\sum_{g=1}^G \bar{\alpha}_g  \Vert\bar{\btheta}_g -{\btheta}_{0}\Vert_2^4.$$
	On the one hand, we have 
	\begin{equation}\label{eq2}
		|o_p(n)|\sum_{g=1}^G \bar{\alpha}_g  \Vert\bar{\btheta}_g -{\btheta}_{0}\Vert_2^4 \leq |o_p(n)| \cdot \Vert\bar{\m}\Vert_2^2.
	\end{equation}
	On the other hand, we have
	\begin{equation}\label{eq3}
		\left|o_p\left(n^{1/2}\right)\right| \sum_{g=1}^G \bar{\alpha}_g  \Vert\bar{\btheta}_g -{\btheta}_{0}\Vert_2^2 \leq \left|o_p\left(n^{1/2}\right)\right| \Vert\bar{\m}\Vert_2 \leq |o_p(1)| + |o_p(n)| \cdot \Vert\bar{\m}\Vert_2^2,
	\end{equation}
	where the last inequality is from
	$$\left|o_p\left(n^{1/2}\right)\right| \Vert\bar{\m}\Vert_2 = \left|o_p(1)\right| n^{1/2} \Vert\bar{\m}\Vert_2 \leq \left|o_p(1)\right|  \left(\frac{1+n \Vert\bar{\m}\Vert_2^2}{2}\right) =  |o_p(1)| + |o_p(n)| \cdot \Vert\bar{\m}\Vert_2^2.$$
	Combining (\ref{eq2}) with (\ref{eq3}), we get (\ref{eps1}).
	
	\noindent
	{\bf{Step 2. Obtaining the convergence rate}}. From (\ref{eps1}), we have 
	\begin{equation} \label{eq:cv1}
		\sum_{i=1}^{n}\delta_i = \sum_{i=1}^n \left( \bar{\m}\trans \b_i +  \varepsilon_{in}\right) = \sum_{i=1}^n \bar{\m}\trans \b_i +  o_p(1) + o_p(n)\Vert\bar{\m}\Vert_2^2.
	\end{equation}
	Similarly, we can prove 
	\begin{equation} \label{eq:cv2}
		\sum_{i=1}^n \delta_i^2 = \sum_{i=1}^n \left(\bar{\m}\trans \b_i +  \varepsilon_{in}\right)^2 = \sum_{i=1}^n \left(\bar{\m}\trans \b_i\right)^2 +  o_p(1) + o_p(n)\Vert\bar{\m}\Vert_2^2,
	\end{equation}
	and 
	\begin{equation} \label{eq:cv3}
		\sum_{i=1}^n \delta_i^3 = \sum_{i=1}^n \left(\bar{\m}\trans \b_i +  \varepsilon_{in}\right)^3 = \sum_{i=1}^n \left(\bar{\m}\trans \b_i\right)^3 +  o_p(1) + o_p(n)\Vert\bar{\m}\Vert_2^2.
	\end{equation}
	In fact, for  (\ref{eq:cv2}), we have 
	$$\sum_{i=1}^n \left(\bar{\m}\trans \b_i +  \varepsilon_{in}\right)^2 = \sum_{i=1}^n \left(\bar{\m}\trans \b_i\right)^2 + \sum_{i=1}^n \left(\varepsilon_{in}^2 + \bar{\m}\trans \b_i \varepsilon_{in}\right).$$
	By Taylor's expansion, we have 
	\begin{equation} \label{eq:rateExpansion}
		\varepsilon_{in} = 	\sum_{j_{1}=1}^d \cdots \sum_{j_{3}=1}^d\sum_{g=1}^G \bar{\alpha}_g \prod_{s=1}^3 (\bar{\theta}_{gj_s} - {\theta}_{0j_s}) \left(\frac{\partial^3 f(x_i;  {\bm \zeta}_{g}(x_i))}{\partial \theta_{j_1} \partial \theta_{j_2}\partial \theta_{j_3}}\right) \bigg/  (3!f(x_i;  \btheta_0)),
	\end{equation}
	where ${\bm \zeta}_{g}(x_i)$ lies between $\bar{\btheta}_g$ and $\btheta_0$. Note that here we only need to represent the remainder term $\varepsilon_{in}$ in terms of the third derivatives.
	Again, from Condition (WC3), for any fixed $j_1,\ldots,j_3$, we have 
	$$  \sup_{\Vert\btheta - \btheta_0\Vert_2 \leq \tau} \left|\left(\frac{\partial^3 f(x;  \btheta)}{\partial \theta_{j_1} \partial \theta_{j_2}\partial \theta_{j_3}}\right) \bigg/  (3!f(x;  \btheta_0))\right| \leq g(x; \btheta_0).$$
	%and  $\left\|g(x)\right\|_{L^3} < \infty$.
	Note that 
	$\left|\bar{\m}\trans \b_i\right| \leq \Vert\bar{\m}\Vert_2 \Vert\b_i\Vert_2 $. By $\bar{\bxi} - \bxi_0=o_p(1)$, together with the inequality (\ref{eq:theta}), Condition (WC3) and the law of large numbers,  $\sum_{i=1}^n \left |\bar{\m}\trans \b_i \varepsilon_{in}\right|$ can be bounded by 
	\begin{align*}
		&\quad \sum_{i=1}^n \left|\bar{\m}\trans \b_i \varepsilon_{in}\right|  \leq \sum_{i=1}^n  \Vert\bar{\m}\Vert_2 \Vert\b_i\Vert_2 |\varepsilon_{in}| \\
		&\leq \Vert\bar{\m}\Vert_2 \sum_{i=1}^{n}    \Vert\b_i\Vert_2	\sum_{g=1}^G \left(\sqrt{d}\right)^3 \Vert\bar{\btheta}_g -{\btheta}_{0}\Vert_2^3  \sum_{j_{1}=1}^d \cdots \sum_{j_{3}=1}^d   \left|\left(\frac{\partial^3 f(x_i;   {\bm \zeta}_{g}(x_i))}{\partial \theta_{j_1} \partial \theta_{j_2}\partial \theta_{j_3}}\right) \bigg/  (3!f(x_i;  \btheta_0))\right| \\ 
		&\leq \Vert\bar{\m}\Vert_2
		\sum_{g=1}^G \left(\sqrt{d}\right)^3 \Vert\bar{\btheta}_g -{\btheta}_{0}\Vert_2^3 d^3\sum_{i=1}^{n} \Vert\b_i\Vert_2|g(x_i;\btheta_0)| \\
		&= o_p(1) O_p(n) \Vert\bar{\m}\Vert_2^2 = o_p(n)\Vert\bar{\m}\Vert_2^2.
	\end{align*}
	For the  $\sum_{i=1}^n \varepsilon_{in}^2$ term, when $\Vert\bar{\bxi} - \bxi_0\Vert_2 \leq \tau$
	, we have 
	\begin{equation}\label{eq:varepsilon}
		\varepsilon_{in}^2 \leq  \left(d^3\sum_{g=1}^G \left(\sqrt{d}\right)^3 \Vert\bar{\btheta}_g -{\btheta}_{0}\Vert_2^3\right)^2 g^2(x_i;\btheta_0). 
	\end{equation}
	Since $\Vert\bar{\bxi} - \bxi_0\Vert_2 \overset{p}{\longrightarrow} 0$, we have
	$\sum_{i=1}^n \varepsilon_{in}^2 = o_p(n)\Vert\bar{\m}\Vert_2^2. $
	Thus, we prove (\ref{eq:cv2}). Similarly, we can prove (\ref{eq:cv3}) . 
	
	Finally, by the law of large numbers, we have 
	\begin{equation} \label{eq:ratelln2}
		\sum_{i=1}^n \left(\bar{\m}\trans \b_i\right)^2 = n\bar{\m}\trans \B \bar{\m}(1+o_p(1)), 
	\end{equation}
	and
	\begin{equation} \label{eq:ratelln3}
		\left|\sum_{i=1}^n \left(\bar{\m}\trans \b_i\right)^3\right| \leq \sum_{i=1}^n  \Vert\b_i\Vert_2^3  \Vert\bar{\m}\Vert_2^3 \leq O_p(n) \Vert\bar{\m}\Vert^3_2 =   o_p(n) \Vert\bar{\m}\Vert^2_2, 
	\end{equation}
	where $\sum_{i=1}^n  \Vert\b_i\Vert_2^3 = O_p(n)$  is from Condition (\WCc).
	Since 
	$$n\bar{\m}\trans \B \bar{\m} \geq \lambda_{\rm min}(\B)n \Vert\bar{\m}\Vert^2_2 = O_p(n)\Vert\bar{\m}\Vert^2_2,$$
	we conclude that 
	$o_p(n) \Vert\bar{\m}\Vert^2_2 =  o_p(1)n\bar{\m}\trans \B \bar{\m}$.
	Combining (\ref{eq:cv1}) --  (\ref{eq:cv3}) with  (\ref{eq:ratelln2}) --  (\ref{eq:ratelln3}), we get 
	\begin{equation} \label{eq:Rbound}
		R_{1n}(\bar{\bxi}, \bar{\ba}) \leq 2\bar{\m}\trans\sum_{i=1}^n \b_i - n\bar{\m}\trans \B \bar{\m}(1+o_p(1)) + o_p(1).
	\end{equation}
	Since we know $R_{1n} \geq 2c$, the inequality (\ref{eq:Rbound}) implies that 
	$$ \frac{2c}{n} \leq 2\Vert\bar{\m}\Vert_2  \left\Vert \frac{\sum_{i=1}^n \b_i }{n}   \right\Vert_2   -  (\lambda_{\rm min}(\B) + o_p(1))\Vert\bar{\m}\Vert_2^2 + o_p(1/n).$$
	Applying the inequality 
	\begin{equation}\label{eq:meaninequ}
		\frac{\lambda_{\rm min}(\B)}{2} \Vert\bar{\m}\Vert_2^2 + \frac{2} {\lambda_{\rm min}(\B)}   \left\Vert \frac{\sum_{i=1}^n \b_i }{n}   \right\Vert_2^2  \geq 2\Vert\bar{\m}\Vert_2 \left\Vert \frac{\sum_{i=1}^n \b_i }{n}   \right\Vert_2 ,  
	\end{equation}
	we conclude that 
	\begin{equation}\label{eq:barm}
		\left(\frac{\lambda_{\rm min}(\B)}{2} +o_p(1) \right)\Vert\bar{\m}\Vert_2^2 \leq \frac{2} {\lambda_{\rm min}(\B)}   \left\Vert \frac{\sum_{i=1}^n \b_i }{n}   \right\Vert_2^2  -  \frac{2c}{n} + o_p(1/n).
	\end{equation}
	Using the fact $ \left\Vert \frac{\sum_{i=1}^n \b_i }{n}   \right\Vert_2^2 = O_p(1/n)$,
	(\ref{eq:barm}) implies that 
	$\Vert\bar{\m}\Vert_2 = O_p\left(n^{-1/2}\right)$, and thus $\Vert\bar{\m}_1\Vert_2 = O_p(n^{-1/2})$. Since $\delta \leq \bar{\alpha}_g \leq 1$, we have 
	$\Vert\bar{\btheta}_g -{\btheta}_0\Vert_2= O_p\left(n^{-1/4}\right) $  ($g = 1,\dots,G$), and thus we complete the proof.
\end{proof}

\begin{proof}[Proof of Lemma \ref{lem:asymalpha}]
	The first step aims to show 
	\begin{equation}\label{eq:asymalpha1}
		n^{-1}\sum_{i=1}^n \bar{w}_{gi} - \bar{\alpha}_g = o_p(1).
	\end{equation}
	By the definition of $\bar{w}_{gi}$, we have
	$$\bar{w}_{gi} - \bar{\alpha}_g = \frac{\bar{\alpha}_g f(x_i;  \bar{\btheta}_g)}{\varphi(x_i; \bar{\bxi}, \bar{\alpha})} - \bar{\alpha}_g.$$
	Let $\delta_{gi} = \frac{f(x_i;  \bar{\btheta}_g)}{f(x_i;  {\btheta}_0)} - 1$ and $\delta_i = \frac{\varphi(x_i; \bar{\bxi}, \bar{\alpha})}{f(x_i;  {\btheta}_0)} - 1 $. We can rewrite $\bar{w}_{gi} - \bar{\alpha}_g $ as 
	$$\bar{w}_{gi} - \bar{\alpha}_g = \bar{\alpha}_g \frac{1+\delta_{gi}}{1+\delta_i} - \bar{\alpha}_g= \bar{\alpha}_g \frac{\delta_{gi} - \delta_i}{1+\delta_i}.$$
	Thus, we only need to prove 
	\begin{equation}\label{eq:wa}
		n^{-1}\sum_i^n \bar{w}_{gi} - \bar{\alpha}_g  =\bar{\alpha}_g  n^{-1}\sum_{i=1}^n\frac{\delta_{gi} - \delta_i}{1+\delta_i} = o_p(1).
	\end{equation}
	To prove (\ref{eq:wa}), we first prove $\max_i| \delta_i| = o_p(1).$ As in the proof of Lemma \ref{lem:rateasy},  (\ref{eq:deltai2}) gives 
	$$\delta_i = \bar{\m}\trans \b_i + \varepsilon_{in},$$
	and (\ref{eq:varepsilon}) gives
	$$			\varepsilon_{in} \leq  \left(d^3\sum_{g=1}^G \left(\sqrt{d}\right)^3 \Vert\bar{\btheta}_g -{\btheta}_{0}\Vert_2^3\right) g(x_i;\btheta_0). $$
	Then, it remains to show $\max_i \Vert \bar{\m}\trans \b_i \Vert_2 = o_p(1)$ and $\max_i \Vert\bar{\btheta}_g -{\btheta}_{0}\Vert_2^3 g(x_i;\btheta_0) = o_p(1)$.
	Since $\Vert \bar{\m}  \Vert_2 = O_p\left( n ^{-1/2}\right) $, we have $n^{3/8}\Vert \bar{\m}  \Vert_2 = o_p(1)$. 
	In order to show $\max_i \Vert \bar{\m}\trans \b_i \Vert_2 = o_p(1)$, we only need to prove 
	$$\max_i n^{-3/8}\Vert  \b_i \Vert_2 = o_p(1).$$
	For any $\epsilon > 0$, we have  
	$$\pr \left(\max_i n^{-3/8}\Vert  \b_i \Vert_2 \geq \epsilon\right) \leq \sum_{i=1}^{n} \pr \left(  n^{-3/8}\Vert  \b_i \Vert_2 \geq \epsilon\right).$$
	By Chebyshev's inequality, we have
	$$
	\pr \left(  n^{-3/8}\Vert  \b_i \Vert_2 \geq \epsilon\right) \leq \frac{\Ex\left( \Vert  \b_i \Vert_2^3 \right)}{n^{9/8}\epsilon^3}.
	$$
	Thus, we have 
	$$\pr \left(\max_i n^{-3/8}\Vert  \b_i \Vert_2 \geq \epsilon\right) \leq \frac{\Ex\left( \Vert  \b_1 \Vert_2^3 \right)}{n^{1/8}\epsilon^3}.$$
	It follows that $\max_i \Vert \bar{\m}\trans \b_i \Vert_2 = o_p(1) $.  Similarly, we can prove $\max_i \Vert\bar{\btheta}_g -{\btheta}_{0}\Vert_2^3 g(x_i;\btheta_0) = o_p(1) $, and thus
	$\max_i| \delta_i| = o_p(1)$.  
	In order to show (\ref{eq:wa}), by the fact $\max_i| \delta_i| = o_p(1)$, 
	we only need to  show that 
	$$n^{-1} \sum_{i=1}^n \left|\delta_{gi}\right|  = n^{-1} \sum_{i=1}^n   \left|\frac{f(x_i;  \bar{\btheta}_g) - f(x_i;  {\btheta}_0)}{f(x_i;  {\btheta}_0)}\right|  = o_p(1),$$
	which is similar to the proof of (\ref{eq:cv1}) without the summation over of $g$. 
	More specifically, by Lagrange's mean value theorem, we have 
	$$\frac{f(x_i;  \bar{\btheta}_g) - f(x_i;  {\btheta}_0)}{f(x_i;  {\btheta}_0)} =   \sum_{j=1}^d   (\bar{\btheta}_{gj_s} - {\btheta}_{0j_s}) \left(\frac{\partial f(x_i;  {\bm \zeta}_{g}(x_i))}{\partial \theta_{j}}\right) \bigg /  f(x_i;  \btheta_0),$$
	where $ {\bm\zeta}_{g}(x_i)$ lies between $\bar{\btheta}_g$ and $\btheta_0$.  By Condition (\WCc),  when $\Vert \bar{\btheta}_g - \btheta\Vert_2 \leq \tau$, we have 
	$$\left | \left(\frac{\partial f(x_i;  {\bm \zeta}_{g}(x_i))}{\partial \theta_{j}}\right) \bigg /  f(x_i;  \btheta_0) \right| \leq g(x_i; \btheta_{0}),$$
	and $\Ex (g(x_i; \btheta_{0})) < \infty$.  Since $ \Vert \bar{\btheta}_g - \btheta\Vert_2  = o_p(1)$, we prove that 
	$n^{-1} \sum_{i=1}^n \left|\delta_{gi}\right|  =  o_p(1).$

	Then, it suffices to prove that $n^{-1}\sum_{i=1}^n \bar{w}_{gi} - \bar{\alpha}^{(1)}_g = o_p(1)$.  Note that 
	$$
	n^{-1}\sum_{i=1}^n \bar{w}_{gi} - \bar{\alpha}^{(1)}_g = \frac{\sum_{i=1}^n \bar{w}_{gi}}{n} - \frac{\sum_{i=1}^n \bar{w}_{gi} + \lambda}{n+ G\lambda} =  \frac{-G\lambda\sum_{i=1}^n \bar{w}_{gi} + n\lambda}{n(n+ G\lambda)}.
	$$
	By (\ref{eq:wa}),  we have $n^{-1}\sum_{i=1}^n \bar{w}_{gi} = \bar{\alpha}_g  + o_p(1).$ It implies that 
	$$
	\frac{-G\lambda\sum_{i=1}^n \bar{w}_{gi} + n\lambda}{n(n+ G\lambda)} = o_p(1),
	$$
	which proves the lemma.
\end{proof}

\subsection{Proofs of Theorem \thmtwo} 
In this subsection, we give the detailed proof of Theorem \thmtwo.  Let $r = \min(G-1,d)$. Recall that 
\begin{equation} \label{eq:V}
	\mathcal{V} = \left\{ {\rm vech }(\V): \V \in \mathbb{R}^{d \times d} \mbox{  is symmetric},~{\rm rank}(\V)\leq r, \V \succeq 0\right\}.
\end{equation}
We require the following lemma. 
\begin{lemma} \label{lem:V}
	For any fixed $
	\ba \in \Delta^{G-1}$, define  $$\mathcal{V}_{\ba} = \left\{{\bf v}:  {\bf v} = {\rm vech} (\A\A\trans), \sum_{g=1}^G  {{\alpha}_g} {\A}_{\cdot g} = {\bf 0}, \A \in \mathbb{R}^{d \times G} \right\}. $$ Then, we have $\mathcal{V}_{\ba} \equiv \mathcal{V}$, where $ \mathcal{V}$ is defined in (\ref{eq:V}). 
\end{lemma}
\begin{proof}[ Proof of Theorem \thmtwo]
	Let 
	$$R_{0n}(\hat{\bxi}_0, \ba_0) = 2\left\{pl_n\left(\hat{\bxi}_0, \ba_0\right) - pl_n({\bxi}_0,\ba_0) \right\},$$
	and 
	$$R_{1n}\left({\bxi}^{(k)},{\balpha}^{(k)}\right) = 2 \left\{pl_n\left( {\bxi}^{(k)},{\balpha}^{(k)} \right) - pl_n( \bxi_0, \ba_0)\right\}.$$
	Firstly, we claim that
	$$R_{0n}(\hat{\bxi}_0, \ba_0) = \left(\sum_{i=1}^n\b_{1i}\right)\trans \left(n\B_{11}\right)^{-1}\left(\sum_{i=1}^n\b_{1i}\right) + o_p(1).$$
	In fact, by Condition (\Ca)--(\Cb) and  Condition (\WCc)--(\WCd), it is the classical expansion of the log likelihood ratio.
	By the likelihood non-decreasing property of the EM algorithm, we have 
	$$pl_n\left({\bxi}^{(k)},{\balpha}^{(k)} \right) \geq pl_n\left(  {\bxi}^{(0)},{\balpha}^{(0)} \right) \geq pl_n( {\bxi}_{0},{\balpha}_0) + p_0,$$
	where $p_0 =    \lambda G \log(\delta G) $ is a constant. 
	Hence, by Theorem \thmone { } and (\ref{eq:Rbound}), for any fixed $k$,  we have
	$$R_{1n}\left({\bxi}^{(k)},{\balpha}^{(k)}\right) \leq  2\left(\m^{(k)}\right)\trans \sum_{i=1}^n \b_i - n \left(\m^{(k)}\right)\trans \B \left(\m^{(k)}\right) \{1+o_p(1)\} + o_p(1),$$
	where $\m^{(k)} = \m(\ba^{(k)}, \bxi^{(k)}, \bxi_0)$ and $\m$ is defined in (\ref{eq:m}). 
	In this proof, for simplicity of notation, from now on we omit the superscript $k$ and abbreviate $\balpha^{(k)}, \bxi^{(k)}, \m^{(k)}$ and $\m_1^{(k)} = \m_1(\ba^{(k)}, \bxi^{(k)}, \bxi_0)$ to $\hat{\balpha}, \hat{\bxi}, \hat{\m}$ and $\hat{\m}_1$, where $\m_1$ is defined in (\ref{eq:m1}). 
	Since $n \hat{\m}\trans \B \hat{\m} = O_p(1)$, we have 
	$$R_{1n}\left(\hat{\bxi},\hat{\balpha}\right)  \leq  2\
	\hat{\m}\trans \sum_{i=1}^n \b_i - n \hat{\m}\trans \B \hat{\m}  + o_p(1).$$
	Let $ \tilde{v}_{hg} = \sqrt{\hat{\alpha}_g} ( \hat{\theta}_{gh}- {\theta}_{0h})$ and $\widetilde{\V}= [\tilde{v}_{hg}] \in \mathbb{R}^{d \times G}$. Define $\widetilde{\V}= (\widetilde{\V}_{\cdot1}, \dots, \widetilde{\V}_{\cdot G} )$.  This gives  $\sum_{g=1}^G  \sqrt{\hat{\alpha}_g}\widetilde{\V}_{\cdot g} = \hat{\m}_1$. Hence, we have 
	$$\widetilde{\V}_{\cdot 1} = \frac{\hat{\m}_1- \sum_{g=2}^G  \sqrt{\hat{\alpha}_g}\widetilde{\V}_{\cdot g}}{\sqrt{\hat{\alpha}_1}}.$$
	Based on this equation, we define $\hat{\V}$ as
	$$\hat{\V}_{\cdot 1} = \frac{- \sum_{g=2}^G  \sqrt{\hat{\alpha}_g}\widetilde{\V}_{\cdot g}}{\sqrt{\hat{\alpha}_1}} \mbox{ and } \hat{\V}_{\cdot g} = \widetilde{\V}_{\cdot g}, g\neq 1.$$
	It follows that $\sum_{g=1}^G  \sqrt{\hat{\alpha}_g}\hat{\V}_{\cdot g} = {\bf 0}$. 
	Let $\hat{\bf v}= {\rm vech} \left(\hat{\V}\hat{\V}\trans\right)$ and ${\bf \tilde{v}}= {\rm vech} \left(\widetilde{\V}\widetilde{\V}\trans\right)$. By Theorem \thmone, we have $\tilde{v}_{hg} = o_p(1)$ and $\left\| \hat{\m}_1\right\|_2 = O_p(n^{-1/2})$. Therefore, we have 
	${\bf \tilde{v}}= \hat{\bf v}+ o_p(n^{-1/2})$. Let $\hat{\t}= \left(\hat{\m}_1\trans, \hat{\bf v}\trans\right)\trans$.  
	Since $\hat{\m} = \left(\hat{\m}_1\trans, {\bf \tilde{v}}\trans\right)\trans$, we have $\hat{\m} = \hat{\t} + o_p\left(n^{-1/2}\right)$.  It follows that 
	$$R_{1n}\left( \hat{\bxi}, \hat{\balpha}\right) \leq 2\hat{\t}\trans \sum_{i=1}^n \b_i - n \hat{\t}\trans \B \hat{\t}  + o_p\left(1\right).$$
	Let $\tilde{\m}_1 = \hat{\m}_1 + \B_{11}^{-1}\B_{12}  \hat{\bf \v}$ and $\widetilde{\B}_{22} =  {\B}_{22} - \B_{21}\B_{11}^{-1}\B_{12}$. It is clear that 
	$$\hat{\t}\trans \B \hat{\t} = \tilde{\m}_1\trans \B_{11} \tilde{\m}_1 + \hat{\v}\trans \widetilde{\B}_{22}  \hat{\bf v},$$
	and 
	$$\hat{\t}\trans \sum_{i=1}^n \b_i =\tilde{\m}_1\trans  \sum_{i=1}^n \b_{1i} + \hat{\v}\trans \left( \sum_{i=1}^n \tilde{\b}_{2i}  \right),$$
	where $\tilde{\b}_{2i} = \b_{2i} - \B_{21}\B_{11}^{-1}\b_{1i}$. 
	Then, we have 
	\begin{align}
		\label{eq:part1} R_{1n}\left(\hat{\bxi},\hat{\balpha}\right)  &\leq 2\tilde{\m}_1\trans \sum_{i=1}^n \b_{1i} - n \tilde{\m}_1\trans \B_{11} \tilde{\m}_1 
		+ 2\hat{\bf \v}\trans \sum_{i=1}^n \tilde{\b}_{2i} -n\hat{\v}\trans \widetilde{\B}_{22} \hat{\bf v} + o_p\left(1\right)  \notag \\
		&\leq \left(\sum_{i=1}^n\b_{1i}\right)\trans \left(n\B_{11}\right)^{-1}\left(\sum_{i=1}^n\b_{1i}\right)  
		+ 2\hat{\v}\trans\sum_{i=1}^n \tilde{\b}_{2i} - n\hat{\v}\trans \widetilde{\B}_{22} \hat{\bf v} + o_p\left(1\right) .
	\end{align}
	Subtracting $R_{0n}(\hat{\bxi}_0, \ba_0)$ from $R_{1n}\left(\hat{\bxi},\hat{\balpha}\right) $, we have 
	$$
	R_{1n}\left(\hat{\bxi},\hat{\balpha}\right) -R_{0n}(\hat{\bxi}_0, \ba_0)  \leq 2\hat{\v}\trans \sum_{i=1}^n \tilde{\b}_{2i} - n\hat{\v}\trans \widetilde{\B}_{22} \hat{\bf v} + o_p\left(1\right).
	$$
	Let $\mathcal{V}_{\hat{\ba}} = \left\{{\bf v}:  {\bf v} = {\rm vech} \left(\V\V\trans\right), \  \sum_{g=1}^G  \sqrt{\hat{\alpha}_g} {\V}_{\cdot g} = 0\right\}$. 
	Since $\hat{\bf v} \in \mathcal{V}_{\hat{\ba}}$, we have
	$$
	R_{1n}\left(\hat{\bxi},\hat{\balpha}\right) -R_{0n}(\hat{\bxi}_0, \ba_0)  \leq \sup_{{\bf v} \in \mathcal{V}_{\hat{\ba}} }2{\bf v}\trans \sum_{i=1}^n \tilde{\b}_{2i} - n{\bf v}\trans \widetilde{\B}_{22} {\bf v} + o_p(1).
	$$
	By Lemma \ref{lem:V},  we have $\mathcal{V}_{\hat{\ba}} \equiv \mathcal{V}$. 
	Based on this fact, we can rewrite the above inequality as 
	$$
	R_{1n}\left(\hat{\bxi},\hat{\balpha}\right) -R_{0n}(\hat{\bxi}_0, \ba_0)   \leq \sup_{{\bf v} \in \mathcal{V}} 2{\bf v}\trans \sum_{i=1}^n \tilde{\b}_{2i} - n{\bf v}\trans \widetilde{\B}_{22} {\bf v} + o_p(1).
	$$
	Hence, 
	$${\rm EM}_n^{(k)} \leq \sup_{{\bf v} \in \mathcal{V}} 2{\bf v}\trans \sum_{i=1}^n 
	\tilde{\b}_{2i} - n{\bf v}\trans \widetilde{\B}_{22} {\bf v} + o_p(1).$$
	On the other hand, let 
	$$\hat{{\bf v}}^{\flat} = \argmax_{\substack{{\bf v} \in \mathcal{V}}} 2{\bf v}\trans \sum_{i=1}^n \tilde{\b}_{2i} - n{\bf v}\trans \widetilde{\B}_{22} {\bf v}.$$
	Since ${\bf 0} \in \mathcal{V}$, it follows that 
	\begin{equation}\label{eq:vfalt}
		0 \leq 2\hat{{\bf v}}^{\flat {\rm T}} \sum_{i=1}^n \tilde{\b}_{2i} - n\hat{{\bf v}}^{\flat {\rm T}}\widetilde{\B}_{22}\hat{{\bf v}}^{\flat} \leq 2\Vert \hat{{\bf v}}^{\flat }\Vert_2 \Vert\sum_{i=1}^n \tilde{\b}_{2i}\Vert_2 - n\lambda_{\rm min}(\widetilde{\B}_{22}) \Vert \hat{{\bf v}}^{\flat }\Vert_2^2.
	\end{equation}
	From (\ref{eq:vfalt}), it is straightforward to show that 
	$\Vert\hat{{\bf v}}^{\flat}\Vert_2 = O_p(n^{-1/2})$. 
	Let $\hat{\V}^{\flat} = (\hat{\V}^{\flat}_{\cdot 1}, \dots,\hat{\V}^{\flat}_{\cdot G})$ be a  matrix such that  $\hat{{\bf v}}^{\flat} = {\rm vech} \left(\hat{\V}^{\flat}\hat{\V}^{\flat {\rm T}}\right)$. It follows that $\Vert\hat{\V}^{\flat}_{\cdot g}\Vert_2 = O_p(n^{-1/4}), (g=1,\dots,G)$.
	Let $${\tilde{\m}}^{\flat}_1 = \B_{11}^{-1} \frac{\sum_{i=1}^n \b_{1i}}{n},$$
	which minimizes $2\tilde{\m}_1\trans \sum_{i=1}^n \b_{1i} - n \tilde{\m}_1\trans \B_{11} \tilde{\m}_1 $.  Therefore, we define $\hat{\m}^{\flat}_1 = {\tilde{\m}}^{\flat}_1 - \B_{11}^{-1}\B_{12}   \hat{\v}^{\flat}$.  Finally, we define ${{\V}}^{\flat} = [{{v}}^{\flat}_{hg}]$
	as 
	$${{\V}}_{\cdot 1}^{\flat} = \frac{\hat{\m}_1^{\flat}- \sum_{g=2}^G  \sqrt{1/G}\hat{\V}^{\flat}_{\cdot g}}{\sqrt{1/G}} \mbox{ and } {{\V}}^{\flat}_{\cdot g} = \hat{\V}^{\flat}_{\cdot g}, g\neq 1.$$   
	Let $ {{\theta}}^{\flat}_{gh} = \sqrt{G} {{v}}^{\flat}_{hg} +  {\theta}_{0h}$ and ${{\btheta}}^{\flat}_g = \left({{\theta}}^{\flat}_{g1},\dots,{{\theta}}^{\flat}_{gd}\right)\trans$. It follows that  $\Vert{{\btheta}}^{\flat}_g - {\btheta}_0\Vert_2 =  O_p\left(n^{-1/4}\right)$.  
	Let ${{\bxi}}^{\flat} =\left( {{\btheta}}^{\flat}_1, \dots, {{\btheta}}^{\flat}_G\right)$.
	Then,  we have
	$${\rm EM}_n^{(k)} \geq R_{1n}\left({{\bxi}}^{\flat}, \ba_0\right) - R_{0n}(\hat{\bxi}_0, \ba_0) = \sup_{{\bf v} \in \mathcal{V}} 2{\bf v}\trans \sum_{i=1}^n \tilde{\b}_{2i} - n{\bf v}\trans \widetilde{\B}_{22} {\bf v} + o_p(1).$$
	Note that  ${\bf v} \in \mathcal{V}$ if and only if  $\sqrt{n}{\bf v} \in \mathcal{V}$. Hence, we have 
	$${\rm EM}_n^{(k)} = \sup_{{\bf v} \in \mathcal{V}} 2{\bf v}\trans \sum_{i=1}^n \tilde{\b}_{2i} / \sqrt{n} - {\bf v}\trans \widetilde{\B}_{22} {\bf v} + o_p(1).$$
	By the central limit theorem, we have $\sum_{i=1}^n \tilde{\b}_{2i} / \sqrt{n} \rightarrow {\rm N}(0, \widetilde{\B}_{22})$. Thus, we get 
	under $\mathbb{H}_0$, for any fixed $k$, we have as $n \rightarrow \infty$, 
	$${\rm EM}_n^{(k)} \overset{d}{\longrightarrow} \sup_{{\bf v} \in \mathcal{V}} 2{\bf v}\trans {\bf{w}}  - {\bf v}\trans\widetilde{\B}_{22} {\bf v}, $$
	where ${\bf{w}} = (w_1,\ldots,w_{d(d+1)/2})\trans$ is a zero-mean multivariate normal random vector with a covariance matrix $\widetilde{\B}_{22}$,  and thus we  prove the theorem.  
\end{proof}

\begin{proof}[Proof of Lemma  \ref{lem:V}]
	It is clear that $\mathcal{V}_{\balpha} \subset \mathcal{V}$. Hence, we aim to prove $\mathcal{V} \subset \mathcal{V}_{\balpha}$. Without loss of generality, we assume $\alpha_G \neq 0$.  
	Let ${\bf M} = (\A_{\cdot 1}, \dots,\A_{\cdot (G-1)})$. Then, $\A$ can be rewritten as $\A = ({\bf M}, \A_{\cdot G})$. Since $\sum_{g=1}^G  {{\alpha}_g} {\A}_{\cdot g} = 0$, we have $ \A_{\cdot G} = -\sum_{g=1}^{G-1}  {{\alpha}_g} {\A}_{\cdot g} /{{\alpha}_G} $. Let 
	$$\bm{\beta} = -({{\alpha}_1}/{{\alpha}_G}, \ldots , {{\alpha}_{G-1}}/{{\alpha}_G})\trans.$$  
	Then, $ \A_{\cdot G}$ can be rewritten as $\A_{\cdot G} = {\bf M}\bm{\beta}$. It follows that $\A = ({\bf M},{\bf M}\bm{\beta})$ and 
	$$\A\A\trans = {\bf M}{\bf M}\trans + {\bf M}\bm{\beta} \bm{\beta}\trans {\bf M}\trans = {\bf M}(\I + \bm{\beta} \bm{\beta}\trans){\bf M}\trans.$$ 
	Since the minimum eigenvalue of $\I + \bm{\beta} \bm{\beta}\trans$ is greater than or equal to 1, $\I + \bm{\beta} \bm{\beta}\trans$ is positive definite. Hence, there exists a full rank matrix ${\bf Q}$ such that $\I + \bm{\beta} \bm{\beta}\trans = {\bf Q}{\bf Q}\trans$.
	Then $\A\A\trans = {\bf M}{\bf Q}{\bf Q}\trans {\bf M}\trans$.
	Therefore, for any $\bV \in \mathcal{V}$, we aim to prove that there exists ${\bf M}$ such that $\V = {\bf M}{\bf Q}{\bf Q}\trans {\bf M}\trans$. When $d \geq G-1$,  since $\V$ is a positive semi-definite matrix and ${\rm rank}(\bV) \leq r$, by eigenvalue decomposition theorem, there exists a diagonal matrix $\D = \diag(\lambda_1, \dots, \lambda_{G-1}), \lambda_1\geq \cdots \geq \lambda_{G-1}$ and an orthogonal matrix $\bP \in \mathbb{R}^{d \times (G-1)}$ such that 
	$\V = \bP\D\bP\trans$. Taking ${\bf M} = \bP\D^{1/2}\bQ^{-1}$ yields this lemma. 
	When $d < G-1$, similarly, we have $\V = \bP\D\bP\trans$, where $\D = \diag(\lambda_1, \dots, \lambda_{d})$ and $\bP \in  \mathbb{R}^{d \times d}$.  Write 
	$$\bP^{\dagger} = (\bP, {\bf 0}) \in \mathbb{R}^{d \times (G-1)}  \mbox{ and } \D^{\dagger} = \diag(\lambda_1, \dots, \lambda_{d}, \dots, 0) \in\mathbb{R}^{ (G-1)\times (G-1)} .$$
	Taking ${\bf M} = \bP^{\dagger}\D^{^{\dagger}1/2}\bQ^{-1}$ yields this lemma. 
\end{proof}

\section{Examples}\label{appC}
In this section, we give distribution examples that satisfy  Condition (\Ca)--(\Ch). Condition (\Ca)--(\Cb) and Condition (\Cg)--(\Ch) are easy to meet.  The following sections show that when the initial value $\ba^{(0)}$ is close to the true value $\ba^{*}$, Condition (\Cf) holds. We first mainly discuss Condition (\Cc) and Condition (\Cd). 
\begin{example}[Exponential families]
	Assume that $x$ is from a canonical exponential family with density
	$f(x; \btheta) = \exp \left\{ \btheta\trans T(x) - \xi(\btheta)  \right\}h(x),$
	where  $ \btheta \in \Theta$ is a  $d$-dimensional vector, $\Theta$ is a convex compact subset of the natural parameter space and $\xi(\btheta)$ is a smooth function. This family contains  most of the commonly used distributions, such as the Poisson distribution and the exponential distribution. For Condition (\Cc),  we show that $m$ can be taken as any positive integer.  When $T(x), {\rm vech }\left(T(x) T(x)\trans \right)$ are  linearly independent, we show that the covariance matrix $\B$ is positive definite, and thus Condition (\Cd) fulfills.   
	%		The proofs are in the supplementary material.
	As an example, consider the  Poisson distribution.  In such a case, we have $T(x) = x$. Since $x,x^2$ are linearly independent,  the covariance matrix $\B$ is positive definite.  Similarly,  for the gamma distribution,  we have $T(x) = (\log x , x)$.  Using the same argument, we can verify   Condition (\Cd). In addition, many exponential family  distributions, such as the Poisson distribution and the gamma distribution  satisfy the assumption that $\mathcal{P}^G$ is an identifiable finite mixture \citep{Yakowitz1968OnTI, BarndorffNielsen1965IdentifiabilityOM}.   
\end{example}

\begin{example}[Negative binomial model]
	The negative binomial distribution $x \sim {\rm NB} (\mu,r)$ has a probability mass function
	\begin{equation}\label{eq:pmfnb}
		\Pr(x = k) = \frac{\Gamma(r+k)}{k! \Gamma(r)} \left(\frac{r}{r+\mu} \right)^r\left(\frac{\mu}{r+\mu} \right)^k, \mbox{ for } k=0,1,2,\ldots ,
	\end{equation}
	where $\mu$ is the mean parameter and $r$ is the size parameter.  In such a case,  we let $\Theta = \{  (\mu, r): 0 < \delta_1 \leq \mu, r  \leq \delta_2 < \infty\}$ be a compact set, where $\delta_1  $ and $\delta_2$ are two constants.  Similarly, we can show that $m$ in Condition (\Cc) can be taken as any positive integer and the negative binomial distribution satisfies Condition (\Cd) and the identifiability assumption \citep{Yakowitz1968OnTI}.  
\end{example}
Our next goal aims to verify  Condition (\Ca)--(\Ch) in the above two examples.  In this section, we use $C,C' > 0$ as a generic 
constant, which may change from occurrence to occurrence.   
The following lemma gives an upper bound for the sums of independent sub-exponential random variables.
\begin{lemma} \label{lem:hdp}
	Let $X_1,\dots,X_n$ be independent mean-zero sub-exponential random variables. Then, we have 
	$$\left\|\sum_{i=1}^n X_i\right\|^2_{\psi_1} \leq C \sum_{i=1}^n \Vert X_i\Vert_{\psi_1}^2,$$ 
	where $C$ is a constant.
\end{lemma} 

\begin{proof}[Proof of Lemma  \ref{lem:hdp}]
	Since $X_i$ is a mean-zero sub-exponential  random variable,   
	there exists $c>0$ such that 
	$$\Ex\left(\exp(\lambda X_i)\right) \leq \exp(c\Vert X_i\Vert_{\psi_1}^2 \lambda^2 ), 
	|\lambda| \leq \frac{1}{\sqrt{c} \Vert X_i\Vert_{\psi_1}}.
	$$
	By independence,  we have
	$$\Ex\left(\exp\left(\lambda \sum_{i=1}^{n} X_i\right)\right) \leq \exp\left(c \lambda^2 \sum_{i=1}^{n} \Vert X_i\Vert_{\psi_1}^2\right), |\lambda| \leq \frac{1}{\sqrt{c \sum_{i=1}^{n}\Vert X_i\Vert_{\psi_1}^2} } .$$
	It follows that there exists a constant $C>0$ such that 
	$$\left\|\sum_{i=1}^n X_i\right\|^2_{\psi_1} \leq C \sum_{i=1}^n \Vert X_i\Vert_{\psi_1}^2,$$
	which proves the lemma. 
\end{proof}
\subsection{A sufficient condition for Condition (\Cf)} \label{subsec:suffi}
In this subsection, we will give a sufficient condition for Condition (\Cf). 
Recall that
\begin{equation}\label{eq:Xi1}
	\Xi_{1} = \left\{ \bxi :  \max_{g \neq g'} \Vert\btheta_g - \btheta_{g'} \Vert_2 \geq \gamma , \bxi \in \Xi \right\}.
\end{equation} 
\begin{lemma}\label{lem:\Cf}
	Suppose that $ \Ex_{\balpha_1,\bxi_1} {\rm log } \ \varphi\left(x;\bxi_2, \ba_2\right)$ is continuous with respect to $ \balpha_1 , \bxi_1,\ba_2,  \bxi_2$ and
	$\Ex_{\balpha_1,\bxi_1} {\rm log} \ f(x;\btheta_0 )$ is continuous with respect to $\balpha_1 , \bxi_1,  \btheta_0$. 
	For any $\gamma>0$, there exists a constant $\tau(\gamma)>0$ such that  if $\left\| \ba^{(0)} - \ba^*  \right\|_2 \leq \tau(\gamma)$, then we have 
	$$\inf_{\bxi^{*} \in \Xi_1  }  \left\{ \sup_{\bxi \in \Xi}\Ex_{\balpha^*,\bxi^{*}} {\rm log } \ \varphi\left(x;\bxi, \ba^{(0)}\right)-  \sup_{\btheta_0 \in \Theta}
	\Ex_{\balpha^*,\bxi^{*}} {\rm log} \ f(x;\btheta_0 )\right\} = \varrho > 0.$$
\end{lemma}

\begin{proof}[Proof of Lemma  \ref{lem:\Cf}]
	We first prove that
	\begin{equation}\label{eq:SC61}
		\inf_{\bxi^{*} \in \Xi_1  }  \left\{ \Ex_{\balpha^*,\bxi^{*}} {\rm log } \ \varphi\left(x;\bxi^*, \ba^*\right)-  \sup_{\btheta_0 \in \Theta}
		\Ex_{\balpha^*,\bxi^{*}} {\rm log} \ f(x;\btheta_0 )\right\} :=  \widetilde{\varrho} > 0,
	\end{equation}
	Note that the above formula (\ref{eq:SC61}) can  be written as
	$$\inf_{\bxi^{*} \in \Xi_1  }  \inf_{\btheta_0 \in \Theta}    \left\{  \Ex_{\balpha^*,\bxi^{*}} {\rm log } \ \varphi\left(x;\bxi^*, \ba^*\right)- 
	\Ex_{\balpha^*,\bxi^{*}} {\rm log} \ f(x;\btheta_0 )\right\} . $$
	Since $\mathcal{P}^{G}$ is an identifiable finite mixture, applying Jensen's  inequality,  we have   for any $ \bxi^{*} \in \Xi_1 $  and $\btheta_0 \in \Theta$,
	$$ \Ex_{\balpha^*,\bxi^{*}} {\rm log } \ \varphi\left(x;\bxi^*, \ba^*\right)- 
	\Ex_{\balpha^*,\bxi^{*}} {\rm log} \ f(x;\btheta_0 ) > 0.$$
	By the continuity and the compactness of $\Xi_1,  \Theta$, we prove (\ref{eq:SC61}).  Then,
	we aim to prove that there exists a constant $\tau(\gamma)>0$ such that  if $\left\|  \ba^{(0)} -\ba^* \right\|_2 \leq \tau(\gamma)$, then 
	\begin{equation}\label{eq:SC62}
		\inf_{\bxi^{*} \in \Xi_1  }  \left\{ \sup_{\bxi \in \Xi}\Ex_{\balpha^*,\bxi^{*}} {\rm log } \ \varphi\left(x;\bxi, \ba^{(0)}\right)-   \Ex_{\balpha^*,\bxi^{*}} {\rm log } \ \varphi\left(x;\bxi^*, \ba^*\right)   \right\} \geq-  \widetilde{\varrho}/2. 	
	\end{equation}
	We only need to prove 
	$$
	\inf_{\bxi^{*} \in \Xi_1  }  \left\{ \Ex_{\balpha^*,\bxi^{*}} {\rm log } \ \varphi\left(x;\bxi^*, \ba^{(0)}\right)-   \Ex_{\balpha^*,\bxi^{*}} {\rm log } \ \varphi\left(x;\bxi^*, \ba^*\right)   \right\} \geq-  \widetilde{\varrho}/2. 	
	$$
	Using the fact that  $ \Ex_{\balpha_1,\bxi_1} {\rm log } \ \varphi\left(x;\bxi_2, \ba_2\right)$ is uniformly continuous on a compact set,  for $ \widetilde{\varrho}/2$, there exists a constant $\tau(\gamma)>0$ such that if $\left\| \ba^{(0)}-\ba^*  \right\|_2 \leq \tau(\gamma)$, 
	then for any  $\bxi^* \in \Xi_1$,  we have 
	$$  \left|\Ex_{\balpha^*,\bxi^{*}} {\rm log } \ \varphi\left(x;\bxi^*, \ba^{(0)}\right)-   \Ex_{\balpha^*,\bxi^{*}} {\rm log } \ \varphi\left(x;\bxi^*, \ba^*\right)   \right| \leq   \widetilde{\varrho}/2,$$
	which proves (\ref{eq:SC62}).  Combining (\ref{eq:SC61}) with (\ref{eq:SC62}) yields the result. 
\end{proof}
\subsection{Exponential families} \label{subsec:exp}
It is clear that Condition (\Ca) and (\Cb) hold.  
A sufficient condition for Condition (\Cf) is in \ref{subsec:suffi}.
The remainder of Section \ref{subsec:exp} will be devoted to verify Condition (\Cc)--(\Cd) and (\Cg)--(\Ch).

\subsubsection{Condition (\Cc)}
We first prove that for any $0<r\leq5, j_1, \dots, j_r \in \{1, \dots, d\}$ and $m>0$, 
\begin{equation}\label{eq:exp1}
	\left\| \sup_{\btheta \in \Theta} \left |\frac{ \partial^r }{\partial \theta_{j_1} \cdots \partial \theta_{j_r}} \log f(x; \btheta)\right|\right\|_{L^m} <  \infty. 
\end{equation}
Note that $\log f(x; \btheta) = \btheta\trans T(x) - \xi(\btheta)  + \log h(x)$, $\btheta \in \Theta \subset \mathbb{R}^d$ and $\Theta$ is a compact set. Then 
$$\frac{ \partial^r }{\partial \theta_{j_1} \cdots \partial \theta_{j_r}} \log f(x; \btheta) = \sum_{\ell = 1}^d T_\ell(x) {\mathbb{I}(r = 1)} +  \frac{ \partial^r }{\partial \theta_{j_1} \cdots \partial \theta_{j_r}} \xi(\btheta),$$
where ${\mathbb{I}(\cdot)} $ is the indicator function.
Since $\left\| T_\ell(x) \right\|_{L^m} < \infty$ and $ \frac{ \partial^r }{\partial \theta_{j_1} \cdots \partial \theta_{j_r}} \xi(\btheta)$ are bounded in $\Theta$,  the conclusion follows. Our next goal is to prove for any $m>0, 0<r\leq5$,  
$$\left\| \sup_{\btheta \in \Theta} \left| \frac{ \partial^r }{\partial \theta_{j_1} \cdots \partial \theta_{j_r}}  f(x; \btheta)  \bigg/ f(x; \btheta)\right|\right\|_{L^m} < \infty.$$ 
In fact, it is a direct consequence of (\ref{eq:exp1}).  To prove this, when $r = 2$, we can write $\frac{ \partial^2 }{\partial \theta_{j_1}  \partial \theta_{j_2}}  f(x; \btheta)  \big/ f(x; \btheta) $ as 
$$ \frac{ \partial^2 }{\partial \theta_{j_1}  \partial \theta_{j_2}}  f(x; \btheta)  \bigg/ f(x; \btheta)  = \frac{ \partial^2 }{\partial \theta_{j_1}  \partial \theta_{j_2}}  \log f(x; \btheta)  +  \frac{ \partial }{\partial \theta_{j_1}}  \log f(x; \btheta)  \frac{ \partial}{\partial \theta_{j_2}}  \log f(x; \btheta) .$$
By (\ref{eq:exp1}), we prove the result.  The same reasoning applies to the case $3 \leq r \leq 5$. 
In order to verify Condition (\Cc),  we note that
\begin{align*}
	&\quad \sup_{\Vert\btheta - \btheta_0\Vert_2 \leq \tau} \left|\frac{\partial^r f(x; \btheta)}{\partial \theta_{j_1} \cdots \partial \theta_{j_r}} \bigg/  f(x; \btheta_0)\right| \\
	&  \leq  \sup_{\Vert\btheta - \btheta_0\Vert_2 \leq \tau} \left|\frac{\partial^r f(x; \btheta)}{\partial \theta_{j_1} \cdots \partial \theta_{j_r}} \bigg/  f(x; \btheta)\right| \sup_{\Vert\btheta - \btheta_0\Vert_2 \leq \tau} \left|f(x; \btheta) /  f(x; \btheta_0)\right|.
\end{align*}
It remains to consider the function  
$\sup_{\Vert\btheta - \btheta_0\Vert_2 \leq \tau} \left|f(x; \btheta) /  f(x; \btheta_0)\right|.$ 
Let $C$ be a constant such that 
$$  \sup_{\btheta \in \Theta} \left\|  \frac{\partial }{\partial \btheta } \xi(\btheta)    \right\|_2 \leq C. $$ 
By Lagrange's theorem, if $\Vert\btheta - \btheta_0\Vert_2 \leq \tau ,$ we have 
\begin{align}\label{eq:extau}
	f(x; \btheta) /  f(x; \btheta_0)  &=  \exp \left\{ [\btheta -\btheta_0 ]\trans T(x) - (\xi(\btheta) - \xi(\btheta_0))  \right\} \notag \\
	&\leq  \exp \left\{ \tau \Vert T(x)\Vert_1 + C \tau  \right\}.
\end{align}
Since $\btheta_{0}$ is an interior point of the natural parameter space, for any $m>0$, by (\ref{eq:extau}), there exists a $\tau>0$ such that
\begin{equation}\label{eq:expdiv}
	\left\|\sup_{\Vert\btheta - \btheta_0\Vert_2 \leq \tau} \left|f(x; \btheta) /  f(x; \btheta_0)\right| \right\|_{L^m} < \infty.
\end{equation}
Further, since (\ref{eq:extau}) is independent of $\btheta_{0}$, this  gives that  for any $0<r\leq5, j_1, \dots, j_r \in \{1, \dots, d\}$ and $m>0$, there exists a $\tau>0$ such that 
$$\sup_{\btheta_0 \in \Theta}\left\Vert  \sup_{\Vert\btheta - \btheta_0\Vert_2 \leq \tau} \left|\frac{\partial^r f(x; \btheta)}{\partial \theta_{j_1} \cdots \partial \theta_{j_r}} \bigg /  f(x; \btheta_0)\right| \right\|_{L^m} < \infty.$$

Our next goal is to prove $r(x)$ exists.  We first prove that 
\begin{equation}\label{eq:r1}
	\int\sup_{\btheta \in \Theta}{{f(x;\btheta)}} \mu(\rd x) < \infty.
\end{equation}
For any $\btheta_0$,  (\ref{eq:expdiv}) shows that there is a $\tau>0$ such that 
$$
\int\sup_{\Vert\btheta - \btheta_0\Vert_2 \leq \tau} f(x; \btheta) \mu(\rd x) < \infty.
$$
Let $U(\btheta, \tau) = \{\btheta': \Vert\btheta' - \btheta\Vert_2 < \tau\}$.  Since $\bTheta$ is a compact set,  by the Heine-Borel theorem,   open cover 
$ \{ U(\btheta, \tau(\btheta)), \btheta \in \bTheta \} $ has a finite subcover $ \{ U(\btheta_j, \tau(\btheta_j))\}_{j=1}^J.$ 
It follows that 
$$\sup_{\btheta \in \Theta}{{f(x;\btheta)}}  \leq \sum_{j=1}^J \sup_{\Vert\btheta - \btheta_j\Vert_2 \leq \tau(\btheta_j)} f(x; \btheta),$$
and thus we have
$$  \int\sup_{\btheta \in \Theta}{{f(x;\btheta)}}\mu(\rd x) < \infty.$$
It remains to show that 
$$\int \sup_{\btheta \in \Theta} \frac{1}{{f(x;\btheta)}}\left\| \frac{\partial f(x;\btheta)}{\partial \btheta} \right\|_2^2  \mu(\rd x)< \infty.$$
To this end, we only need to prove that for any $\ell$,  
\begin{equation}\label{eq:r2}
	\int \sup_{\btheta \in \Theta} \frac{1}{{f(x;\btheta)}} \left | \frac{\partial f(x;\btheta)}{\partial \theta_\ell} \right | ^2  \mu(\rd x)< \infty.
\end{equation}
Note that 
$$
\frac{1}{{f(x;\btheta)}}\left | \frac{\partial f(x;\btheta)}{\partial \theta_\ell} \right | ^2  = f(x;\btheta)\left | \frac{\partial \log f(x;\btheta)}{\partial \theta_\ell} \right | ^2.
$$
Since 
$
\frac{\partial \log f(x;\btheta)}{\partial \theta_\ell} = T_\ell(x) + \frac{\partial \xi(\theta)}{\partial \theta_\ell}
$
and there exists a constant $C$ such that 
$
\sup_{\btheta \in \Theta} \left|\frac{\partial \xi(\btheta)}{\partial \theta_\ell}\right| \leq C,
$
it follows that 
$$
\sup_{\btheta \in \Theta} \frac{1}{{f(x;\btheta)}}\left | \frac{\partial f(x;\btheta)}{\partial \theta_\ell} \right | ^2 \leq (|T_\ell(x)| + C)^2 \sup_{\btheta \in \Theta}  f(x;\btheta).
$$
Similarly, for any $\btheta_0$, by (\ref{eq:extau}), there exists a $\tau>0$ such that 
\begin{equation}\label{eq:r3}
	\int\sup_{\Vert\btheta - \btheta_0\Vert_2 \leq \tau} f(x; \btheta)  (|T_\ell(x)| + C)^2 \mu(\rd x) < \infty.
\end{equation}
With (\ref{eq:r3}), applying the Heine-Borel theorem and using the same argument as in the proof of (\ref{eq:r1}), we can easily prove (\ref{eq:r2}).
Therefore, we verify Condition (\Cc).

\subsubsection{Condition (\Cd)}
We claim that if $\left( T(x), {\rm vech}\left(T(x)T(x)\trans\right)  \right) $ are linearly independent, then  Condition (\Cd) fulfills.
Observe that 
$$
\frac{ \partial  f(x; \btheta)}{\partial\btheta}  \bigg/  f(x; \btheta)=T(x)+  \frac{ \partial}{\partial \btheta} \xi(\btheta),
$$ 
and 
$$
\frac{ \partial^2 f(x; \btheta)}{\partial\btheta \btheta\trans}   \bigg/ f(x;\btheta)=\left(T(x)+  \frac{ \partial}{\partial \btheta} \xi(\btheta)\right)\left(T(x)+  \frac{ \partial}{\partial \btheta} \xi(\btheta)\right)\trans + \frac{ \partial^2 }{\partial\btheta \btheta\trans} \xi(\btheta).
$$
Hence, to prove $\B(\btheta_{0}) = {\rm Cov} (\b)$ is positive definite, we only need  that the covariance of 
$$\left( \frac{ \partial  f(x; \btheta)}{\partial\btheta}  \bigg/  f(x; \btheta), {\rm vech}\left( \frac{ \partial^2 f(x; \btheta)}{\partial\btheta \btheta\trans}   \bigg/ f(x;\btheta) \right) \right)$$
is positive definite.  Note that $ \frac{ \partial}{\partial \btheta} \xi(\btheta)$ and $\frac{ \partial^2 }{\partial\btheta \btheta\trans} \xi(\btheta)$ are independent of $x$.  Thus, it suffices to show 
$$\left( T(x), {\rm vech}\left(\left(T(x)+  \frac{ \partial}{\partial \btheta} \xi(\btheta)\right)\left(T(x)+  \frac{ \partial}{\partial \btheta} \xi(\btheta)\right)\trans\right)  \right) $$ 
are linearly independent.  
However, it is equivalent to $\left( T(x), {\rm vech}\left(T(x)T(x)\trans\right)  \right)$ are linearly independent,
and thus $\B(\btheta_{0})$ is positive definite.
Finally, using the fact that $\lambda_{\rm min}(\B(\btheta_{0}))$ is continuous and $\Theta$ is compact, we verify  Condition (\Cd). 
\subsubsection{Condition (\Cg)}
To verify Condition (\Cg), we only need to show that  there exists $M_{\psi_1}$ such that
$$\left\| \log\left(\sum_{g=1}^{G} \balpha_g^{\left(0\right)}f\left(x;{\btheta}^{\dagger}_g\right)\right)  -\log\left(f\left(x;{\btheta}_0^{\dagger}\right)\right) \right\|_{\psi_1} \leq M_{\psi_1}. $$
Applying (\ref{eq:extau}), it follows that 
\begin{equation*}
	\sum_{g=1}^{G} \balpha_g^{\left(0\right)}f\left(x;{\btheta}^{\dagger}_g\right) \big/ f\left(x;{\btheta}_0^{\dagger}\right)  
	\leq  \exp \left\{C {\rm diam(\Theta)} \Vert T(x)\Vert_1  + C {\rm diam(\Theta)}\right\} .
\end{equation*}
Thus, there exists a $t$ such that 
$$\Ex\left[ \left(\sum_{g=1}^{G} \balpha_g^{\left(0\right)}f\left(x;{\btheta}^{\dagger}_g\right) \big / f\left(x;{\btheta}_0^{\dagger}\right) \right)^t\right] < \infty,$$
and thus we verify Condition (\Cg).  

\subsubsection{Condition (\Ch)}
Finally, we aim to verify  Condition (\Ch).  Note that 
$$  Z_{\btheta}(\bxi^*) -  Z_{\btheta'}(\bxi^*)=   n^{-1/2} \sum_{i=1}^n \left\lbrace  \log f(x;\btheta) - \log f(x; \btheta') - (D(\btheta )- D(\btheta'))  \right\rbrace.$$
Then, by Lemma \ref{lem:hdp}, we have 
\begin{align*}
	C \left\| Z_{\btheta}(\bxi^*) -  Z_{\btheta'}(\bxi^*) \right\|_{\psi_1} &\leq  \left\|  [\btheta - \btheta' ]\trans T(x)  \right\|_{\psi_1}  + \left\|\xi(\btheta )- \xi(\btheta') \right\|_{\psi_1}  + \left\|D(\btheta )- D(\btheta') \right\|_{\psi_1}  \\
	&\leq   \left\|  \sum_{\ell=1}^d [\btheta_\ell- \btheta'_\ell] T_\ell(x)  \right\|_{\psi_1} + C' \left\|\btheta - \btheta' \right\|_2 \\
	&\leq  C'' \left\|\btheta - \btheta' \right\|_2  \sum_{\ell=1}^d  \left\| T_\ell(x)  \right\|_{\psi_1} + C' \left\|\btheta - \btheta' \right\|_2,
\end{align*}
where the second inequality is from the fact that $\xi(\btheta )$ and $D(\btheta ) $ have continuous derivative functions.  
Since $\Xi$ is a compact set, there is a constant $C>0$ such that $\sum_{\ell=1}^d  \left\| T_\ell(x)  \right\|_{\psi_1} \leq C$,
and thus we verify Condition (\Ch). 

\subsection{Negative binomial model}
The same proof remains valid for the negative binomial example, and thus we omit it.  
\section{Details for the  simulation data generation}
For the low-noise and high-noise scenarios, we independently generate $r_j$ from the uniform distributions ${\rm U}(10,11)$ and ${\rm U}(5,6)$, respectively. For the clustering-relevant features ($j=1,\dots,20$), the mean parameters $\mu_{gj}$ are either set as $\exp(u_j)$ or $\exp(u_j)+D_j$, where $u_j$ is generated from $ {\rm U} (\log \ 2, \log \ 5)$, and $D_j$ is to control the signal strength (the differences between clusters). We generate $D_j$ from ${\rm U}(5,6)$, ${\rm U}(7,8)$ or ${\rm U}(9,10)$ for the low, medium and high signal strength settings, respectively. For the first 5 features ($1\leq j\leq 5$), we set $\mu_{2j}=\exp(u_j)+D_j$ and $\mu_{gj} = \exp(u_j) (g\neq 2)$. Similarly, for $5k+1\leq j\leq 5k+5 (k=1,2,3)$, we set $\mu_{k+2,j}=\exp(u_j)+D_j$ and $\mu_{gj} = \exp(u_j) (g\neq k+2)$. For all cluster-irrelevant features ($j=21,\dots,p$), we set $\mu_j=\exp(u_j)$, where $u_j$ is generated from $ {\rm U} (\log \ 2, \log \ 5)$. 
\section{Additional simulation results}
In this section, we present the additional simulation results.

\subsection{Simulations for EM-test with  mis-specified group number $G$}
\begin{table}[htbp]
	\centering
	\caption{The means and standard deviations (in parenthesis) of ARIs over 100 replications by EM-test with mis-specified $G$.  The values in the table are shown as the actual values $\times$ 100.  Simulation are generated from the negative binomial model (Section 4.1 in the main manuscript). The true number of clusters is 5. EM-adjust means that  the features are selected by the adjusted p-values and EM-0.35 means that we choose the threshold as $n^{0.35}$.}
	\renewcommand\arraystretch {1.2}
	\footnotesize
	\setlength{\tabcolsep}{3mm}
	\begin{tabular}{ccccccccc}
		\hline
		\hline
		&       & \multicolumn{3}{c}{EM-adjust} &       & \multicolumn{3}{c}{EM-0.35} \\
		\cline{3-5} \cline{7-9} 	
		&       & $G=5$  (True) & $G=2$   & $G=8$   &       & $G=5$  (True) & $G=2$   & $G=8$ \\
		\hline 
		\multicolumn{9}{c}{Case 1:  High signal  and low noise} \\  
		\hline 
		$p=500$ &       & 98 (1.5) & 98 (1.5) & 98 (1.5) &       & 97 (2.7) & 98 (1.3) & 97 (2.7) \\
		$p=5000$ &       & 98 (0.9) & 98 (0.9) & 98 (0.9) &       & 97 (1.5) & 97 (1.6) & 97 (1.5) \\
		$p=20,000$ &       & 97 (2.9) & 97 (2.9) & 97 (2.9) &       & 97 (1.3) & 98 (0.9) & 98 (0.9) \\
		\hline 
		\multicolumn{9}{c}{Case 2:  High signal  and high noise} \\  
		\hline 
		$p=500$ &       & 94 (2.8) & 94 (2.8) & 94 (2.8) &       & 94 (3.0) & 94 (1.3) & 94 (2.9) \\
		$p=5000$ &       & 94 (1.7) & 94 (1.7) & 94 (1.7) &       & 94 (1.7) & 94 (1.6) & 94 (2.0) \\
		$p=20,000$ &       & 94 (1.9) & 93 (2.9) & 94 (1.9) &       & 93 (1.5) & 93 (1.4) & 93 (1.5) \\
		\hline 
		\multicolumn{9}{c}{Case 3:  Medium signal  and low noise} \\  
		\hline 
		$p=500$ &       & 95 (1.8) & 95 (1.8) & 95 (1.8) &       & 95 (1.9) & 95 (1.9) & 95 (2.0) \\
		$p=5000$ &       & 95 (1.2) & 95 (1.2) & 95 (1.2) &       & 96 (1.1) & 96 (1.1) & 95 (1.1) \\
		$p=20,000$ &       & 95 (1.9) & 95 (1.9) & 95 (1.9) &       & 95 (1.3) & 95 (1.1) & 95 (1.3) \\
		\hline 
		\multicolumn{9}{c}{Case 4:  Medium signal  and high noise} \\  
		\hline 
		$p=500$ &       & 90 (2.0) & 90 (2.0) & 90 (2.0) &       & 90 (1.7) & 90 (1.7) & 90 (1.7) \\
		$p=5000$ &       & 90 (2.5) & 90 (2.6) & 90 (2.5) &       & 90 (2.1) & 90 (2.2) & 90 (2.1) \\
		$p=20,000$ &       & 88 (3.0) & 88 (3.0) & 88 (2.9) &       & 89 (1.8) & 89 (1.7) & 89 (1.8) \\
		\hline 
		\multicolumn{9}{c}{Case 5:  Low signal  and low noise} \\  
		\hline 
		$p=500$ &       & 84 (7.6) & 83 (7.7) & 84 (7.6) &       & 88 (2.6) & 88 (2.8) & 88 (2.6) \\
		$p=5000$ &       & 78 (7.1) & 77 (7.3) & 79 (7.1) &       & 88 (2.6) & 88 (2.6) & 88 (2.6) \\
		$p=20,000$ &       & 74 (8.7) & 73 (9.7) & 75 (8.4) &       & 87 (3.0) & 87 (2.9) & 87 (3.1) \\
		\hline 
		\multicolumn{9}{c}{Case 6:  Low signal  and high noise} \\  
		\hline 
		$p=500$ &       & 73 (6.5) & 71 (6.7) & 73 (6.4) &       & 80 (3.6) & 79 (4.5) & 80 (3.5) \\
		$p=5000$ &       & 66 (8.3) & 65 (9.0) & 66 (8.2) &       & 79 (3.9) & 79 (4.0) & 79 (3.9) \\
		$p=20,000$ &       & 59 (11.7) & 57 (12.3) & 59 (11.9) &       & 72 (8.9) & 75 (6.8) & 72 (8.8) \\
		\hline
	\end{tabular}%
	\label{tab:wrongGARI1}%
\end{table}%
% Table generated by Excel2LaTeX from sheet 'addition'
\begin{table}[htbp]
	\centering
	\caption{The mean of numbers of correctly retained features ($\mathcal{R}$) and falsely retained features ($\mathcal{F}$) by EM-test with mis-specified $G$. The true number of clusters is 5. Simulation are generated from the negative binomial model (Section 4.1 in the main manuscript).  EM-adjust means that  the features are selected by the adjusted p-values and EM-0.35 means that we choose the threshold as $n^{0.35}$.}
	\renewcommand\arraystretch {1.1}
	\footnotesize
	\setlength{\tabcolsep}{3mm}
	\begin{tabular}{ccccccccc}
		\hline
		\hline
		&       & \multicolumn{3}{c}{EM-adjust} &       & \multicolumn{3}{c}{EM-0.35} \\	
		\cline{3-5} \cline{7-9} 	
		&       & $G=5$ (True)  & $G=2$   & $G=8$   &       & $G=5$  (True) & $G=2$   & $G=8$ \\
		\hline 
		\multicolumn{9}{c}{Case 1:  High signal  and low noise} \\  
		\hline 
		\multirow{2}[0]{*}{$p=500$} & $\mathcal{R}$  & 20.0 (0.1) & 20.0 (0.1) & 20.0 (0.1) &       & 20.0 (0.1) & 20.0 (0.1) & 20.0 (0.1) \\
		& $\mathcal{F}$  & 0.0 (0.1) & 0.0 (0.1) & 0.0 (0.1) &       & 1.2 (1.1) & 0.5 (0.7) & 1.3 (1.2) \\
		\multirow{2}[0]{*}{$p=5000$} & $\mathcal{R}$  & 20.0 (0.0) & 20.0 (0.0) & 20.0 (0.0) &       & 20.0 (0.0) & 20.0 (0.0) & 20.0 (0.0) \\
		& $\mathcal{F}$  & 0.0 (0.1) & 0.0 (0.1) & 0.0 (0.1) &       & 10.2 (3.3) & 6.3 (2.5) & 11.5 (3.5) \\
		\multirow{2}[0]{*}{$p=20,000$} & $\mathcal{R}$  & 20.0 (0.0) & 20.0 (0.0) & 20.0 (0.0) &       & 20.0 (0.0) & 20.0 (0.0) & 20.0 (0.0) \\
		& $\mathcal{F}$  & 0.0 (0.2) & 0.0 (0.1) & 0.0 (0.2) &       & 40.5 (5.6) & 26.1 (4.4) & 46.7 (5.9) \\
		\hline 
		\multicolumn{9}{c}{Case 2:  High signal  and high noise} \\  
		\hline 
		\multirow{2}[0]{*}{$p=500$} & $\mathcal{R}$  & 20.0 (0.1) & 20.0 (0.1) & 20.0 (0.1) &       & 20.0 (0.0) & 20.0 (0.0) & 20.0 (0.0) \\
		& $\mathcal{F}$  & 0.1 (0.2) & 0.0 (0.1) & 0.1 (0.3) &       & 2.0 (1.4) & 1.1 (1.0) & 2.1 (1.5) \\
		\multirow{2}[0]{*}{$p=5000$} & $\mathcal{R}$  & 20.0 (0.0) & 20.0 (0.0) & 20.0 (0.0) &       & 20.0 (0.0) & 20.0 (0.0) & 20.0 (0.0) \\
		& $\mathcal{F}$  & 0.0 (0.1) & 0.0 (0.0) & 0.0 (0.2) &       & 17.9 (4.1) & 12.1 (3.5) & 20.1 (4.5) \\
		\multirow{2}[0]{*}{$p=20,000$} & $\mathcal{R}$  & 20.0 (0.2) & 20.0 (0.2) & 20.0 (0.2) &       & 20.0 (0.0) & 20.0 (0.0) & 20.0 (0.0) \\
		& $\mathcal{F}$  & 0.1 (0.3) & 0.0 (0.2) & 0.1 (0.3) &       & 74.9 (7.7) & 49.1 (6.1) & 84.2 (7.9) \\
		\hline 
		\multicolumn{9}{c}{Case 3:  Medium signal  and low noise} \\  
		\hline 
		\multirow{2}[0]{*}{$p=500$} & $\mathcal{R}$  & 19.9 (0.3) & 19.9 (0.3) & 19.9 (0.3) &       & 20.0 (0.1) & 20.0 (0.1) & 19.9 (0.2) \\
		& $\mathcal{F}$  & 0.0 (0.1) & 0.0 (0.1) & 0.0 (0.1) &       & 1.2 (1.1) & 0.5 (0.7) & 1.2 (1.2) \\
		\multirow{2}[0]{*}{$p=5000$} & $\mathcal{R}$  & 19.8 (0.4) & 19.8 (0.4) & 19.8 (0.4) &       & 20.0 (0.1) & 20.0 (0.1) & 20.0 (0.1) \\
		& $\mathcal{F}$  & 0.0 (0.1) & 0.0 (0.1) & 0.0 (0.1) &       & 10.1 (3.2) & 6.4 (2.6) & 11.6 (3.6) \\
		\multirow{2}[0]{*}{$p=20,000$} & $\mathcal{R}$  & 19.6 (0.6) & 19.6 (0.6) & 19.6 (0.6) &       & 20.0 (0.0) & 20.0 (0.0) & 20.0 (0.0) \\
		& $\mathcal{F}$  & 0.1 (0.2) & 0.0 (0.1) & 0.1 (0.2) &       & 40.8 (5.7) & 26.4 (4.5) & 46.9 (6.0) \\
		\hline 
		\multicolumn{9}{c}{Case 4:  Medium signal  and high noise} \\  
		\hline 
		\multirow{2}[0]{*}{$p=500$} & $\mathcal{R}$  & 19.8 (0.5) & 19.7 (0.5) & 19.8 (0.5) &       & 20.0 (0.1) & 20.0 (0.2) & 20.0 (0.1) \\
		& $\mathcal{F}$  & 0.0 (0.1) & 0.0 (0.1) & 0.1 (0.3) &       & 2.0 (1.4) & 1.3 (1.1) & 2.2 (1.5) \\
		\multirow{2}[0]{*}{$p=5000$} & $\mathcal{R}$  & 19.2 (0.9) & 19.1 (0.9) & 19.2 (0.9) &       & 20.0 (0.1) & 20.0 (0.1) & 20.0 (0.1) \\
		& $\mathcal{F}$  & 0.0 (0.2) & 0.0 (0.0) & 0.0 (0.2) &       & 17.9 (4.1) & 12.1 (3.5) & 20.0 (4.4) \\
		\multirow{2}[0]{*}{$p=20,000$} & $\mathcal{R}$  & 18.7 (1.2) & 18.6 (1.2) & 18.8 (1.1) &       & 20.0 (0.2) & 19.9 (0.2) & 19.9 (0.2) \\
		& $\mathcal{F}$  & 0.1 (0.3) & 0.1 (0.2) & 0.1 (0.3) &       & 74.7 (7.5) & 48.9 (6.0) & 83.7 (8.2) \\
		\hline 
		\multicolumn{9}{c}{Case 5:  Low signal  and low noise} \\  
		\hline 
		\multirow{2}[0]{*}{$p=500$} & $\mathcal{R}$  & 16.5 (2.0) & 15.9 (1.9) & 16.6 (2.0) &       & 18.9 (1.0) & 18.7 (1.1) & 19.0 (0.9) \\
		& $\mathcal{F}$  & 0.0 (0.1) & 0.0 (0.1) & 0.0 (0.1) &       & 1.1 (1.1) & 0.7 (0.7) & 1.3 (1.1) \\
		\multirow{2}[0]{*}{$p=5000$} & $\mathcal{R}$  & 13.7 (2.0) & 13.3 (2.1) & 13.8 (2.1) &       & 19.0 (1.0) & 18.9 (1.0) & 19.0 (1.0) \\
		& $\mathcal{F}$  & 0.0 (0.1) & 0.0 (0.0) & 0.0 (0.1) &       & 10.1 (3.3) & 6.3 (2.6) & 11.6 (3.6) \\
		\multirow{2}[0]{*}{$p=20,000$} & $\mathcal{R}$  & 12.1 (2.3) & 11.8 (2.4) & 12.3 (2.3) &       & 18.9 (1.1) & 18.7 (1.1) & 18.9 (1.1) \\
		& $\mathcal{F}$  & 0.0 (0.1) & 0.0 (0.0) & 0.0 (0.2) &       & 40.5 (5.6) & 26.4 (4.6) & 46.9 (5.8) \\
		\hline 
		\multicolumn{9}{c}{Case 6:  Low signal  and high noise} \\  
		\hline 
		\multirow{2}[0]{*}{$p=500$} & $\mathcal{R}$  & 14.7 (2.1) & 14.0 (2.0) & 14.9 (2.1) &       & 18.4 (1.2) & 17.9 (1.5) & 18.5 (1.1) \\
		& $\mathcal{F}$  & 0.0 (0.1) & 0.0 (0.1) & 0.0 (0.2) &       & 1.9 (1.4) & 1.3 (1.0) & 2.1 (1.5) \\
		\multirow{2}[0]{*}{$p=5000$} & $\mathcal{R}$  & 12.0 (2.3) & 11.6 (2.4) & 12.0 (2.2) &       & 18.4 (1.3) & 18.1 (1.3) & 18.4 (1.2) \\
		& $\mathcal{F}$  & 0.0 (0.1) & 0.0 (0.0) & 0.0 (0.1) &       & 18.1 (4.2) & 12.1 (3.5) & 20.0 (4.5) \\
		\multirow{2}[0]{*}{$p=20,000$} & $\mathcal{R}$  & 10.0 (2.5) & 9.6 (2.5) & 10.1 (2.5) &       & 18.1 (1.4) & 17.8 (1.4) & 18.1 (1.4) \\
		& $\mathcal{F}$  & 0.1 (0.3) & 0.0 (0.2) & 0.1 (0.3) &       & 74.9 (7.4) & 49.0 (5.9) & 84.0 (7.9) \\
		\hline
	\end{tabular}%
	\label{tab:wrongGPN1}%
\end{table}%
\newpage
\subsection{Simulations for EM-test with different penalties $\lambda$}

\begin{table}[htbp]
	\centering
	\caption{The mean of numbers of correctly retained features ($\mathcal{R}$) and falsely retained features ($\mathcal{F}$) over 100 replications by  EM-test with different penalties $\lambda$. Simulation are generated from the negative binomial model (Section 4.1 in the main manuscript).}
	\renewcommand\arraystretch {1.2}
	\footnotesize
	\setlength{\tabcolsep}{3mm}
	\begin{tabular}{cccccccccc}
		\hline
		\hline
		&       & \multicolumn{2}{c}{$p=500$} &       & \multicolumn{2}{c}{$p=5000$} &       & \multicolumn{2}{c}{$p=20,000$} \\
		\cline{3-4} \cline{6-7} \cline{9-10}
		&       & EM-adjust & EM-0.35 &       & EM-adjust & EM-0.35 &       & EM-adjust & EM-0.35 \\
		\hline
		\multicolumn{10}{c}{ Medium signal  and high noise} \\  
		\hline
		\multirow{2}[0]{*}{$\lambda=10^{-7}$} & $\mathcal{R}$  & 19.8 (0.5) & 20.0 (0.1) &       & 19.2 (0.9) & 20.0 (0.1) &       & 18.7 (1.2) & 20.0 (0.2) \\
		& $\mathcal{F}$  & 0.0 (0.1) & 2.0 (1.4) &       & 0.0 (0.2) & 17.9 (4.1) &       & 0.1 (0.3) & 74.7 (7.5) \\
		\multirow{2}[0]{*}{$\lambda=10^{-5}$} & $\mathcal{R}$  & 19.8 (0.5) & 20.0 (0.1) &       & 19.2 (0.9) & 20.0 (0.1) &       & 18.7 (1.2) & 20.0 (0.2) \\
		& $\mathcal{F}$  & 0.0 (0.1) & 2.0 (1.4) &       & 0.0 (0.2) & 17.9 (4.1) &       & 0.1 (0.3) & 74.7 (7.5) \\
		\multirow{2}[0]{*}{$\lambda=10^{-3}$} & $\mathcal{R}$  & 19.8 (0.5) & 20.0 (0.1) &       & 19.2 (0.9) & 20.0 (0.1) &       & 18.7 (1.2) & 20.0 (0.2) \\
		& $\mathcal{F}$  & 0.0 (0.1) & 2.0 (1.4) &       & 0.0 (0.2) & 17.8 (4.1) &       & 0.1 (0.3) & 74.6 (7.5) \\
		\multirow{2}[0]{*}{$\lambda=10^{-1}$} & $\mathcal{R}$  & 19.7 (0.5) & 20.0 (0.1) &       & 19.2 (0.9) & 20.0 (0.1) &       & 18.7 (1.2) & 20.0 (0.2) \\
		& $\mathcal{F}$  & 0.0 (0.1) & 1.8 (1.3) &       & 0.0 (0.2) & 17.5 (4.1) &       & 0.1 (0.3) & 73.2 (7.4) \\
		\multirow{2}[0]{*}{$\lambda=1$} & $\mathcal{R}$  & 19.7 (0.5) & 20.0 (0.1) &       & 19.2 (0.9) & 20.0 (0.1) &       & 18.7 (1.2) & 20.0 (0.2) \\
		& $\mathcal{F}$  & 0.0 (0.1) & 1.6 (1.3) &       & 0.0 (0.1) & 15.8 (3.9) &       & 0.1 (0.3) & 65.8 (7.5) \\
		\multirow{2}[0]{*}{$\lambda=10$} & $\mathcal{R}$  & 19.7 (0.5) & 20.0 (0.1) &       & 19.2 (0.9) & 20.0 (0.1) &       & 18.7 (1.2) & 20.0 (0.2) \\
		& $\mathcal{F}$  & 0.0 (0.1) & 1.5 (1.3) &       & 0.0 (0.1) & 14.3 (3.7) &       & 0.1 (0.2) & 58.9 (7.0) \\
		\multirow{2}[0]{*}{$\lambda=100$} & $\mathcal{R}$  & 19.7 (0.5) & 20.0 (0.1) &       & 19.2 (0.9) & 20.0 (0.1) &       & 18.7 (1.2) & 20.0 (0.2) \\
		& $\mathcal{F}$  & 0.0 (0.1) & 1.5 (1.3) &       & 0.0 (0.1) & 14.0 (3.7) &       & 0.1 (0.2) & 58.1 (6.8) \\
		\hline
	\end{tabular}%
	\label{tab:lambda}%
\end{table}%
\subsection{Simulations for EM-test with different iteration steps $K$}
% Table generated by Excel2LaTeX from sheet 'addition'
\begin{table}[htbp]
	\centering
	\caption{The mean of numbers of correctly retained features ($\mathcal{R}$) and falsely retained features ($\mathcal{F}$) over 100 replications by EM-test with different steps $K$. Simulation are generated from the negative binomial model (Section 4.1 in the main manuscript).}
	\renewcommand\arraystretch {1.2}
	\footnotesize
	\setlength{\tabcolsep}{3mm}
	\begin{tabular}{cccccccccc}
		\hline
		\hline
		&       & \multicolumn{2}{c}{$p=500$} &       & \multicolumn{2}{c}{$p=5000$} &       & \multicolumn{2}{c}{$p=20,000$} \\
		\cline{3-4} \cline{6-7} \cline{9-10}
		&       & EM-adjust & EM-0.35 &       & EM-adjust & EM-0.35 &       & EM-adjust & EM-0.35 \\
		\hline
		\multicolumn{10}{c}{ Medium signal  and high noise} \\  
		\hline
		\multirow{2}[0]{*}{$K=1$} & $\mathcal{R}$  & 19.6 (0.6) & 20.0 (0.2) &       & 4.2 (3.4) & 12.5 (3.5) &       & 0.0 (0.0) & 2.8 (2.1) \\
		& $\mathcal{F}$  & 0.0 (0.0) & 0.0 (0.0) &       & 0.0 (0.0) & 0.0 (0.0) &       & 0.0 (0.0) & 0.0 (0.1) \\
		\multirow{2}[0]{*}{$K=3$} & $\mathcal{R}$  & 19.7 (0.5) & 20.0 (0.2) &       & 15.0 (2.8) & 19.0 (1.3) &       & 6.8 (2.5) & 15.5 (1.8) \\
		& $\mathcal{F}$  & 0.0 (0.0) & 0.0 (0.0) &       & 0.0 (0.0) & 0.1 (0.2) &       & 0.0 (0.0) & 0.2 (0.5) \\
		\multirow{2}[0]{*}{$K=5$} & $\mathcal{R}$  & 19.7 (0.5) & 20.0 (0.2) &       & 18.1 (1.7) & 19.8 (0.5) &       & 14.7 (2.1) & 19.1 (0.9) \\
		& $\mathcal{F}$  & 0.0 (0.0) & 0.1 (0.2) &       & 0.0 (0.0) & 0.5 (0.6) &       & 0.0 (0.0) & 1.7 (1.4) \\
		\multirow{2}[0]{*}{$K=10$} & $\mathcal{R}$  & 19.7 (0.5) & 20.0 (0.2) &       & 18.9 (1.0) & 20.0 (0.2) &       & 18.0 (1.4) & 19.9 (0.3) \\
		& $\mathcal{F}$  & 0.0 (0.0) & 0.2 (0.4) &       & 0.0 (0.0) & 2.1 (1.4) &       & 0.0 (0.1) & 8.6 (3.1) \\
		\multirow{2}[0]{*}{$K=20$} & $\mathcal{R}$  & 19.7 (0.5) & 20.0 (0.1) &       & 19.1 (1.0) & 20.0 (0.1) &       & 18.5 (1.2) & 19.9 (0.2) \\
		& $\mathcal{F}$  & 0.0 (0.0) & 0.6 (0.7) &       & 0.0 (0.0) & 5.8 (2.5) &       & 0.1 (0.2) & 23.3 (4.4) \\
		\multirow{2}[0]{*}{$K=50$} & $\mathcal{R}$  & 19.7 (0.5) & 20.0 (0.1) &       & 19.1 (0.9) & 20.0 (0.1) &       & 18.6 (1.2) & 19.9 (0.2) \\
		& $\mathcal{F}$  & 0.0 (0.1) & 1.1 (1.0) &       & 0.0 (0.0) & 10.8 (3.5) &       & 0.1 (0.2) & 42.8 (6.2) \\
		\multirow{2}[0]{*}{$K=100$} & $\mathcal{R}$  & 19.7 (0.5) & 20.0 (0.1) &       & 19.2 (0.9) & 20.0 (0.1) &       & 18.6 (1.2) & 19.9 (0.2) \\
		& $\mathcal{F}$  & 0.0 (0.1) & 1.6 (1.3) &       & 0.0 (0.0) & 13.6 (3.6) &       & 0.1 (0.2) & 55.4 (7.1) \\
		\multirow{2}[0]{*}{$K=200$} & $\mathcal{R}$  & 19.8 (0.5) & 20.0 (0.1) &       & 19.2 (0.9) & 20.0 (0.1) &       & 18.7 (1.2) & 20.0 (0.2) \\
		& $\mathcal{F}$  & 0.0 (0.1) & 2.0 (1.4) &       & 0.0 (0.2) & 17.9 (4.1) &       & 0.1 (0.3) & 74.7 (7.5) \\
		\hline
	\end{tabular}%
	\label{tab:step}%
\end{table}%
\newpage
\subsection{Simulations for EM-test with different thresholds $\vartheta$}
\begin{table}
	\centering
	\caption{The mean of numbers of correctly retained features ($\mathcal{R}$) and falsely retained features ($\mathcal{F}$) by the EM-test with different thresholds over 100 replications. Simulation are generated from the negative binomial model (Section 4.1 in the main manuscript). The numbers in the parenthesis are the standard deviation of $\mathcal{R}$ and  $\mathcal{F}$ over 100 replications. EM-0.2 means that we choose the threshold as $n^{0.2}$, similar for EM-0.25 -- EM-0.45. }
	\renewcommand\arraystretch {1.15}
	\footnotesize
	\setlength{\tabcolsep}{3mm}
	\begin{tabular}{cccccccc}
		\hline
		\hline
		&       & \multicolumn{1}{c}{EM-0.2} & \multicolumn{1}{c}{EM-0.25} & \multicolumn{1}{c}{EM-0.3} & \multicolumn{1}{c}{EM-0.35} & \multicolumn{1}{c}{EM-0.4} & \multicolumn{1}{c}{EM-0.45} \\
		\hline 
		\multicolumn{8}{c}{Case 1:  High signal  and low noise} \\  
		\hline  
		\multirow{2}[0]{*}{$p=500$} & $\mathcal{R}$  & 20 (0.1) & 20 (0.1) & 20 (0.1) & 20 (0.1) & 20 (0.1) & 20 (0.1) \\
		& $\mathcal{F}$  & 45 (6.4) & 19 (4.3) & 6 (2.3) & 1 (1.1) & 0 (0.3) & 0 (0.1) \\
		\multirow{2}[0]{*}{$p=5000$} & $\mathcal{R}$  & 20 (0.0) & 20 (0.0) & 20 (0.0) & 20 (0.0) & 20 (0.0) & 20 (0.0) \\
		& $\mathcal{F}$  & 459 (20.7) & 189 (12.2) & 57 (7.4) & 10 (3.3) & 1 (1.0) & 0 (0.1) \\
		\multirow{2}[0]{*}{$p=20,000$} & $\mathcal{R}$  & 20 (0.0) & 20 (0.0) & 20 (0.0) & 20 (0.0) & 20 (0.0) & 20 (0.0) \\
		& $\mathcal{F}$  & 1833 (38.3) & 757 (24.7) & 223 (12.9) & 40 (5.6) & 4 (2.0) & 0 (0.4) \\
		\hline 
		\multicolumn{8}{c}{Case 2:  High signal  and high noise} \\  
		\hline
		\multirow{2}[0]{*}{$p=500$} & $\mathcal{R}$  & 20 (0.0) & 20 (0.0) & 20 (0.0) & 20 (0.0) & 20 (0.0) & 20 (0.2) \\
		& $\mathcal{F}$  & 63 (7.5) & 29 (5.6) & 9 (2.9) & 2 (1.4) & 0 (0.4) & 0 (0.0) \\
		\multirow{2}[0]{*}{$p=5000$} & $\mathcal{R}$  & 20 (0.0) & 20 (0.0) & 20 (0.0) & 20 (0.0) & 20 (0.0) & 20 (0.0) \\
		& $\mathcal{F}$  & 658 (23.4) & 290 (16.0) & 90 (9.4) & 18 (4.1) & 2 (1.3) & 0 (0.1) \\
		\multirow{2}[0]{*}{$p=20,000$} & $\mathcal{R}$  & 20 (0.0) & 20 (0.0) & 20 (0.0) & 20 (0.0) & 20 (0.0) & 20 (0.0) \\
		& $\mathcal{F}$  & 2637 (52.6) & 1177 (36.3) & 375 (21.5) & 75 (7.7) & 8 (3.1) & 0 (0.6) \\
		\hline 
		\multicolumn{8}{c}{Case 3:  Medium signal  and low noise} \\  
		\hline  
		\multirow{2}[0]{*}{$p=500$} & $\mathcal{R}$  & 20 (0.1) & 20 (0.1) & 20 (0.1) & 20 (0.1) & 20 (0.2) & 20 (0.4) \\
		& $\mathcal{F}$  & 45 (6.7) & 19 (4.2) & 6 (2.4) & 1 (1.1) & 0 (0.2) & 0 (0.1) \\
		\multirow{2}[0]{*}{$p=5000$} & $\mathcal{R}$  & 20 (0.0) & 20 (0.0) & 20 (0.1) & 20 (0.1) & 20 (0.1) & 20 (0.4) \\
		& $\mathcal{F}$  & 459 (20.4) & 189 (12.7) & 57 (7.1) & 10 (3.2) & 1 (1.0) & 0 (0.1) \\
		\multirow{2}[0]{*}{$p=20,000$} & $\mathcal{R}$  & 20 (0.0) & 20 (0.0) & 20 (0.0) & 20 (0.0) & 20 (0.2) & 20 (0.4) \\
		& $\mathcal{F}$  & 1834 (41.0) & 757 (25.6) & 223 (13.0) & 41 (5.7) & 4 (2.0) & 0 (0.4) \\
		\hline 
		\multicolumn{8}{c}{Case 4:  Medium signal  and high noise} \\  
		\hline 
		\multirow{2}[0]{*}{$p=500$} & $\mathcal{R}$  & 20 (0.0) & 20 (0.0) & 20 (0.1) & 20 (0.1) & 20 (0.4) & 19 (0.8) \\
		& $\mathcal{F}$  & 63 (7.6) & 28 (5.5) & 9 (2.9) & 2 (1.4) & 0 (0.4) & 0 (0.0) \\
		\multirow{2}[0]{*}{$p=5000$} & $\mathcal{R}$  & 20 (0.0) & 20 (0.0) & 20 (0.0) & 20 (0.1) & 20 (0.4) & 19 (0.9) \\
		& $\mathcal{F}$  & 659 (23.0) & 291 (16.2) & 90 (9.1) & 18 (4.1) & 2 (1.3) & 0 (0.2) \\
		\multirow{2}[0]{*}{$p=20,000$} & $\mathcal{R}$  & 20 (0.0) & 20 (0.1) & 20 (0.1) & 20 (0.2) & 20 (0.4) & 19 (0.9) \\
		& $\mathcal{F}$  & 2637 (51.5) & 1177 (34.9) & 375 (21.3) & 75 (7.5) & 8 (3.0) & 0 (0.6) \\
		\hline 
		\multicolumn{8}{c}{Case 5:  Low signal  and low noise} \\  
		\hline  
		\multirow{2}[0]{*}{$p=500$} & $\mathcal{R}$  & 20 (0.3) & 20 (0.5) & 19 (0.7) & 19 (1.0) & 18 (1.5) & 15 (2.0) \\
		& $\mathcal{F}$  & 45 (6.9) & 19 (4.1) & 6 (2.3) & 1 (1.1) & 0 (0.2) & 0 (0.1) \\
		\multirow{2}[0]{*}{$p=5000$} & $\mathcal{R}$  & 20 (0.3) & 20 (0.3) & 20 (0.5) & 19 (1.0) & 17 (1.6) & 14 (1.9) \\
		& $\mathcal{F}$  & 459 (20.7) & 189 (11.8) & 57 (7.4) & 10 (3.3) & 1 (1.0) & 0 (0.1) \\
		\multirow{2}[0]{*}{$p=20,000$} & $\mathcal{R}$  & 20 (0.3) & 20 (0.4) & 20 (0.6) & 19 (1.1) & 17 (1.4) & 14 (1.9) \\
		& $\mathcal{F}$  & 1833 (38.3) & 758 (25.2) & 223 (13.3) & 41 (5.6) & 4 (1.9) & 0 (0.4) \\
		\hline 
		\multicolumn{8}{c}{Case 6:  Low signal  and high noise} \\  
		\hline  
		\multirow{2}[0]{*}{$p=500$} & $\mathcal{R}$  & 20 (0.4) & 20 (0.6) & 19 (0.8) & 18 (1.2) & 16 (1.7) & 13 (2.0) \\
		& $\mathcal{F}$  & 65 (7.8) & 29 (5.4) & 10 (3.2) & 2 (1.4) & 0 (0.4) & 0 (0.0) \\
		\multirow{2}[0]{*}{$p=5000$} & $\mathcal{R}$  & 20 (0.4) & 20 (0.6) & 19 (0.8) & 18 (1.3) & 16 (1.7) & 13 (2.1) \\
		& $\mathcal{F}$  & 660 (24.4) & 291 (16.3) & 91 (9.2) & 18 (4.2) & 2 (1.2) & 0 (0.1) \\
		\multirow{2}[0]{*}{$p=20,000$} & $\mathcal{R}$  & 20 (0.5) & 20 (0.6) & 19 (0.9) & 18 (1.4) & 16 (1.7) & 12 (2.0) \\
		& $\mathcal{F}$  & 2636 (52.9) & 1176 (36.7) & 374 (21.7) & 75 (7.4) & 8 (3.0) & 0 (0.6) \\
		\hline
	\end{tabular}%
	\label{tab:threshold}%
\end{table}%

\newpage
\subsection{Simulations for continuous data}
{\colblue In this section, we perform simulations for continuous data.  The distribution family is chosen as the normal distribution.  We consider three dimension setups $p = 500, 5000$ and $20, 000$. The sample size is set as $n = 1000$ and the number of cluster-relevant features is $s = 20$. 
	
	We first consider simulation setup of  balanced scenario. We set the number of clusters as $G=5$ and the proportions of the clusters as $\ba =(\alpha_1,\dots, \alpha_5) = (0.2,0.2,0.2,0.2,0.2)$.   For the $i$th sample, we first randomly assign it to a cluster $g$ with the probability $\alpha_g$. Then, if the $j$th feature is cluster-relevant ($j=1,\dots,20$), we randomly sample $x_{ij}$ from ${\rm Normal} (\mu_{gj}, \sigma^2_{j})$; If it is cluster-irrelevant ($j=21,\dots,p$), we randomly sample $x_{ij}$ from ${\rm Normal} (\mu_j, \sigma^2_{j})$.  
	We independently generate $\sigma_j$ from the uniform distributions ${\rm U}(1,1.5)$.  For the clustering-relevant features ($j=1,\dots,20$), the mean parameters $\mu_{gj}$ are either set as $u_j$ or $u_j+D_j$, where $u_j$ is generated from $ {\rm U} (-5, 5)$, and $D_j$ is to control the signal strength (the differences between clusters). We generate $D_j$ from ${\rm U}(10,11)$. For the first 5 features ($1\leq j\leq 5$), we set $\mu_{2j}=u_j+D_j$ and $\mu_{gj} = u_j(g \neq 2)$. 
	%Similarly, for $11\leq j\leq 20$, we set $\mu_{1j}=u_j+D_j$ and $\mu_{2j} = u_j$. For all cluster-irrelevant features ($j=21,\dots,p$), we set $\mu_j=u_j$, where $u_j$ is generated from $ {\rm U} (-5,  5)$.  For the first 5 features ($1\leq j\leq 5$), we set $\mu_{2j}=u_j+D_j$ and $\mu_{gj} = u_j (g\neq 2)$. 
	Similarly, for $5k+1\leq j\leq 5k+5 (k=1,2,3)$, we set $\mu_{k+2,j}=u_j+D_j$ and $\mu_{gj} = u_j (g\neq k+2)$. For all cluster-irrelevant features ($j=21,\dots,p$), we set $\mu_j=u_j$, where $u_j$ is generated from $ {\rm U} (-5,  5)$.

	Then we consider simulation setup of unbalanced scenario. We set the number of clusters as $G=5$ and the proportions of the clusters as $\ba =(\alpha_1,\dots,\alpha_5) = (0.5,0.125,0.125,0.125,0.125)$.  
	We independently generate $\sigma_j$ from the uniform distributions ${\rm U}(1,2)$. For the clustering-relevant features ($j=1,\dots,20$), the mean parameters $\mu_{gj}$ are either set as $u_j$ or $u_j+D_j$, where $u_j$ is generated from $ {\rm U} (-5, 5)$, and $D_j$ is to control the signal strength (the differences between clusters). We generate $D_j$ from ${\rm U}(3,4)$. For the first 5 features ($1\leq j\leq 5$), we set $\mu_{2j}=u_j+D_j$ and $\mu_{gj} = u_j(g\neq 2)$. 
	%Similarly, for $11\leq j\leq 20$, we set $\mu_{1j}=u_j+D_j$ and $\mu_{2j} = u_j$. For all cluster-irrelevant features ($j=21,\dots,p$), we set $\mu_j=u_j$, where $u_j$ is generated from $ {\rm U} (-5,  5)$.  For the first 5 features ($1\leq j\leq 5$), we set $\mu_{2j}=u_j+D_j$ and $\mu_{gj} = u_j (g\neq 2)$. 
	Similarly, for $5k+1\leq j\leq 5k+5 (k=1,2,3)$, we set $\mu_{k+2,j}=u_j+D_j$ and $\mu_{gj} = u_j (g\neq k+2)$. For all cluster-irrelevant features ($j=21,\dots,p$), we set $\mu_j=u_j$, where $u_j$ is generated from $ {\rm U} (-5,  5)$. }

\begin{table}
	\centering
	\caption{The mean of numbers of correctly retained features ($\mathcal{R}$) and falsely retained features ($\mathcal{F}$) by different methods in the normal simulations. EM-adjust means that select the features by the adjusted p-values and EM-0.35 means that we choose the threshold as $n^{0.35}$. KS-test is the test used by IF-PCA \citep{jin2016influential}, Dip-test is the uni-modality test \citep{chan2010using} and COSCI is the test mentioned  by the reviewer \citep{banerjee2017feature}. COSCI is not evaluated for $p = 20, 000$ because of its high computational cost.}
	\setlength{\tabcolsep}{4mm}
	\renewcommand\arraystretch {1.2}
	\footnotesize
	\begin{tabular}{cccccccc}
		\hline
		\hline
		&       & EM-adjust & EM-0.35 & SC-FS  & KS-test & Dip-test & COSCI \\
		\hline
		\multicolumn{8}{c}{Case 1:  Balanced } \\  
		\hline
		\multirow{2}[0]{*}{$p=500$} & $\mathcal{R}$  & 20 (0.0) & 20 (0.0) & 20 (0.0) & 20 (0.0) & 20 (0.0) & 16 (2.1) \\
		& $\mathcal{F}$  & 0 (0.5) & 4 (2.2) & 0 (0.0) & 0 (0.0) & 0 (0.0) & 45 (4.6) \\
		\multirow{2}[0]{*}{$p=5000$} & $\mathcal{R}$  & 20 (0.0) & 20 (0.0) & 20 (0.0) & 20 (0.0) & 20 (0.0) & 17 (1.6) \\
		& $\mathcal{F}$  & 0 (0.6) & 39 (6.2) & 0 (0.0) & 0 (0.0) & 0 (0.0) & 532 (16.4) \\
		\multirow{2}[0]{*}{$p=20,000$} & $\mathcal{R}$  & 20 (0.0) & 20 (0.0) & 16 (4.3) & 20 (0.0) & 20 (0.0) & NA \\
		& $\mathcal{F}$  & 0 (0.6) & 153 (12.5) & 0 (0.0) & 0 (0.0) & 0 (0.0) & NA \\
		\hline
		\multicolumn{8}{c}{Case 2: Unbalanced} \\  
		\hline
		\multirow{2}[0]{*}{$p=500$} & $\mathcal{R}$  & 18 (1.4) & 19 (0.8) & 20 (0.7) & 6 (2.3) & 0 (0.0) & 1 (0.9) \\
		& $\mathcal{F}$  & 0 (0.5) & 4 (2.0) & 0 (0.0) & 0 (0.0) & 0 (0.0) & 53 (5.3) \\
		\multirow{2}[0]{*}{$p=5000$} & $\mathcal{R}$  & 16 (1.7) & 19 (0.9) & 0 (0.3) & 4 (2.0) & 0 (0.0) & 1 (1.1) \\
		& $\mathcal{F}$  & 0 (0.5) & 39 (5.9) & 0 (0.0) & 0 (0.0) & 0 (0.0) & 538 (17.2) \\
		\multirow{2}[0]{*}{$p=20,000$} & $\mathcal{R}$  & 15 (2.0) & 19 (0.9) & 0 (0.0) & 3 (1.7) & 0 (0.0) & NA \\
		& $\mathcal{F}$  & 0 (0.6) & 153 (11.9) & 0 (0.0) & 0 (0.0) & 0 (0.0) & NA \\
		\hline
	\end{tabular}%
	\label{tab:PNnormalCase}%
\end{table}%

\begin{table}[htbp]
	\centering
	\caption{Similar to Table \ref{tab:PNnormalCase} but for clustering accuracy. The means and standard deviations (in parenthesis) of ARIs over 100 replications by different methods in the normal simulations. The values in the table are shown as the actual values $\times$ 100. No-Screening means that we use all features for clustering. Oracle means that we only use the $s=20$ clustering-relevant features for clustering.}
	\renewcommand\arraystretch {1.2}
	\footnotesize
	\begin{tabular}{ccccccccc}
		\hline
		\hline
		& No-Screening & Oracle & EM-adjust & EM-0.35 & SC-FS  & KS-test & Dip-test & COSCI \\
		\hline
		\multicolumn{9}{c}{Case 1:  Balanced} \\  
		\hline
		$p=500$ & 71 (10.6) & 96 (9.5) & 97 (8.9) & 94 (11.4) & 97 (8.6) & 96 (9.5) & 96 (9.7) & 87 (11.6) \\
		$p=5000$ & 62 (2.8) & 94 (11.1) & 94 (11.1) & 90 (9.3) & 94 (11.4) & 94 (11.1) & 96 (9.7) & 75 (4.1) \\
		$p=20,000$ & 16 (4.9) & 93 (11.7) & 93 (12.3) & 80 (5.1) & 66 (18.1) & 93 (11.7) & 96 (10.0) & NA \\
		\hline
		\multicolumn{9}{c}{Case 2:  Unbalanced} \\  
		\hline
		$p=500$ & 61 (9.9) & 97 (1.5) & 95 (3.4) & 96 (2.1) & 96 (5.3) & 52 (21.5) & 0 (0.0) & 1 (1.8) \\
		$p=5000$ & 1 (0.6) & 96 (3.0) & 93 (5.1) & 94 (3.7) & 0 (1.6) & 36 (20.4) & 0 (0.0) & 0 (0.2) \\
		$p=20,000$ & 0 (0.2) & 97 (1.5) & 92 (7.1) & 86 (8.2) & 0 (0.0) & 32 (19.8) & 0 (0.0) & NA \\
		\hline
	\end{tabular}%
	\label{tab:ARInormalCase}%
\end{table}%

{\colblue Results of these two scenarios are presented in Table \ref{tab:PNnormalCase} and \ref{tab:ARInormalCase}.  We first compare different algorithms in terms of the number of correctly retained features and falsely retained features (Table \ref{tab:PNnormalCase}). The  balanced case is a relative simple scenario and many methods work well. Overall, the two versions of EM-test could correctly select most cluster-relevant features and have very few false positives. The  unbalanced case is much more challenging and EM-test outperforms other methods by a large margin, especially when $p$ is larger. Dip-test \citep{chan2010using} tends to be very conservative in this case and does not select any features. KS-test (the test used by IF-PCA) is able to select a few important features with very few false positives. COSCI has a large number false positives. SC-FS performs well in the low dimensional case ($p=500$), but cannot select any feature in the higher dimensional cases. We also compare the clustering accuracy based on the selected features (Table \ref{tab:ARInormalCase}). Again, we see that EM-test outperforms other methods, especially for the more challenging unbalanced case.}

\subsection{EM-test under mis-specified model}
{\colblue We consider two mis-specified models to investigate the robustness of the proposed method.  
	%	Considering that a Poisson-Gamma mixture is a negative binomial distribution,  we first replace the Gamma distribution with the truncated normal distribution ${\rm TN}(\mu, \sigma, \gamma)$ with the   probability density function  
	%	$$
	%	{\displaystyle f(x;\mu ,\sigma ,a,b)={\frac {1}{\sigma }}\,{\frac {\phi ({\frac {x-\mu }{\sigma }})}{1-\Phi ({\frac {\gamma-\mu }{\sigma }})}}}, x\geq \gamma,
	%	$$
	%	where $\phi(\cdot)$ s the probability density function of the standard normal distribution and 
	%	$\Phi (\cdot )$ is its cumulative distribution function.  
	The two mis-specified models are the Poisson-truncated-normal and the binomial-Gamma distributions. A random variable $y$ is said to follow a Poisson-truncated-normal distribution ${\rm PTN}(\mu, \sigma, \gamma)$ ($\mu\in \mathcal{R}, \sigma>0, \gamma>0$), if conditional on a latent variable $\lambda$, $x$ follows a Poisson distribution with a mean $\lambda$, and the latent variable $\lambda$ follows the truncated normal distribution with the probability density function
	$$
	{\displaystyle f(x;\mu ,\sigma ,a,b)={\frac {1}{\sigma }}\,{\frac {\phi ({\frac {x-\mu }{\sigma }})}{1-\Phi ({\frac {\gamma-\mu }{\sigma }})}}}, x\geq \gamma,
	$$
	where $\phi(\cdot)$ is the probability density function of the standard normal distribution and 
	$\Phi (\cdot )$ is its cumulative distribution function. In this simulation, we set $\gamma = 0.5$, $\sigma = 1$. 
	%	The simulation data generation is the same as the negative binomial  case except for setting $\gamma = 0.5$, $\sigma = 1$ and the truncated normal  mean differences  between clusters $D \sim {\rm U}(7,8)$.   
	
	The simulation data generation is the same as the negative binomial  case. We set the number of clusters as $G=5$ and the proportions of the clusters as $$\ba =(\alpha_1,\dots,\alpha_5) = (0.5,0.125,0.125,0.125,0.125).$$  We consider three dimension setups $p = 500, 5000$ and $20, 000$. The sample size is set as $n = 1000$ and the number of cluster-relevant features is $s = 20$.   For the clustering-relevant features ($j=1,\dots,20$), the mean parameters of the truncated normal distribution $\mu_{gj}$ are either set as $\exp(u_j)$ or $\exp(u_j)+D_j$, where $u_j$ is generated from $ {\rm U} (\log \ 2, \log \ 5)$, and $D_j$ is to control the signal strength (the differences between clusters). We generate $D_j$ from ${\rm U}(7,8)$. For the first 5 features ($1\leq j\leq 5$), we set $\mu_{2j}=\exp(u_j)+D_j$ and $\mu_{gj} = \exp(u_j) (g \neq 2)$. 
	%Similarly, for $11\leq j\leq 20$, we set $\mu_{1j}=u_j+D_j$ and $\mu_{2j} = u_j$. For all cluster-irrelevant features ($j=21,\dots,p$), we set $\mu_j=u_j$, where $u_j$ is generated from $ {\rm U} (-5,  5)$.  For the first 5 features ($1\leq j\leq 5$), we set $\mu_{2j}=u_j+D_j$ and $\mu_{gj} = u_j (g\neq 2)$. 
	Similarly, for $5k+1\leq j\leq 5k+5 (k=1,2,3)$, we set $\mu_{k+2,j}=\exp(u_j)+D_j$ and $\mu_{gj} = \exp(u_j) (g\neq k+2)$. For all cluster-irrelevant features ($j=21,\dots,p$), we set $\mu_j=\exp(u_j)$, where $u_j$ is generated from $ {\rm U} (\log \ 2, \log \ 5)$.

	We next consider another  mis-specified model the binomial-Gamma distribution as follows.  A random variable $x$ is said to follow a binomial-Gamma ${\rm BG}(z, \mu, r)$ if 
	$$
	x|\lambda \sim {\rm Binomial}(\lceil \max(z,\lambda) \rceil, \lambda / \lceil \max(z,\lambda) \rceil),~\lambda \sim {\rm Gamma}(r, \mu/r)
	$$
	where $r$ is the shape parameter and $\mu/r$ is the scale parameter of the Gamma distribution.   In this simulation, we set $z = 100$.   		
	%	Similarly, the simulation data generation is the same as the negative binomial  case except for generating $r_j$ from the uniform distributions  ${\rm U}(5,6)$ and the mean ($\mu$) differences between clusters $D_j$  from ${\rm U}(2,3)$.
	We set the number of clusters as $G=5$ and the proportions of the clusters as $\ba =(\alpha_1,\dots,\alpha_5) = (0.5,0.125,0.125,0.125,0.125)$.  We consider three dimension setups $p = 500, 5000$ and $20, 000$. The sample size is set as $n = 1000$ and the number of cluster-relevant features is $s = 20$.   We independently generate $r_j$ from the uniform distributions ${\rm U}(5, 6)$.  For the clustering-relevant features ($j=1,\dots,20$), the mean parameters $\mu_{gj}$ of binomial-Gamma  (${\rm BG}(z, \mu_{gj}, r_j)$) are either set as $\exp(u_j)$ or $\exp(u_j)+D_j$, where $u_j$ is generated from $ {\rm U} (\log \ 2, \log \ 5)$, and $D_j$ is to control the signal strength (the differences between clusters). We generate $D_j$ from ${\rm U}(2,3)$. For the first 5 features ($1\leq j\leq 5$), we set $\mu_{2j}=\exp(u_j)+D_j$ and $\mu_{gj} = \exp(u_j) (g\neq 2)$. 
	%Similarly, for $11\leq j\leq 20$, we set $\mu_{1j}=u_j+D_j$ and $\mu_{2j} = u_j$. For all cluster-irrelevant features ($j=21,\dots,p$), we set $\mu_j=u_j$, where $u_j$ is generated from $ {\rm U} (-5,  5)$.  For the first 5 features ($1\leq j\leq 5$), we set $\mu_{2j}=u_j+D_j$ and $\mu_{gj} = u_j (g\neq 2)$. 
	Similarly, for $5k+1\leq j\leq 5k+5 (k=1,2,3)$, we set $\mu_{k+2,j}=\exp(u_j)+D_j$ and $\mu_{gj} = \exp(u_j) (g\neq k+2)$. For all cluster-irrelevant features ($j=21,\dots,p$), we set $\mu_j=\exp(u_j)$, where $u_j$ is generated from $ {\rm U} (\log \ 2, \log \ 5)$. 
	
	For comparison, we also included the continuous methods IF-PCA (KS-test), Dip-test and COCSI. Before applying the continuous methods, we apply a log transformation ($\log (x+1)$) to the count data to make the data more like continuous data. Because of the computational burden, COCSI is not considered for the $p=20,000$ simulations. Table \ref{tab:ARImisnor}  and \ref{tab:PNmisnor} show the simulation results for the Poisson-truncated normal model.	EM-test still performs the best under this mis-specified model. For example, the ARIs of the clustering results based on features selected by EM-test are consistently larger than those of other methods, especially in the higher dimensional setups (Table \ref{tab:ARImisnor}). EM-test could select almost all clustering-relevant features with few false positives. Other methods either select too few clustering-relevant features or report too much false positives. The continuous methods do not perform well for these count data. KS-test and Dip-test report many false positives and select almost all features as clustering-relevant features (adjusted p-values $<0.01$ ). COCSI instead is very conservative  in this simulation and could not select any features. The simulation results for the binomial-Gamma mis-specified model are similar and are shown in Table S10-S11.  }
\begin{table}[htbp]
	\centering
	
	\caption{The means and standard deviations (in parenthesis) of ARIs over 100 replications by different methods under the Poisson-truncated normal (mis-specified) model. The values in the table are shown as the actual values $\times$ 100. EM-adjust means that  the features  are selected  by the adjusted p-values and EM-0.35 means that we choose the threshold as $n^{0.35}$. The three continuous methods Dip-test, KS-test and COSCI are applied to the normalized data (log normalization).}
	\renewcommand\arraystretch {1.2}
	\footnotesize
	\begin{tabular}{ccccccccc}
		\hline
		\hline
		& No-Screening & EM-adjust & EM-0.35 & Chi-square & SC-FS & Dip-test & KS-test & COSCI \\
		\hline
		$p=500$ & 93 (6.1) & 99 (2.7) & 99 (2.7) & 82 (15.4) & 98 (4.0) & 95 (1.4) & 95 (1.5) & 0 (0.0) \\
		$p=5000$ & 7 (2.7) & 99 (1.2) & 99 (0.5) & 50 (24.6) & 37 (28.0) & 9 (2.6) & 9 (2.6) & 0 (0.0) \\
		$p=20,000$ & 0 (0.3) & 99 (0.9) & 98 (0.7) & 24 (21.4) & 0 (0.0) & 1 (0.3) & 1 (0.3) & NA \\
		
		\hline
	\end{tabular}%
	
	\label{tab:ARImisnor}%
\end{table}%
% Table generated by Excel2LaTeX from sheet 'addition'
\begin{table}[htbp]
	\centering
	\caption{The mean of numbers of correctly retained features ($\mathcal{R}$) and falsely retained features ($\mathcal{F}$) by different methods under the Poisson-truncated normal (mis-specified) model. }
	
	\renewcommand\arraystretch {1.2}
	\footnotesize
	\begin{tabular}{ccccccccc}
		\hline
		\hline
		&       & EM-adjust & EM-0.35 & Chi-square & SC-FS & Dip-test & KS-test & COSCI \\
		\hline
		\multirow{2}[0]{*}{$p=500$} & $\mathcal{R}$  & 20 (0.4) & 20 (0.1) & 10 (2.7) & 20 (0.4) & 20 (0.0) & 20 (0.0) & 0 (0.0) \\
		& $\mathcal{F}$  & 0 (0.3) & 3 (1.7) & 0 (0.4) & 0 (0.0) & 480 (0.0) & 480 (0.0) & 0 (0.0) \\
		\multirow{2}[0]{*}{$p=5000$} & $\mathcal{R}$  & 19 (0.7) & 20 (0.0) & 5 (2.5) & 14 (4.7) & 20 (0.0) & 20 (0.0) & 0 (0.0) \\
		& $\mathcal{F}$  & 0 (0.3) & 31 (6.0) & 0 (0.5) & 0 (0.0) & 4980 (0.0) & 4980 (0.0) & 0 (0.0) \\
		\multirow{2}[0]{*}{$p=20,000$} & $\mathcal{R}$  & 19 (1.1) & 20 (0.1) & 2 (1.9) & 0 (0.2) & 20 (0.0) & 20 (0.0) & NA \\
		& $\mathcal{F}$  & 0 (0.4) & 124 (10.3) & 0 (0.3) & 0 (0.0) & 19980 (0.0) & 19980 (0.0) & NA  \\
		
		\hline
	\end{tabular}%
	
	\label{tab:PNmisnor}%
\end{table}%

\begin{table}[htbp]
	\centering
	\caption{The means and standard deviations (in parenthesis) of ARIs over 100 replications by different methods under the binomial-Gamma (mis-specified) model. The values in the table are shown as the actual values $\times$ 100.  EM-adjust means that  the features  are selected by the adjusted p-values and EM-0.35 means that we choose the threshold as $n^{0.35}$, similar for EM-0.35. The three continuous methods Dip-test, KS-test and COSCI are applied to the normalized data (log normalization).}
	\renewcommand\arraystretch {1.2}
	\footnotesize
	\begin{tabular}{ccccccccc}
		\hline
		\hline
		& No-Screening & EM-adjust & EM-0.35 & Chi-square & SC-FS & Dip-test & KS-test & COSCI \\
		\hline
		$p=500$ & 73 (33.1) & 83 (32.9) & 86 (29.9) & 75 (39.7) & 82 (30.9) & 79 (33.7) & 78 (33.9) & 2 (4.0) \\
		$p=5000$ & 14 (12.0) & 82 (35.2) & 84 (31.2) & 70 (39.2) & 56 (47.6) & 16 (12.6) & 16 (13.3) & 0 (1.1) \\
		$p=20,000$ & 1 (0.6) & 81 (35.4) & 81 (33.9) & 45 (42.3) & 0 (2.1) & 1 (0.7) & 1 (0.8) & NA \\
		\hline
	\end{tabular}%
	
	\label{tab:ARImisbio}%
\end{table}%

\begin{table}[htbp]
	\centering
	\caption{The mean of numbers of correctly retained features ($\mathcal{R}$) and falsely retained features ($\mathcal{F}$) by different methods  under the binomial-Gamma (mis-specified) model.}
	\setlength{\tabcolsep}{1.6mm}
	\renewcommand\arraystretch {1.2}
	\footnotesize
	\begin{tabular}{ccccccccc}
		\hline
		\hline
		&       & EM-adjust & EM-0.35 & Chi-square & SC-FS & Dip-test & KS-test & COSCI \\
		\hline
		\multirow{2}[0]{*}{$p=500$} & $\mathcal{R}$  & 16 (6.8) & 17 (5.7) & 14 (8.0) & 18 (5.0) & 20 (0.3) & 19 (1.5) & 1 (1.1) \\
		& $\mathcal{F}$  & 0 (0.3) & 3 (2.4) & 68 (95.4) & 0 (0.3) & 480 (0.2) & 453 (35.2) & 40 (22.5) \\
		\multirow{2}[0]{*}{$p=5000$} & $\mathcal{R}$  & 16 (7.2) & 17 (5.6) & 14 (8.3) & 13 (9.0) & 20 (0.3) & 19 (1.5) & 1 (1.1) \\
		& $\mathcal{F}$  & 0 (0.4) & 25 (17.5) & 676 (975.2) & 0 (0.0) & 4980 (0.9) & 4697 (362.0) & 425 (227.6) \\
		\multirow{2}[0]{*}{$p=20,000$} & $\mathcal{R}$  & 16 (7.2) & 17 (5.7) & 13 (8.4) & 0 (1.4) & 20 (0.3) & 19 (1.2) & NA \\
		& $\mathcal{F}$  & 0 (0.4) & 100 (68.9) & 2703 (3934.1) & 0 (0.0) & 19979 (4.4) & 18835 (1459.0) & NA \\
		
		\hline
	\end{tabular}%
	
	\label{tab:PNmisbio}%
\end{table}%

\newpage
\subsection{Simulations for the limiting distribution}
\begin{table}[htbp]
	\centering
	\caption{The means and standard deviations (in parenthesis) of FDR and power over 100 replications. EM(1) means that the p-values are obtained from the $\chi^2(3)$ distribution, and EM(2) means that the p-values are calculated using the limiting distribution in Theorem \thmtwo.  The Benjamini–Hochberg procedure controls the FDR at 0.01. Simulation are generated from the negative binomial model (Section 4.1 in the main manuscript).}
	\renewcommand\arraystretch {1.2}
	\footnotesize
	\setlength{\tabcolsep}{4mm}
	\begin{tabular}{ccccccc}
		\hline
		\hline
		& \multicolumn{2}{c}{EM (1)} & \multicolumn{2}{c}{EM (2)} & \multicolumn{2}{c}{Chi-square} \\
		\cline{2-3} \cline{4-5} 	 \cline{6-7} 
		& FDR   & Power & FDR   & Power & FDR   & Power \\
		\hline 
		\multicolumn{7}{c}{Case 1:  High signal  and low noise} \\  
		\hline 
		$p=500$ & 0.00 (0.01) & 1.00 (0.01) &0.01 (0.02) &	1.00 (0.00) & 0.02 (0.03) & 0.98 (0.03) \\
		$p=5000$ & 0.00 (0.00) & 1.00 (0.00) & 0.02 (0.03) & 1.00 (0.00) & 0.02 (0.03) & 0.93 (0.06) \\
		$p=20,000$ & 0.00 (0.01) & 1.00 (0.00) & NA  & NA  & 0.01 (0.03) & 0.90 (0.06) \\
		\hline 
		\multicolumn{7}{c}{Case 2:  High signal  and high noise} \\  
		\hline 
		$p=500$ & 0.00 (0.01) & 1.00 (0.01) & 0.02 (0.03) &	1.00 (0.00) & 0.01 (0.03) & 0.92 (0.06) \\
		$p=5000$ & 0.00 (0.01) & 1.00 (0.00) & 0.04 (0.04) & 1.00 (0.00) & 0.02 (0.03) & 0.83 (0.09) \\
		$p=20,000$ & 0.00 (0.01) & 1.00 (0.01) & NA  & NA  & 0.02 (0.04) & 0.75 (0.10) \\
		\hline 
		\multicolumn{7}{c}{Case 3:  Medium signal  and low noise} \\  
		\hline 
		$p=500$ & 0.00 (0.01) & 1.00 (0.01) & 0.01 (0.02) &	1.00 (0.01) & 0.01 (0.03) & 0.78 (0.12) \\
		$p=5000$ & 0.00 (0.00) & 0.99 (0.02) & 0.02 (0.03) & 1.00 (0.02) & 0.02 (0.03) & 0.61 (0.11) \\
		$p=20,000$ & 0.00 (0.01) & 0.98 (0.03) & NA  & NA  & 0.01 (0.04) & 0.50 (0.13) \\
		\hline 
		\multicolumn{7}{c}{Case 4:  Medium signal  and high noise} \\  
		\hline 
		$p=500$ & 0.00 (0.01) & 0.99 (0.03) & 0.02 (0.03)	& 1.00 (0.02) & 0.02 (0.04) & 0.65 (0.11) \\
		$p=5000$ & 0.00 (0.01) & 0.96 (0.04) & 0.04 (0.05) & 0.98 (0.03) & 0.01 (0.03) & 0.46 (0.13) \\
		$p=20,000$ & 0.00 (0.01) & 0.93 (0.06) & NA  & NA  & 0.02 (0.05) & 0.34 (0.13) \\
		\hline 
		\multicolumn{7}{c}{Case 5:  Low signal  and low noise} \\  
		\hline 
		$p=500$ & 0.00 (0.01) & 0.82 (0.10) &0.01 (0.02) &	0.90 (0.08) & 0.02 (0.07) & 0.20 (0.12) \\
		$p=5000$ & 0.00 (0.01) & 0.69 (0.10) & 0.03 (0.04) & 0.80 (0.09) & 0.03 (0.12) & 0.08 (0.08) \\
		$p=20,000$ & 0.00 (0.01) & 0.61 (0.12) & NA  & NA  & 0.02 (0.11) & 0.04 (0.05) \\
		\hline 
		\multicolumn{7}{c}{Case 6:  Low signal  and high noise} \\  
		\hline 
		$p=500$ & 0.00 (0.01) & 0.73 (0.10) & 0.02 (0.04) &	0.83 (0.09) & 0.02 (0.07) & 0.14 (0.10) \\
		$p=5000$ & 0.00 (0.01) & 0.60 (0.11) & 0.05 (0.05) & 0.74 (0.09) & 0.02 (0.11) & 0.05 (0.05) \\
		$p=20,000$ & 0.01 (0.02) & 0.50 (0.12) & NA  & NA  & 0.02 (0.12) & 0.03 (0.05) \\
		\hline
	\end{tabular}%
	\label{tab:FDRpower}%
\end{table}%
\newpage
\subsection{The computation time of different methods.}
\begin{table}[htbp]
	\centering
	\caption{The computation time of different methods. Simulation are generated from the negative binomial model (Section 4.1 in the main manuscript).}
	\renewcommand\arraystretch {1.2}
	\footnotesize
	\setlength{\tabcolsep}{4mm}
	\begin{tabular}{cccccccc}
		\hline
		\hline
		Time (s) & EM-test & Chi-square & SC-FS & Skmeans & KS-test & Dip-test & COSCI \\
		\hline
		$p=500$   & 14.59  & 25.57  & 0.98  & 265.36  & 1.83  & 1.86  & 53.73  \\
		$p=5000$  & 127.21  & 246.97  & 7.60  & 3074.06  & 7.01  & 7.09  & 470.50  \\
		$p=20,000$ & 505.21  & 986.82  & 30.53  & NA    & 25.90  & 25.84  & NA \\
		\hline
	\end{tabular}%
	\label{tab:time}%
\end{table}%

\section{Details for the application on scRNA-seq data}
In the analysis of the scRNA-seq data from \citet{Heming2021NeurologicalMO}, we mainly follow the analysis protocol of Seurat \citep{Butler2018IntegratingST}. Since there are 31 patients in this dataset, we must consider the batch effect \citep{Haghverdi2018BatchEI} which may have a non-negligible effect on the count matrix from different patients.  Also, it is known that systematic differences in library size between different cells are often observed in scRNA-seq data \citep{Stegle2015ComputationalAA}. Therefore, before applying our screening procedure, we remove this confounding effect via down-sampling  such that each cell has the same number of unique molecular identifier (UMI) counts.

Under the assumption that clustering-informative genes must be heterogeneously distributed in at least one batch, we apply the EM-test to each batch $b \,(b=1,\dots,B)$, respectively, and get a p-value $p_j^{(b)}$ for  each gene $j$. Then, we calculate the Bonferroni-type combined p-values \citep{vovk2020combining}:
$$p_j^\text{comb} = B \cdot  \min \left \{p_j^{(1)}, p_j^{(2)},\dots,p_j^{(B)} \right\},$$
and then perform the Bonferroni-Hochberg (BH) procedure \citep{Benjamini1995ControllingTF} of false discovery rate control on the $p_j^\text{comb}$'s. Finally, we select genes with an adjusted p-value smaller than 0.01 for downstream analysis.

Before applying dimensional reduction methods, we first normalize and scale the count matrix by \texttt{NormalizeData()} and \texttt{ScaleData()} in Seurat. Then, we perform PCA analysis and select the first $40$ principle components as the input for \texttt{harmony} to do batch effect removal \citep{Korsunsky2018FastSA}. After that, following the standard analysis protocol of Seurat, we construct a SNN graph with the derived ``harmony dimensions", perform clustering with the Louvain's method, and further reduce the dimensions of the data to two via UMAP for virtualization. To annotate each derived clusters, we check the expression of each marker genes provided by \cite{Heming2021NeurologicalMO}.

As to the implementation of the Chi-square test, we group the data into four bins according to the median and the upper and lower quartiles  and perform similar analyses as the EM-test.

%%%%%%%%%%%%%%%%%%%%%%%

\bigskip

\bibliographystyle{Chicago}

\bibliography{FeatureScreening}
\end{document}